\documentclass[11pt,fleqn,a4paper]{article} %,draft

\usepackage{lscape,longtable}
\usepackage{amsmath,amssymb,amsthm,amsfonts,cite} %,refcheck ,enumerate ,backref
\usepackage[mathscr]{eucal}
\usepackage{xcolor}

\usepackage{hyperref}
\hypersetup{colorlinks, linkcolor=blue, citecolor=blue, urlcolor=blue}

\allowdisplaybreaks

\makeatletter
\long\def\@makecaption#1#2{%
	\vskip\abovecaptionskip\footnotesize
	\sbox\@tempboxa{#1. #2}%
	\ifdim \wd\@tempboxa >\hsize
	#1. #2\par
	\else
	\global \@minipagefalse
	\hb@xt@\hsize{\hfil\box\@tempboxa\hfil}%
	\fi
	\vskip\belowcaptionskip}
\makeatother

\newtheorem{theorem}{Theorem}%[section]
\newtheorem{lemma}[theorem]{Lemma}
\newtheorem{corollary}[theorem]{Corollary}

\newtheorem{proposition}[theorem]{Proposition}
{\theoremstyle{definition}
	
	\newtheorem{remark}[theorem]{Remark}
	
}

\newcommand{\todo}[1][\null]{\ensuremath{\clubsuit}}

\newcommand{\noprint}[1]{}
\newcommand{\lsemioplus}{\mathbin{\mbox{$\lefteqn{\hspace{.77ex}\rule{.4pt}{1.2ex}}{\in}$}}}

\def\g{\mathfrak{g}}
\def\h{\mathfrak{h}}

\DeclareMathOperator{\Inn}{Inn}
\DeclareMathOperator{\Nor}{Nor}
\DeclareMathOperator{\Ad}{Ad}
\DeclareMathOperator{\ad}{ad}

\begin{document}

\par\noindent {\LARGE\bf
Subalgebras of Lie algebras. Example of $\mathfrak{sl}_3(\mathbb R)$ revisited
\par}

\vspace{4.5mm}\par\noindent{\large
\large Yevhenii Yu. Chapovskyi$^\dag$,
Serhii D. Koval$^{\dag\ddag}$
and Olha Zhur$^\S$
}

\vspace{5.5mm}\par\noindent{\it\small
$^\dag$\,Institute of Mathematics of NAS of Ukraine, 3 Tereshchenkivska Str., 01024 Kyiv, Ukraine
\par}	

\vspace{2mm}\par\noindent{\it\small
$^\ddag$\,Department of Mathematics and Statistics, Memorial University of Newfoundland,\\
$\phantom{^\ddag}$\,St.\ John's (NL) A1C 5S7, Canada
\par}

\vspace{2mm}\par\noindent{\it\small
$^\S$\,Faculty of Mechanics and Mathematics, National Taras Shevchenko University of Kyiv,\\
$\phantom{^\S}$\,2 Academician Glushkov Ave., 03127 Kyiv, Ukraine
\par}

\vspace{4mm}\par\noindent
E-mails:
e.chapovskyi@imath.kiev.ua,
skoval@mun.ca,
oliazhur@knu.ua

%\maketitle

\vspace{8mm}\par\noindent\hspace*{10mm}\parbox{140mm}{\small
Revisiting the results by Winternitz
[{\it Symmetry in physics}, {\it CRM Proc. Lecture Notes} {\bf 34}, American Mathematical Society, Providence, RI, 2004, pp.~215--227], 
we thoroughly refine his classification of Lie subalgebras
of the real order-three special linear Lie algebra 
and thus present the correct version of this classification for the first time. 
A similar classification over the complex numbers is also carried out. 
We follow the general approach by Patera, Winternitz and Zassenhaus 
but in addition enhance it and rigorously prove its theoretical basis 
for the required specific cases of classifying subalgebras of real or complex finite-dimensional Lie algebras.
As a byproduct, we first construct complete lists of inequivalent subalgebras of the rank-two affine Lie algebra 
over both the real and complex fields.
}\par\vspace{4mm}

\vspace{3.5mm}
\par\noindent{\large\it
\hspace*{\stretch{1}}To the memory of Ji\v r\'i  Patera and Pavel Winternitz
\par}

\noprint{

}

\section{Introduction}~\label{sec:Introduction}

The problem of classifying Lie subalgebras of real and complex Lie algebras 
arises in many fields of mathematics and its applications.
In particular, for a system of (partial) differential equations, listing inequivalent Lie subalgebras of
its maximal Lie invariance algebra is used for the classification of Lie reductions of the system and in turn
can be applied to constructing its (inequivalent) explicit exact solutions,
see~\cite[Chapter 3]{olve1993A} and \cite{blum2010A,hydo2000A,popo2024a,wint2003a} for details and examples.
This is why classifications of subalgebras of Lie algebras have aroused 
a lively interest among researchers from the field of symmetry analysis of differential equations.
Such classifications are also efficient tools in theoretical physics and in the theory of integrable systems, 
e.g.,~\cite{camp2023a,pate1975a}.
At the same time, they themselves still remain interesting algebraic problems.

When studying continuous subgroups of the fundamental groups of physics
in the series of papers \cite{burd1978a,pate1976b,pate1976c,pate1977b,pate1976a,pate1975a,pate1975b},
Patera, Winternitz, Zassenhaus and others developed the general methods for classifying Lie subalgebras
of finite-dimensional Lie algebras with nontrivial ideals
and of direct products of Lie algebras. 
The latter method was named the {\it Lie--Goursat method} therein.
Although these methods became a reference point for carrying out such classifications,
to apply them it is usually necessary to examine the properties 
of the Lie algebra under consideration, its representations and the explicit form of its automorphism group \cite{pate1977a,wint2004a}.
This makes the subalgebra classification an inherently {\it ad hoc} problem,
which involves cumbersome and complex computations.
Due to the above reasons, 
the subalgebra classification problem was thoroughly and completely solved, 
to the best of our knowledge, only for a small number of low-dimensional Lie algebras
(over the fields of real and complex numbers),
see, in particular,~\cite{pate1977a}.

In \cite{wint2003a,wint2004a}, the classification of subalgebras of real
order-three special linear Lie algebra~$\mathfrak{sl}_3(\mathbb R)$ was carried out
using the general approach from~\cite{pate1975a}
reinforced by the properties of the defining representation of~$\mathfrak{sl}_3(\mathbb R)$.
This approach made it possible to reduce the classification of subalgebras of the simple Lie algebra~$\mathfrak{sl}_3(\mathbb R)$
to that of the real rank-two affine Lie algebra
$\mathfrak{aff}_2(\mathbb R)=\mathfrak{gl}_2(\mathbb R)\ltimes\mathbb R^2$.
The obtained classification is notably distinct in the sense that it uses brilliant and elegant ideas
and, as our analysis reveal, the obtained results are almost accurate.
Nevertheless, the final classification list,
which is presented in~\cite[Table~1]{wint2004a}, contains several misprints, errors,
redundant and missed subalgebras.
However, in view of the diverse variety of applications of this classification,
it is essential to have the correct list of subalgebras of the algebra $\mathfrak{sl}_3(\mathbb R)$.
This is why we comprehensively revisit this complicated problem.
It is also worth pointing out that the presentation of a complete list of inequivalent subalgebras of~$\mathfrak{aff}_2(\mathbb R)$
in~\cite{wint2004a} was omitted, which makes the main result difficult to reproduce or double-check.
In fact, 
the lists of ``twisted'' and ``nontwisted'' subalgebras of~$\mathfrak{aff}_2(\mathbb R)$
were presented in \cite[Section~3.3]{wint2004a} but the classification was not completed.
Moreover, the validity of these lists is still questionable. 
The reason for this doubt is that for classifying subalgebras 
of the algebra~$\mathfrak{aff}_2(\mathbb R)$ as the semidirect product $\mathfrak{gl}_2(\mathbb R)\ltimes\mathbb R^2$,
the method from~\cite{pate1975a} requires a correct list of inequivalent subalgebras of~$\mathfrak{gl}_2(\mathbb R)$,
but the one in~\cite[Eq.~(3.11)]{wint2004a} contains a misprint.
Although these disadvantages are significant,
it is possible to overcome them using~\cite{wint2004a} as the source of excellent ideas and approaches.

It is quite noteworthy that the classification of subalgebras of the Lie algebra~$\mathfrak{sl}_3(\mathbb C)$
was only recently presented by Douglas and Repka~\cite{doug2016a}.
This classification was carried out in a straightforward manner,
yet yielded successful results.
Nonetheless, thoroughly revising this classification presents a significant challenge.
By using the methods from~\cite{wint2004a}, which we enhance in this paper,
we classify subalgebras of the algebra~$\mathfrak{sl}_3(\mathbb C)$ and then compare the obtained results
with those presented in~\cite{doug2016a}.
As our comparative analysis shows, the number of discrepancies in~\cite{doug2016a} is remarkably minimal,
given the complex and cumbersome nature of the research.

The classification of subalgebras of $\mathfrak{sl}_3(\mathbb F)$, where $\mathbb F=\mathbb R$ or $\mathbb C$,
is important for many open problems in algebra, mathematical physics and group analysis of differential equations.
For instance, as shown in~\cite{boyk2021a},
the group classification problem for normal linear systems of second-order ordinary differential equations with $n$ dependent variables
can be easily solved once the classification of subalgebras of~$\mathfrak{sl}_n(\mathbb F)$ is known.
Therefore, having a complete list of the subalgebras of~$\mathfrak{sl}_3(\mathbb F)$ enables us to solve this problem for $n=3$.
The same problem for $n=2$ was solved in~\cite{boyk2021a},
based on the well-known classification of subalgebras of~$\mathfrak{sl}_2(\mathbb F)$.
Moreover, many systems of linear and nonlinear partial differential equations
admit Lie symmetry algebras that are isomorphic to the algebra $\mathfrak{sl}_n(\mathbb F)$.
Such systems are often constituted by equations of the same structure with coupled dependent variables,
e.g., as the systems of reaction--diffusion equations~\cite{niki2001a}.
Any (in)homogeneous (real or complex) Monge--Amp\'ere equation with two independent variables 
possesses a Lie invariance algebra isomorphic to~$\mathfrak{sl}_3(\mathbb F)$~\cite{khab1990a}.
We can also refer to the equation of \c{T}i\c{t}eica surfaces~\cite{udri1999a}
as one more example of a single partial differential equation 
whose Lie invariance algebra is isomorphic to~$\mathfrak{sl}_3(\mathbb R)$.

Another important open problem that relies on listing inequivalent subalgebras of the algebra~$\mathfrak{sl}_3(\mathbb F)$
is the classification of its realizations, i.e., its representations as Lie algebras of vector fields,
with a view towards solving the inverse group classification problem for this algebra.
This inverse problem consists in constructing systems of differential equations
that admit Lie invariance algebras isomorphic to~$\mathfrak{sl}_3(\mathbb F)$,
see, e.g.,~\cite{popo2012a}.
The subalgebras of $\mathfrak{sl}_3(\mathbb F)$ can be also useful for constructing
superintegrable systems and algebraic Hamiltonians \cite{camp2023a,marq2023a}.

Among the open problems in algebra, the classification of nilpotent Lie algebras over a field~$\mathbb F$ occupies a significant place.
The subalgebras of $\mathfrak{sl}_3(\mathbb F)$ also play a specific, but at the same time
important role in solving this problem for Lie algebras of low dimension,
namely, the classification of inequivalent maximal abelian nilpotent subalgebras of $\mathfrak{sl}_3(\mathbb F)$
is used for classifying five-dimensional nilpotent Lie algebras, see~\cite[Section~8.1]{snob2014A}.

The study of specific (graded) contractions of~$\mathfrak{sl}_3(\mathbb C)$ was initiated in~\cite{hriv2006a}.
The usual contractions and degenerations of~$\mathfrak{sl}_3(\mathbb F)$ are also of interest and deserve investigation.
It is clear that given a list of inequivalent subalgebras of the algebra~$\mathfrak{sl}_3(\mathbb F)$, 
one can straightforwardly construct and classify its In\"on\"u--Wigner contractions 
following, for example,~\cite{inon1953a,inon1954a,nest2006a} and references therein.

In fact, there were attempts to apply the list of subalgebras of $\mathfrak{sl}_3(\mathbb R)$ from~\cite{wint2004a}
for studying the local limits of connected subgroups of the group~${\rm SL}_3(\mathbb R)$~\cite{laza2021a}.
This gives us an additional reason to revisit the results of~\cite{wint2004a}.

\looseness=-1
The purpose of this paper is to correct and refine the classification of subalgebras of
the real order-three special linear Lie algebra~$\mathfrak{sl}_3(\mathbb R)$ following~\cite{wint2004a}.
The structure of the paper is as follows.
In Section~\ref{sec:LieSubalgebras}, we revisit classical approaches
for the classification of subalgebras of real and complex Lie algebras,
suggest novel perspectives on them and provide them with a rigorous theoretical framework.
% The choice of an approach depends on the structure of the Lie algebra under consideration.
Section~\ref{sec:aff2Classification} is devoted to the classification 
of subalgebras of the affine Lie algebra~$\mathfrak{aff}_2(\mathbb R)$.
To the best of our knowledge, such an exhaustive classification has never been presented in the literature,
cf.\ \cite[Table 1]{poch2017a}, where only the ``appropriate'' subalgebras were classified with respect to a ``weaker'' equivalence
than that generated by the action of the group~${\rm Aff}_2(\mathbb R)$,
and \cite[Section~3.3]{wint2004a}, where a complete list of subalgebras of~$\mathfrak{aff}_2(\mathbb R)$ was omitted.
In Section~\ref{sec:sl3Classification}, we carry out the classification of subalgebras of~$\mathfrak{sl}_3(\mathbb R)$,
which essentially relies on the classification of subalgebras of~$\mathfrak{aff}_2(\mathbb R)$.
After listing the subalgebras of~$\mathfrak{sl}_3(\mathbb R)$ in Theorem~\ref{thm:SubalgebrasOfSL3},
we compare the obtained list with that in~\cite[Table~1]{wint2004a} in  Section~\ref{sec:ClassCompare}, 
Remark~\ref{rem:DiffBetwClassificationsR} and Table~\ref{tab:ListsCompare}.
Our analysis reveals that Winternitz's classification has two incorrect families of subalgebras,
one redundant single subalgebra, one omitted subalgebra and two subalgebras with misprints.
In the same way, in Section~\ref{sec:sl3CClassification}, we carry out
the classification of subalgebras of the Lie algebra~$\mathfrak{sl}_3(\mathbb C)$
and then in Section~\ref{sec:ClassCompare2} we compare the results of the classification
that are given in Theorem~\ref{thm:SubalgebrasOfSL3C} with those presented in~\cite[Table~1]{doug2016a}.
The results of the paper are summarized in Section~\ref{sec:Conclusion}.

Throughout the paper, all Lie algebras and Lie groups are assumed to be finite-dimensional
and the underlying field is $\mathbb R$ or $\mathbb C$.

\section{Lie algebras and classification of their subalgebras}\label{sec:LieSubalgebras}

Let $\g$ be a Lie algebra
and $G$ be a connected Lie group corresponding to~$\g$.
The algebra~$\g$ is identified with the algebra of
tangent vectors at the origin of the group $G$ with the standard commutator. 
The group~$G$ acts on the algebra~$\g$ by the adjoint representation,
\[
\mathop{\rm Ad}\colon G\to {\rm Aut}(\g),
\quad
g\mapsto\mathop{{\rm Ad}_g}.
\]
Here $\mathop{{\rm Ad}_g}$ denotes the derivative at the identity element $e\in G$ of the group conjugation $c_g\colon h\mapsto ghg^{-1}$,
$\mathop{{\rm Ad}_g}:=(\mathrm dc_g)_e\colon \g\to\g$,
where $\mathrm d$ is the differential and $g,h\in G$.
Since the group $G$ is connected, this action coincides (up to the trivial action of the center of~$G$)
with the action
of the inner automorphism group ${\rm Inn}(\g)$ on the Lie algebra~$\g$.

Two subalgebras $\mathfrak s_1$ and $\mathfrak s_2$ of the algebra~$\g$ are called {\it conjugate}
(or {\it $G$-equivalent}) if there exists an element $g\in G$ such that $\mathop{{\rm Ad}_g}\mathfrak s_1=\mathfrak s_2$.
The problem of subalgebra classification consists in determining a complete list of representatives of
the conjugacy classes of subalgebras of the algebra~$\g$ up to the action of the group $G$.
In other words, considering the action of the group~$G$ on the set of subalgebras of the algebra $\g$,
which is induced by the adjoint representation,
construct a list of canonical representatives (which are of the simplest possible form) of the orbits of this action.%
\footnote{
It is natural that this question can be formulated in a more general setting:
given a Lie algebra $\g$ over a~field~$\mathbb F$,
construct a complete and irredundant list of canonical representatives of the orbits of the subalgebras of the algebra $\g$
under the action of its inner automorphism group~${\rm Inn}(\g)$.%
}

In view of Levi--Malcev theorem, $\g$ splits over its radical~$\mathfrak r$,
$\g=\mathfrak f\ltimes\mathfrak r$,
where $\mathfrak f$ is a Levi subalgebra of $\g$,
which is complementary to $\mathfrak r$ in $\g$.
Algebra $\g$ can take one of the following forms:
\begin{enumerate}\itemsep=0ex
\item[(1)]
$\g$ is simple;
\item[(2)]
$\g$ is a direct product of its subalgebras, $\g=\g_1\times\g_2$;
\item[(3)]
$\g$ is a semidirect product of its subalgebras, $\g=\g_1\ltimes\g_2$.
\end{enumerate}

Cases 1 and 2 include the situation in which~$\g$ is a semisimple Lie algebra, i.e., $\mathfrak r=\{0\}$.
Moreover, case 2 takes place when the corresponding action of~$\mathfrak f$ on~$\mathfrak r$
in the semidirect decomposition~$\mathfrak f\ltimes\mathfrak r$ is trivial.
It is also evident that case 3 occurs when both the factors~$\mathfrak f$ and~$\mathfrak r$ in the Levi decomposition are nontrivial,
but it also covers a possibility when~$\g$ is itself a solvable Lie algebra
(in this case, the algebra~$\g$ has a codimension-one ideal).

The approaches for classifying Lie subalgebras of a Lie algebra~$\g$ depend on the structure of~$\g$,
more specifically, they depend on which of the above cases takes place.
This is why we revisit the approaches for the classification of subalgebras of Lie algebras
over the fields of real or complex numbers
following~\cite{wint2004a}, addressing each of the possible cases separately.
In fact, we restrict our consideration of case~3 to the semidirect products of Lie algebras~$\g_1\ltimes\g_2$,
where the ideal~$\g_2$ is an abelian Lie algebra, for the sake of clarity and readability.
This assumption is sufficient for classifying subalgebras of~$\mathfrak{sl}_3(\mathbb R)$,
which is the primary objective of the paper.
The general case is discussed in~\cite{wint2004a}.

\subsection{Subalgebras of simple Lie algebras}\label{subsec:Simple}

The idea for classifying subalgebras of a simple Lie algebra~$\g$ 
consists in finding all its maximal subalgebras~$\mathfrak m$
and then proceeding with the classification of subalgebras for each~$\mathfrak m$
with respect to ${\rm Inn}(\mathfrak m)$-equivalence.
This can be done using one of the methods described in this section
depending on the structure of $\mathfrak m$.
The task of finding all maximal subalgebras of a Lie algebra
is within the domain of representation theory, and we elaborate on the main approach at the end of this subsection.

After listing the ${\rm Inn}(\mathfrak m)$-inequivalent subalgebras of~$\mathfrak m$,
the obtained lists should be combined modulo the $G$-equivalence.
The following proposition shows the general idea of how to do this efficiently.

\begin{proposition}\label{prop:ListMerging}
Let $\mathfrak m\subset\g$ be a Lie subalgebra and $M\subset G$ be the corresponding connected Lie subgroup.
Choose some subset $C\subset G$ such that $MC = G$.
Then the Lie subalgebras $\h_1 \subset \mathfrak m$ and $\h_2\subset\g$
are conjugate if and only if there exists an element $g \in C$ such that
$\Ad_g\h_2\subset\mathfrak m$ and moreover $\Ad_g\h_2$
is equivalent to $\h_1$ up to $\Inn(\mathfrak m)$-equivalence.
\end{proposition}

Proposition~\ref{prop:ListMerging} simplifies the problem of combining lists as follows.
For each maximal Lie subalgebra $\mathfrak m$ and a Lie subalgebra $\h\subset\g$
suspected to be conjugate to a Lie subalgebra~of~$\mathfrak m$,
find all elements $g\in C$ such that $\Ad_g\h\subset\mathfrak m$.
For all such Lie subalgebras $\Ad_g\h\subset\mathfrak m$ 
find the corresponding representatives in the list of inequivalent subalgebras
of $\mathfrak m$ up to the $\Inn(\mathfrak m)$-equivalence.
The subalgebras obtained in this way are the only representatives of Lie subalgebras of~$\mathfrak m$
that are equivalent to~$\h$ modulo the $\Inn(\g)$-equivalence.

\begin{remark}\label{rem:ListMerging}
To find a suitable subset~$C$ of~$G$, it is natural to consider
the exponent of some complement subspace $V\subset\g$ to the Lie subalgebra~$\mathfrak m$.
From the notion of canonical coordinate systems on a Lie group, it follows
that $C$ can be chosen as at most countable union of subsets of the form
$\exp(V)g$ for some elements $g\in G$.
For compact Lie groups, this union can be chosen to be finite.
The explicit form of such subsets~$C$ should be found
on a case-by-case basis, depending on the structure of the Lie group~$G$.
\end{remark}

To find maximal subalgebras of the Lie algebra $\g$,
consider a faithful irreducible finite-dimensional representation $\rho\colon\g\to\mathfrak{gl}(V)$.
It is convenient to choose~$V$ of the least possible dimension,
e.g., for the classical Lie algebras choose their defining representations.
Fixing the representation $\rho$, any subalgebra of~$\g$ is either {\it irreducibly} or {\it reducibly
embedded} in~$\g$ with respect to~$\rho$. 

%Since any subalgebra of~$\g$ is a subalgebra in some maximal subalgebra~$\mathfrak m$ of~$\g$,
%it is sufficient to consider the maximal subalgebras arising from the {\it irreducibly} and {\it reducibly
%embedded subalgebras} of~$\g$ with respect to the representation~$\rho$.

A subalgebra $\mathfrak s$ of $\g$ is called irreducibly embedded if $\rho(\mathfrak s)$
has no proper invariant subspace of~$V$.
Then in view of~\cite[Proposition~5, p.~56]{bour1975LiePart1},
the algebra $\mathfrak s$ is a reductive Lie algebra,
i.e., a simple or a semisimple one, or a direct sum of a semisimple and an abelian Lie algebras.
Finding such subalgebras becomes a question of representation theory,
since $\mathfrak s$ has a faithful representation of dimension $\dim V$.

A subalgebra $\mathfrak s$ of $\g$ is called reducibly embedded
if $\rho(\mathfrak s)$ has a proper invariant subspace $V_0\subset V$.
In this case, it is sufficient  to classify such subspaces $V_0$ of $V$
and find the subalgebra~$\mathfrak s$ of~$\g$ such that $\rho(\mathfrak s)$ leaves $V_0$ invariant.
This is again a representation theory problem.
In particular, if $\g$ is $\mathfrak{sl}_n(\mathbb R)$ or $\mathfrak{sl}_n(\mathbb C)$,
then $V_0$ is completely characterized by its dimension.
If $\g$ is an orthogonal or a symplectic Lie algebra,
then $V_0$ is determined by its dimension and signature.

\subsection{Subalgebras of direct products} \label{subsec:Direct}
Let $G_1$, $G_2$ be connected Lie groups and $\g_1$, $\g_2$ be the corresponding Lie algebras,
$G=G_1\times G_2$,
$\g=\g_1\times\g_2$
and $\pi_i\colon\g\to\g_i$, $i=1,2$, are the natural projections.
For a subalgebra $\mathfrak h$ of $\mathfrak g$, we denote by ${\rm Nor}(\mathfrak h,G)$
the normalizer subgroup of $\mathfrak h$ in $G$.

%\begin{proposition} \label{sub1}
%Consider the map~$\Phi$ that assigns, to each Lie subalgebra $\h\subset\g$,
%the tuple of objects $(\h_1,\h_2,I,\alpha)$,
%where $\h_i\subset \g_i$, $i=1,2$, are Lie subalgebras,
%$I$ is an ideal of the Lie subalgebra~$\h_2$,
%and~$\alpha\colon\h_1\to\h_2/I$ is a Lie algebra epimorphism.
%The map is given by $\h_i:=\pi_i(\h)$, $I:=\h\cap\g_2$,
%\[
%\alpha \colon \h_1 \ni x_1 \longmapsto
%\{x_2 \in \h_2 \mid (x_1, x_2) \in \h\} \in \h_2 / I.
%\]
%There exists the inverse $\Phi^{-1}$, which maps a quadruple $(\h_1, \h_2, I, \alpha)$
%to the Lie subalgebra $\h \subset \g$ as follows:
%\[
%\h = \{(x_1, x_2) \mid x_1 \in \h_1,\; x_2 \in \alpha(x_1)\}.
%\]
%\end{proposition}

\begin{proposition} \label{sub1}
Consider the map~$\Phi$ that assigns, to each Lie subalgebra $\h\subset\g$,
the tuple of objects $(\h_1,\h_2,I,\alpha)$,
where $\h_i\subset \g_i$, $i=1,2$, are Lie subalgebras,
$I$ is an ideal of the Lie subalgebra~$\h_2$,
and~$\alpha\colon\h_1\to\h_2/I$ is a Lie algebra epimorphism.
Specifically, $\Phi$ is given by the correspondence
\begin{gather*}
\Phi\colon\mathfrak h\longmapsto(\h_1,\h_2,I,\alpha)\quad\mbox{with}\quad
\h_i:=\pi_i(\h),\quad I:=\h\cap\g_2\quad\mbox{and}
\\
\qquad
\alpha \colon \h_1 \ni x_1 \longmapsto
\{x_2 \in \h_2 \mid (x_1, x_2) \in \h\} \in \h_2 / I.
\end{gather*}
There exists the inverse $\Phi^{-1}$, which maps a quadruple $(\h_1, \h_2, I, \alpha)$
to the Lie subalgebra $\h \subset \g$ as follows:
\[
\Phi^{-1}\colon(\h_1, \h_2, I, \alpha)\longmapsto
\h:=\{(x_1, x_2) \mid x_1 \in \h_1,\; x_2 \in \alpha(x_1)\}.
\]
\end{proposition}

\begin{proof}
The essential step of the proof is to show that both the maps from the assertion are well defined.
It will immediately follow from their construction that they are inverses of each other.

We start by showing that $\Phi$ is well defined.
Let $\h$ be a Lie subalgebra of the Lie algebra~$\g$.
Projections $\pi_i$ are Lie algebra homomorphisms, and therefore $\h_i$ are Lie subalgebras of $\g_i$.
	
Further, we have
$I = \h \cap \g_2 = \ker \pi_1 \vert_{\h}$, and hence
$I$ is an ideal of the Lie subalgebra~$\h$.
Since the restriction~$\pi_2\vert_\h \colon \h \to \h_2$ is surjective by definition,
it follows that $I = \pi_2(I)$ is an ideal in $\h_2$.
	
Next, consider the map $\alpha \colon \h_1 \to \h_2 / I$.
It is easy to see that it is a linear map.
For arbitrary elements $(x_1, x_2), (y_1, y_2) \in \h$,
their commutator $([x_1, y_1], [x_2, y_2])$ also belongs to $\h$,
hence
\[
\alpha([x_1, y_1]) = [x_2, y_2] + I = [x_2 + I, y_2 + I] = [\alpha(x_1), \alpha(y_1)],
\]
i.e., $\alpha$ is a Lie algebra homomorphism.
It is clear that $\alpha$ is an epimorphism.

To show that $\Phi^{-1}$ is well defined,
consider a quadruple consisting of
Lie subalgebras $\h_i \subset \g_i$, $i=1,2$, an ideal $I \subset \h_2$ and
a Lie algebra epimorphism $\alpha \colon \h_1 \to \h_2 / I$.
We show that the subspace
$\h := \{(x_1, x_2) \mid x_1 \in \h_1,\; x_2 \in \alpha(x_1)\} \subset \g$
is a Lie subalgebra.
For any elements $(x_1,x_2),(y_1,y_2)\in\h$, 
we have that $[x_1,y_1]\in\h_1$, and moreover
\[
\alpha([x_1,y_1])=[\alpha(x_1),\alpha(y_1)]=[x_2+I,y_2+I]=[x_2,y_2]+I.
\]
This means that the commutator $([x_1, y_1], [x_2, y_2])$ is an element of~$\h$,
hence $\h$ is indeed a Lie subalgebra.
\end{proof}

\begin{proposition}\label{tuple1}
The adjoint representation $\Ad$ of the Lie group $G$
and its action on the set of the subalgebras of the Lie algebra~$\g$
induces the natural action on the quadruples $(\h_1, \h_2, I, \alpha)$
via the map $\Phi$ from Proposition~{\rm\ref{sub1}},
\[
\Phi\circ{\rm Ad}_g\circ\Phi^{-1}\colon(\h_1,\h_2,I,\alpha)\longmapsto
\bigl({\rm Ad}_{g_1}\h_1,{\rm Ad}_{g_2}\h_2,{\rm Ad}_{g_2}I,\smash{{\rm Ad}_{g_2}\circ\alpha\circ{\rm Ad}_{g_1^{-1}}}\bigr),
\]
where $g=(g_1,g_2)\in G$.%
\footnote{
By abuse of notation the Lie algebra epimorphism
$\smash{{\rm Ad}_{g_2}\circ\alpha\circ{\rm Ad}_{g_1^{-1}}}\colon{\rm Ad}_{g_1}\h_1\to{\rm Ad}_{g_2}\h_2/{\rm Ad}_{g_2}I$
is defined via the map ${\rm Ad}_{g_1}x\mapsto {\rm Ad}_{g_2}x+{\rm Ad}_{g_2}I$,
where $x\in\h_1$.
}
\end{proposition}

To classify all Lie subalgebras in $\g$ up to its inner automorphisms,
it is sufficient to go through the following steps:
\begin{enumerate}\itemsep=0ex
\item
Construct a list of representatives $\hat{\h}_i^a$
of the conjugacy classes of Lie subalgebras of $\g_i$ modulo $G_i$-equivalence,
$i=1,2$.%
\footnote{
Notation~$\hat{\h}_i^a$ is usually replaced by the one of the form $\h^*_{d.k}$,
where $d$, $k$ and $*$ stand for the subalgebra's dimension, its number in the list of subalgebras of dimension~$d$ 
and the parameters for a subalgebra family, respectively.
}

\item
For each pair $(\hat \h^a_1,\hat \h^b_2)$, find
representatives $\hat I^c_{ab}$
of the conjugacy classes of ideals of $\hat{\h}^b_2$
up to the action of the group $\Nor(\hat{\h}^b_2, G_2)$.

\item 
For each tuple $(\hat{\h}^a_1,\hat{\h}^b_2,\hat{I}^c_{ab})$, obtain representatives
of conjugacy classes of Lie algebras epimorphisms
$\hat\alpha^d\colon\hat{\h}^a_1\to\hat{\h}^b_2/\hat{I}^c_{ab}$
up to the action of the group
\[
\Nor(\hat{\h}^a_1, G_1)\times\big(\Nor(\hat{\h}^b_2, G_2)\cap\Nor(\hat I^c_{ab}, G_2)\big),
\]
where the latter action on the epimorphisms $\alpha\colon\hat{\h}^a_1\to\hat{\h}^b_2/\hat{I}^c_{ab}$ is given by
$
(g_1,g_2)\alpha={\rm Ad}_{g_2}\circ\alpha\circ{\rm Ad}_{g_1}^{-1}.
$

\item
For each quadruple $(\hat{\h}^a_1, \hat{\h}^b_2,\hat I^c_{ab}, \hat\alpha^d_{abc})$,
recover the corresponding Lie subalgebra~$\h^{abcd}$ according to Proposition~\ref{sub1}.

\end{enumerate}

\begin{remark}
The problem of finding the normalizer groups $\Nor(\hat{\h}^a_1, G_1)$ and $\Nor(\hat{\h}^b_2, G_2)$
from the above method is a nonlinear problem, which can be sophisticated.
However, to slightly reduce the complexity of the computations
while applying the algorithm,
it is possible to replace the normalizer groups $\Nor(\hat{\h}^a_1, G_1)$ and $\Nor(\hat{\h}^b_2, G_2)$ by their subgroups generated 
by the exponents of the normalizers of $\hat{\h}^a_1$ and $\hat{\h}^b_2$ in the corresponding Lie algebras, respectively,
i.e., the subgroups in $G_1$ and $G_2$ respectively generated by the sets $\exp\big(\Nor(\hat{\h}^a_1, \g_1)\big)$ and $\exp\big(\Nor(\hat{\h}^b_2, \g_2)\big)$.
Since these subgroups can be proper in the corresponding normalizer groups,
one should check that the resulting list of the Lie subalgebras
contains no conjugate pairs.
\end{remark}

\begin{remark}
The algorithm presented in this section is a modified version of the Goursat method described in \cite[Section~2.3]{wint2004a}.
However, the original method can only construct the Lie subalgebras corresponding to the quadruples $(h_1,h_2,I,\alpha)$,
where~$I$ is either the zero ideal or $\mathfrak h_2$,
in other words, $I$ is an improper ideal.
The advantage of our algorithm is that it has an invariant formulation and a proven validity,
which allows us to overcome the latter mistake.
\end{remark}

\subsection{Subalgebras of semidirect products}\label{subsec:semidirect}

Let $G_1$ and $G_2$ be connected Lie groups with $G_2$ isomorphic to $\mathbb{R}^n$, $\g_1$ and $\g_2$ are the corresponding Lie algebras, in particular $\g_2$ is an abelian Lie algebra,
$G = G_1 \ltimes_{\Psi} G_2$,
$\g = \g_1 \ltimes_{\psi} \g_2$,
$\pi_i \colon \g \to \g_i$, $i=1,2$, are the natural projections,
$\Psi\colon G_1\to{\rm Aut}(G_2)$ is a homomorphism of Lie groups and
$\psi\colon\g_1\to{\rm Der}(\g_2)$ is the corresponding Lie algebra homomorphism.

\begin{proposition}\label{prop:SubalgebrasSemidirect}
Consider the map~$\Gamma$ that assigns, to each Lie subalgebra~$\h\subset\g$,
the triple of objects $(\h_1,U,\alpha)$,
where $\h_1\subset\g_1$ is a Lie subalgebra,
$U\subset\g_2$ is a $\h_1$-submodule
and a $1$-cocycle $\alpha\colon\h_1\to\g_2/U$,
i.e., a linear map satisfying the identity (see {\rm \cite[p.~89]{bour1975LiePart1}})
\[
\alpha([x,y])=[\alpha(x),y]+[x,\alpha(y)]=\psi(x)\alpha(y)-\psi(y) \alpha(x).
\]
Specifically, $\Gamma$ is given by the correspondence
\begin{gather*}
\Gamma\colon\mathfrak h\longmapsto(\h_1,U,\alpha)\quad\mbox{with}\quad
\h_1=\pi_1(\h),\quad U =\h\cap\g_2\quad\mbox{and}
\\ \qquad
\alpha\colon\h_1\ni x\longmapsto\{v\in\g_2\mid(x,v)\in\h\}\in\g_2/U.
\end{gather*}
There exists the inverse $\Gamma^{-1}$, which maps a triple $(\h_1,U,\alpha)$
to the Lie subalgebra $\h\subset\g$ as follows:
\[
\Gamma^{-1}\colon(\h_1,U,\alpha)\longmapsto
\h:=\{(x,v)\mid x\in\h_1,\; v\in\alpha(x)\}.
\]
\end{proposition}
\begin{proof}
In the same manner as for the direct product case, it suffices to show that both the above maps are well-defined.
It is clear that they are mutual inverses from their construction.

Starting with $\Gamma$, consider a Lie subalgebra $\h$ in $\g$.
Then $\h_1=\pi_1(\h)$ is a Lie subalgebra in~$\g_1$,
$U = \h \cap \g_2$ is a $\h$-invariant subspace in $\g_2$,
and hence it is $\h_1$-invariant as well.
It is obvious that $\alpha \colon \h_1 \to \g_2 / U$ is a linear map.
For any elements $(x_1,v_1),(x_2,v_2)\in\h$,
we have that $[(x_1,v_1),(x_2,v_2)]\in\h$, that is
$\big([x_1,x_2],\psi(x_1)v_2-\psi(x_2)v_1\big)\in\h$
and therefore
\begin{align*}
\alpha([x_1,x_2])
&=\psi(x_1)v_2-\psi(x_2)v_1+U \\
&=\psi(x_1)(v_2 + U)-\psi(x_2)(v_1+U) \\
&=\psi(x_1)\alpha(x_2)-\psi(x_2)\alpha(x_1),
\end{align*}
which means that $\alpha$ is a one-cocycle.

Next, we show that the inverse $\Gamma^{-1}$ is well-defined as well.
Let $\h_1$, $U$ and $\alpha\colon\h_1\to\g_2/U$
be a Lie subalgebra of~$\g_1$,
an $\h_1$-submodule of $\g_2$ and
a one-cocycle, respectively.
Consider the subspace $\h=\{(x,v)\mid x\in\h_1,v\in\alpha(x)\}$ of $\g$.
For any elements $(x_1,v_1),(x_2,v_2)\in\h$ the commutator
$[x_1,x_2]$ belongs to $\h_1$,
and moreover
\begin{align*}
\alpha([x_1,x_2])
&=\psi(x_1)\alpha(x_2)-\psi(x_2)\alpha(x_1) \\
&=\psi(x_1)(v_2 + U)-\psi(x_2)(v_1+U) \\
&=\psi(x_1)v_2-\psi(x_2)v_1+U.
\end{align*}
Therefore, $[(x_1,v_1),(x_2,v_2)]\in\h$.
\end{proof}

\begin{proposition} \label{tuple2}
Analogously to the case of direct product,
the action of the adjoint representation~$\Ad G$ on the set of all subalgebras
of the Lie algebra~$\g$ induces the natural action on the triples
$(\h_1, U, \alpha)$ via the map~$\Gamma$ from Proposition~{\rm\ref{prop:SubalgebrasSemidirect}},
\begin{gather*}
\Gamma \circ \Ad (g, 0) \circ \Gamma^{-1}\colon (\h_1, U, \alpha) \longmapsto
(\Ad g \ \h_1, \Psi(g) U, \Psi(g) \circ \alpha \circ \Ad g^{-1}), \\
\Gamma \circ \Ad (e, v) \circ \Gamma^{-1}\colon (\h_1, U, \alpha) \longmapsto
(\h_1, U, \alpha + \ad v\vert_{\h_1} ).
\end{gather*}
\end{proposition}

To classify all Lie subalgebras in $\g$ up to its inner automorphisms,
it is sufficient to go through the following steps.
\begin{enumerate}\itemsep=0ex
\item
Construct representatives $\hat{\h}^{a}_1$
of the conjugacy classes of Lie subalgebras in $\g_1$ modulo $G_1$-equivalence.
\item
For each $\hat{\h}^{a}_1$, construct
representatives $\hat{U}^{b}_{a}$
of the orbits of $\hat{\h}^a_1$-submodules in $\g_2$
up to the group $\Psi\big(\Nor(\hat{\h}^{a}_1, G_1)\big)$.
\item
For each pair $(\hat\h^a_1,\hat U^b_a)$,
construct representatives
$\hat{\alpha}^{c}_{ab}\colon\hat{\h}^{a}_1 \to \g_2 / \hat{U}^{b}_a$
of the orbits of cocycles up to the group
\[
\big(\Nor(\hat{\h}^{a}_1, G_1)\cap\Nor(\hat{U}^{b}_a, G_1)\big) \ltimes G_2,
\]
where the latter action is defined as
$
(g,v)\alpha=\Psi(g)\circ\alpha\circ\Ad g^{-1}+\ad v\vert_{\h_1}.
$
\item
For each triple
$(\hat{\h}^{a}_1,\hat{U}^{b}_{a},\hat{\alpha}^{c}_{ab})$
recover the corresponding Lie subalgebra~$\h^{abc}$ according to Proposition \ref{prop:SubalgebrasSemidirect}.
\end{enumerate}
%\end{proposition}

\begin{remark}
The presented algorithm corresponds to the method introduced in \cite[Section 2.3]{wint2004a}.
We reformulate the original method by Winternitz in an invariant way for the case where the ideal~$\mathfrak g_2$ in the semidirect product is an abelian Lie algebra
and prove its validity.
Up to some correction, the steps in the algorithm are the same as in the original version.
\end{remark}

\section{Subalgebras of $\mathfrak{aff}_2(\mathbb R)$}\label{sec:aff2Classification}

Using the approach from Section~\ref{subsec:Simple}, it is shown in Section~\ref{sec:sl3Classification} 
that the classification of subalgebras of the Lie algebra~$\mathfrak{sl}_3(\mathbb R)$ essentially reduces to that for the algebra $\mathfrak{aff}_2(\mathbb R)$.
Moreover, to the best of our knowledge, the latter classification has not been carried out completely in literature before,
cf. \cite[Table 1]{poch2017a},
where this problem has been solved in the more particular case,
more specifically, only the ``appropriate'' subalgebras of $\mathfrak{aff}_2(\mathbb R)$ were classified with respect to a ``weaker'' equivalence
than that generated by the action of ${\rm Aff}_2(\mathbb R)$.
In fact, in~\cite[Section~3.3]{wint2004a} the lists of ``twisted'' and ``nontwisted'' subalgebras of $\mathfrak{aff}_2(\mathbb R)$ were presented without a proof,
but combining these lists with respect to ${\rm Inn}\big(\mathfrak{aff}_2(\mathbb R)\big)$-equivalence was omitted.
Furthermore, the classification of subalgebras of $\mathfrak{aff}_2(\mathbb R)$ requires a correct list of the subalgebras of $\mathfrak{gl}_2(\mathbb R)$,
but the one in~\cite[eq.~(3.1)]{wint2004a} had a mistake: the subalgebra $\langle D\rangle$ (in their notation) was missing.
Therefore, the validity of the lists of ``twisted'' and ``nontwisted'' subalgebras is doubtful.
This is why, in this section, we present a complete list of inequivalent subalgebras
of~$\mathfrak{aff}_2(\mathbb R)$ and then, in the following section,
use it in the course of classifying subalgebras of~$\mathfrak{sl}_3(\mathbb R)$. 

The affine Lie algebra $\mathfrak{aff}_2(\mathbb R)$ of rank two can be written as the semidirect product 
$\mathfrak{aff}_2(\mathbb R)=\mathfrak{gl}_2(\mathbb R)\ltimes \mathbb R^2$,
where the algebra $\mathfrak{gl}_2(\mathbb R)$ is spanned by the matrices
\begin{gather*}
e_1=
\begin{pmatrix}
0&0\\
1&0
\end{pmatrix},\quad
e_2=\frac12
\begin{pmatrix}
1& 0\\
0&-1
\end{pmatrix},\quad
e_3=
\begin{pmatrix}
0&-1\\
0&0
\end{pmatrix},\quad
e_4=\frac12
\begin{pmatrix}
1&0\\
0&1
\end{pmatrix},
\end{gather*}
and $\mathbb R^2=\langle f_1,f_2\rangle$,
where $f_1=(1,0)^{\mathsf T}$, $f_2=(0,1)^{\mathsf T}$,
is considered as a two-dimensional abelian Lie algebra.
The action of $\mathfrak{gl}_2(\mathbb R)$ on $\mathbb R^2$
in the semidirect product is given by the natural faithful representation on~$\mathbb R^2$,
i.e., by the action of matrices on vectors.
Up to the skew-symmetry of the Lie brackets, the nontrivial commutation relations in $\mathfrak{aff}_2(\mathbb R)$
are exhausted~by
\begin{gather*}
[e_1,e_2]=e_1,\quad
[e_2,e_3]=e_3,\quad
[e_1,e_3]=2e_2,
\\
[e_1,f_1]=f_2,\quad
[e_2,f_1]=\tfrac12f_1,\quad
[e_2,f_2]=-\tfrac12f_2,\quad
[e_3,f_2]=-f_1,
\\
[e_4,f_1]=\tfrac12f_1,\quad
[e_4,f_2]=\tfrac12f_2.
\end{gather*}

Since the algebra $\mathfrak{aff}_2(\mathbb R)$ can be viewed as a semidirect product of Lie algebras,
we use the algorithm from Section~\ref{subsec:semidirect} to classify its subalgebras.
The following proposition completes the first step of that algorithm.

\begin{theorem}[\cite{pate1977a, popo2003a}]\label{thm:SubalgebrasOfGL2}
A complete list of inequivalent proper subalgebras of the algebra~$\mathfrak{gl}_2(\mathbb R)$
is exhausted by the following subalgebras:
\begin{align*}
{\rm 1D}\colon\ &
\h_{1.1}            =\langle e_1\rangle,\ 
\h_{1.2}            =\langle e_4\rangle,\ 
\h_{1.3}^\kappa     =\langle e_2+\kappa e_4\rangle,\ 
\h_{1.4}^\varepsilon=\langle e_1+\varepsilon e_4\rangle,\ 
\h_{1.5}^\gamma     =\langle e_1+e_3+\gamma e_4\rangle,
\\[1ex]
{\rm 2D}\colon\ &
\h_{2.1}       =\langle e_1,e_4\rangle,\quad
\h_{2.2}       =\langle e_2,e_4\rangle,\quad
\h_{2.3}       =\langle e_1+e_3,e_4\rangle,\quad
\h_{2.4}^\gamma=\langle e_2+\gamma e_4,e_1\rangle,
\\[1ex]
{\rm 3D}\colon\ &
\h_{3.1}=\langle e_1,e_2,e_4\rangle,\quad
\h_{3.2}=\langle e_1,e_2,e_3\rangle,
\end{align*}
where $\varepsilon=\pm1$, $\kappa\geqslant0$ and $\gamma\in\mathbb R$.
\end{theorem}

\begin{theorem}\label{thm:SubalgebrasOfAffineAlg}
A complete list of ${\rm Inn}\big(\mathfrak{aff}_2(\mathbb R)\big)$-inequivalent proper subalgebras of the real rank-two affine Lie algebra~$\mathfrak{aff}_2(\mathbb R)$
is given by
\begin{align*}
1{\rm D}\colon\ &
\mathfrak s_{1.1}            =\langle f_1\rangle,\ \
\mathfrak s_{1.2}^{\delta}  =\langle e_1+\delta f_1\rangle,\ \
\mathfrak s_{1.3}            =\langle e_4\rangle,\ \
\mathfrak s_{1.4}^\kappa     =\langle e_2+\kappa e_4\rangle,\\
\hphantom{1{\rm D}\colon\ }&
\mathfrak s_{1.5}            =\langle e_2+e_4+f_2\rangle,\ \
\mathfrak s_{1.6}^\varepsilon=\langle e_1+\varepsilon e_4\rangle,\ \
\mathfrak s_{1.7}^\gamma     =\langle e_1+e_3+\gamma e_4\rangle,\ \
\\[1ex]
2{\rm D}\colon\ &
\mathfrak s_{2.1}          =\langle f_1,f_2\rangle,\ \
\mathfrak s_{2.2}^{\delta}=\langle e_1+\delta f_1,f_2\rangle,\ \
\mathfrak s_{2.3}          =\langle e_4,f_1\rangle,\ \
\mathfrak s_{2.4}^\kappa   =\langle e_2+\kappa e_4,f_1\rangle,\ \
\\
\hphantom{2{\rm D}\colon\ }&
\mathfrak s_{2.5}^{\kappa'}=\langle e_2+\kappa' e_4,f_2\rangle,\ \
\mathfrak s_{2.6}          =\langle e_2+e_4+f_2,f_1\rangle,\ \
\mathfrak s_{2.7}^\varepsilon=\langle e_1+\varepsilon e_4,f_2\rangle,\ \
\\
\hphantom{2{\rm D}\colon\ }&
\mathfrak s_{2.8}            =\langle e_1,e_4\rangle,\ \
\mathfrak s_{2.9}            =\langle e_2,e_4\rangle,\ \
\mathfrak s_{2.10}           =\langle e_1+e_3,e_4\rangle,\ \
\mathfrak s_{2.11}^\gamma    =\langle e_2+\gamma e_4,e_1\rangle,\ \
\\
\hphantom{2{\rm D}\colon\ }&
\mathfrak s_{2.12}          =\langle e_2+ e_4+f_2,e_1\rangle,\ \
\mathfrak s_{2.13}            =\langle e_2-3e_4,e_1+f_1\rangle,
\\[1ex]
3{\rm D}\colon\ &
\mathfrak s_{3.1}            =\langle e_1,f_1,f_2\rangle,\ \
\mathfrak s_{3.2}            =\langle e_4,f_1,f_2\rangle,\ \
\mathfrak s_{3.3}^\kappa     =\langle e_2+\kappa e_4,f_1,f_2\rangle,\ \
\\
\hphantom{3{\rm D}\colon\ }&
\mathfrak s_{3.4}^\varepsilon=\langle e_1+\varepsilon e_4,f_1,f_2\rangle,\ \
\mathfrak s_{3.5}^\gamma     =\langle e_1+e_3+\gamma e_4,f_1,f_2\rangle,\ \
\mathfrak s_{3.6}            =\langle e_1,e_4,f_2\rangle,\ \
\\
\hphantom{3{\rm D}\colon\ }&
\mathfrak s_{3.7}            =\langle e_2,e_4,f_1\rangle,\ \
\mathfrak s_{3.8}^\gamma     =\langle e_2+\gamma e_4,e_1,f_2\rangle,\ \
\mathfrak s_{3.9}         =\langle e_2-e_4+f_1, e_1, f_2\rangle,
\\
\hphantom{3{\rm D}\colon\ }&
\mathfrak s_{3.10}            =\langle e_2-3e_4, e_1+f_1, f_2\rangle,\ \
\mathfrak s_{3.11}           =\langle e_1,e_2,e_4\rangle,\ \
\mathfrak s_{3.12}           =\langle e_1,e_2,e_3\rangle,
\\[1ex]
4{\rm D}\colon\ &
\mathfrak s_{4.1}        =\langle e_1,e_4,f_1,f_2\rangle,\ \
\mathfrak s_{4.2}        =\langle e_2,e_4,f_1,f_2\rangle,\ \
\mathfrak s_{4.3}        =\langle e_1+e_3,e_4,f_1,f_2\rangle,
\\
\hphantom{4{\rm D}\colon\ }&
\mathfrak s_{4.4}^\gamma =\langle e_2+\gamma e_4,e_1,f_1,f_2\rangle,\ \
\mathfrak s_{4.5}        =\langle e_1,e_2,e_4,f_2\rangle,\ \
\mathfrak s_{4.6}        =\langle e_1,e_2,e_3,e_4\rangle,
\\[1ex]
5{\rm D}\colon\ &
\mathfrak s_{5.1}=\langle e_1,e_2,e_4,f_1,f_2\rangle,\ \
\mathfrak s_{5.2}=\langle e_1,e_2,e_3,f_1,f_2\rangle,
\end{align*}
where $\varepsilon\in\{-1,1\}$, $\delta\in\{0,1\}$, $\kappa\geqslant0$, $\kappa'>0$ and $\gamma\in\mathbb R$.
\end{theorem}
\begin{proof}
By using Theorem~\ref{thm:SubalgebrasOfGL2}, we perform the first step of the algorithm from Section~\ref{subsec:semidirect}.
Now we should execute steps 2, 3 and 4 for the families of subalgebras from the list given in Theorem~\ref{thm:SubalgebrasOfGL2},
including the trivial subalgebra $\{0\}$.
The computations below are tedious and cumbersome,
this is why for the optimal performance and better presentation,
we structure the proof as follows.
For each subalgebra $\h_{i.j}^*$ from Theorem~\ref{thm:SubalgebrasOfGL2}
and the trivial subalgebra $\h_0=\{0\}$, we begin with constructing
the normalizer subgroup ${\rm Nor}\big(\h_{i.j}^*,{\rm GL}^+_2(\mathbb R)\big)$
where it is needed,
then classify $\h_{i.j}^*$-submodules of $\mathbb R^2$.
It is clear that there is only one zero-dimensional and one two-dimensional submodule,
i.e.,~$\{0\}$ and $\mathbb R^2=\langle f_1,f_2\rangle$.
Here $i$, $j$ and $*$ in notation of the subalgebra $\h_{i.j}^*$
stand for the dimension of the subalgebra, its number in the list of $i$-dimensional subalgebras of the algebra $\mathfrak{gl}_2(\mathbb R)$
and the list of subalgebra parameters for a subalgebra family, respectively.
We do not introduce separate notations for submodules but rather denote them as vector spaces spanned by some specific basis or as $U\subset\mathbb R^2$ if necessary.
We also do not use separate notations for cocycles for the sake of text readability,
thus for a fixed $\h_{i.j}^*$-submodule $U\subset\mathbb R^2$ we start from the most general form of a one-cocycle
$\alpha\colon\h_{i.j}^*\to\mathbb R^2/U$ and simplify it using elements of ${\rm Nor}\big(\h_{i.j}^*,{\rm GL}^+_2(\mathbb R)\big)\ltimes\mathbb R^2$.
For the obtained inequivalent cocycles we construct the subalgebras following Proposition~\ref{prop:SubalgebrasSemidirect}.
Note that if $\dim\h_{i.j}^*=1$, then any linear map $\alpha\colon\h_{i.j}^*\to\mathbb R^2/U$ is a~one-cocycle.
Otherwise, one should check that the cocycle condition is satisfied,
which in case of the affine algebra $\mathfrak{aff}_2(\mathbb R)=\mathfrak{gl}_2(\mathbb R)\ltimes\mathbb R^2$ takes the form
\[
\alpha([x,y])=x\alpha(y)-y\alpha(x)\quad\forall x,y\in\h_{i.j}^*.
\]

Another important remark is that any cocycle $\alpha\colon\h_{i.j}^*\to\mathbb R^2/\mathbb R^2$ is the trivial cocycle,
thus the subalgebra corresponding to this cocycle is $\h_{i.j}^*\ltimes\langle f_1,f_2\rangle$.
We will not mention this trivial case throughout the proof,
however, we indicate the subalgebra obtained in this way $\h_{i.j}^*\ltimes\langle f_1,f_2\rangle$
as the last element of the list of subalgebras constructed for each $\h_{i.j}^*$.

The following simple observation is useful in the course of cocycle simplification.

\begin{lemma}\label{lem:FullRank}
Let $x\in{\rm Mat}_2(\mathbb R)$ such that $\mathop{\rm rank}x=2$.
Then any cocycle $\alpha\colon\langle x\rangle\to\mathbb R^2$ is equivalent to the trivial one
under the action of the inner automorphism group of $\mathfrak{aff}_2(\mathbb R)$.
\end{lemma}
\begin{proof}
Conjugate any cocycle $\alpha\colon x\mapsto c_1f_1+c_2f_2$ by $\big({\rm diag}(1,1),x^{-1}(c_1f_1+c_2f_2)\big)$.
\end{proof}

\medskip\par\noindent
$\boldsymbol{\h_0=\langle0\rangle}$.
It is clear that the normalizer subgroup of $\h_0$ is the entire group ${\rm GL}^+_2(\mathbb R)$.
Therefore, the inequivalent $\h_0$-submodules of $\mathbb R^2$ are characterized by their dimensions,
so we can choose $\{0\}$, $\langle f_1\rangle$ and $\langle f_1,f_2\rangle$.
As a result, we obtain the subalgebras $\langle 0\rangle$, $\mathfrak s_{1.1}$ and $\mathfrak s_{2.1}$.

\medskip\par\noindent
$\boldsymbol{\h_{1.1}=\langle e_1\rangle}$.
The normalizer of the subalgebra $\h_{1.1}$ in the group ${\rm GL}^+_2(\mathbb R)$
is generated by the lower-triangular matrices, i.e., $(a_{lm})_{l,m=1,2}$ with $a_{12}=0$,
$a_{11}a_{22}>0$ and $a_{21}\in\mathbb R$.
The next step is to choose ${\rm Nor}\big(\h_{1.1},{\rm GL}^+_2(\mathbb R)\big)$-inequivalent $\h_{1.1}$-submodules $U$ of $\mathbb R^2$,
and the only nontrivial step is to find one-dimensional modules.
If $\dim U=1$, then consider $U$ spanned by $c_1f_1+c_2f_2$.
Since $e_1(c_1f_1+c_2f_2)=c_1f_2$, we have that $c_1=0$,
yielding $U=\langle f_2\rangle$.
Thus, the submodules are $\{0\}$, $\langle f_2\rangle$ and $\langle f_1,f_2\rangle$.

The next step is to consider one-cocycles $\alpha\colon\h_{1.1}\to\mathbb R^2/U$ and reduce them to the simplest form modulo the action
of the group ${\rm Nor}\big(\h_{1.1},{\rm GL}^+_2(\mathbb R)\big)\ltimes\mathbb R^2$.
To construct one-cocycles, we note that
any linear map between Lie algebras whose domain is one-dimensional is a one-cocycle.
Starting the consideration from the $\h_{1.1}$-submodule $\{0\}$,
consider a linear map $\alpha\colon e_1\mapsto c_1f_1+c_2f_2$.
Acting on $\alpha$ by $\Ad\big({\rm diag}(1,1),c_2f_1\big)$ via formula from  step 3 in algorithm presented in Section~\ref{subsec:semidirect},
we change the general form of the one-cocycle $\alpha$ to $\alpha\colon e_1\mapsto c_1f_1$.
If $c_1\ne0$, we simplify it further acting by $\Ad\big({\rm diag}(\mathop{\rm sgn}c_1,c_1),0\big)$
and obtain $\alpha\colon e_1\mapsto f_1$.
Otherwise, $\alpha\colon e_1\mapsto 0$.

By this, we obtained the candidates for the representatives of the equivalence classes of cocycles,
$\hat\alpha^1\colon e_1\mapsto f_1$ and $\hat\alpha^0\colon e_1\mapsto 0$.
If $\hat\alpha^0$ is equivalent to $\hat\alpha^1$,
then there exists $(g,v)\in{\rm Nor}\big(\langle e_1\rangle,{\rm GL}^+_2(\mathbb R)\big)\ltimes\mathbb R^2$
such that
\begin{gather}\label{eq:InEqCocyclesP1}
	\Ad(g,v)\hat\alpha^0=\hat\alpha^1.
\end{gather}
The left-hand side of~\eqref{eq:InEqCocyclesP1} can be written as
\[
\Ad(g,v)\hat\alpha^0=g\circ\hat{\alpha}^0\circ\Ad g^{-1}+\ad v\vert_{\h_{1.1}}.
\]
Since $\hat{\alpha}^0(e_1)=0$, the term $g\circ\hat{\alpha}^0\circ\Ad g^{-1}$ vanishes.
We have that $\ad v\vert_{\h_{1.1}}(e_1)\in\langle f_2\rangle$,
which implies that $\Ad(g,v)\hat\alpha^0(e_1)\in\langle f_2\rangle$.
This contradicts the equality~\eqref{eq:InEqCocyclesP1}
since $\hat\alpha^1(e_1)=f_1$.
Hence, $\hat\alpha^0$ and $\hat\alpha^1$ are not equivalent.
The subalgebras that correspond to $\hat\alpha^0$ and $\hat\alpha^1$
are $\mathfrak s_{1.2}^0$ and $\mathfrak s_{1.2}^1$.

The next submodule to consider is $\langle f_2\rangle$.
Consider the map $\alpha\colon e_1\mapsto c_1f_1+\langle f_2\rangle\in\mathbb R^2/\langle f_2\rangle$.
If $c_1\ne0$, then, acting on $\alpha$ by ${\rm diag}(\mathop{\rm sgn}c_1,c_1)$, we can set $c_1=1$.
Otherwise, $\alpha\colon e_1\mapsto \langle f_2\rangle\in\mathbb R^2/\langle f_2\rangle$.
Hence, the only admissible one-cocycles are $\alpha^1\colon e_1\mapsto f_1+\langle f_2\rangle$ and $\alpha^0\colon e_1\mapsto\langle f_2\rangle$
and they are inequivalent.
The corresponding subalgebras are $\mathfrak s_{2.2}^1$ and $\mathfrak s_{2.2}^0$.

The last subalgebra for this case, which is obtained after considering two-dimensional submodule $\mathbb R^2$, is $\mathfrak s_{3.1}$.

\medskip\par\noindent
$\boldsymbol{\h_{1.2}=\langle e_4\rangle}$.
It is clear that ${\rm Nor}\big(\h_{1.2},{\rm GL}^+_2(\mathbb R)\big)={\rm GL}^+_2(\mathbb R)$,
thus the inequivalent submodules of $\mathbb R^2$ are $\{0\}$, $\langle f_1\rangle$ and $\langle f_1,f_2\rangle$.
Since $\mathop{\rm rank}e_4=2$, Lemma~\ref{lem:FullRank} implies that
any one-cocycle $\alpha\colon\h_{1.2}\to\mathbb R^2$ can be reduced to the trivial one.
Therefore, this case results in the subalgebras $\mathfrak s_{1.3}$, $\mathfrak s_{2.3}$ and~$\mathfrak s_{3.2}$.

\medskip\par\noindent
$\boldsymbol{\h_{1.3}^\kappa=\langle e_2+\kappa e_4\rangle\textbf{ with }\kappa\geqslant0.}$
This case is by far the most complicated.
When $\kappa=0$, the normalizer ${\rm Nor}\big(\langle e_2\rangle,{\rm GL}^+_2(\mathbb R)\big)$
is constituted by the matrices
\begin{gather*}
	\begin{pmatrix}
		a_{11} & 0\\
		0      & a_{22}
	\end{pmatrix}
	\quad
	\mbox{and}
	\quad
	\begin{pmatrix}
		0      & a_{12}\\
		a_{21} & 0
	\end{pmatrix}
\end{gather*}
with $a_{11}a_{22}>0$ and $a_{12}a_{21}<0$.
Otherwise, the normalizer ${\rm Nor}\big(\langle e_2+\kappa e_4\rangle,{\rm GL}^+_2(\mathbb R)\big)$, where $\kappa>0$, is constituted
by the matrices ${\rm diag}(a_{11},a_{22})$ satisfying the condition $a_{11}a_{22}>0$.
This is why we split the consideration into three cases:
$\kappa\in\mathbb R_{\geqslant0}\setminus\{0,1\}$,
$\kappa=0$ and $\kappa=1$.
The first case is the regular one,
the case $\kappa=0$ is caused by the structure of the normalizer subgroup,
and the case with $\kappa=1$ requires separate consideration in view of $\mathop{\rm rank}(e_1+e_4)=1$.

\medskip\par\noindent
$\boldsymbol{\kappa=0}$.
To classify inequivalent one-dimensional $\langle e_2\rangle$-submodules of $\mathbb R^2$,
we use the fact that $e_2(c_1f_1+c_2f_2)=\frac12c_1f_1-\frac12c_2f_2$, yielding either $c_1=0$ or $c_2=0$.
Hence, the only one-dimensional submodules are $\langle f_1\rangle$ and $\langle f_2\rangle$,
however, they are equivalent:
consider the action by the matrix $A\in{\rm Nor}\big(\langle{\rm diag}(1,-1)\rangle,{\rm GL}_2^+(\mathbb R)\big)$
with entries $a_{11}=a_{22}=0$ and $a_{12}=-a_{21}=1$.
This means that the representatives of the $\langle e_2\rangle$-submodules of $\mathbb R^2$ are
$\{0\}$, $\langle f_1\rangle$ and $\mathbb R^2$.

In view of Lemma~\ref{lem:FullRank}, all one-cocycles can be reduced to the trivial one.
Hence, this subcase results in the subalgebras
$\mathfrak s_{1.4}^0=\langle e_2\rangle$,
$\mathfrak s_{2.4}^0=\langle e_2,f_1\rangle$ and
$\mathfrak s_{3.3}^0=\langle e_2,f_1,f_2\rangle$.

\medskip\par\noindent
$\boldsymbol{\kappa=1}$.
Since $(e_2+e_4)(c_1f_1+c_2f_2)=c_1f_1$,
we have that either $c_1=0$ or $c_2=0$.
This is why the inequivalent $\langle e_2+e_4\rangle$-submodules of $\mathbb R^2$ are
$\{0\}$,
$\langle f_1\rangle$,
$\langle f_2\rangle$
and
$\mathbb R^2$.

Consider a $1$-cocycle $\alpha\colon e_2+e_4\mapsto c_1f_1+c_2f_2\in\mathbb R^2$.
Since $\mathop{\rm im} (e_2+e_4)=\langle f_1\rangle$, we can reduce the cocycle $\alpha$ to $\alpha\colon e_2+e_4\mapsto c_2f_2$.
If $c_2=0$, then the $1$-cocycle is trivial, $\hat\alpha^0\colon e_2+e_4\mapsto 0$.
If $c_2$ is nonzero, applying ${\rm Ad}\big({\rm diag}(\mathop{\rm sgn} c_2,c_2^{-1}),0\big)$ to $\alpha\colon e_2+e_4\mapsto c_2f_2$
leads to $\hat\alpha^1\colon e_2+e_4\mapsto f_2$.
The resulting two cocycles are inequivalent.
Hence, we obtained $\hat\alpha^0\colon e_2+e_4\mapsto 0$ and $\hat\alpha^1\colon e_2+e_4\mapsto f_2$.
The subalgebras corresponding to these cocycles are $\mathfrak s_{1.4}^1=\langle e_2+e_4\rangle$ and $\mathfrak s_{1.5}=\langle e_2+e_4+f_2\rangle$.

The next step is to consider the submodule $\langle f_1\rangle$
and a one-cocycle $\alpha\colon e_2+e_4\mapsto c_2f_2+\langle f_1\rangle\in\mathbb R^2/\langle f_1\rangle$.
Analogously to the previous case, if $c_2=0$, then the one-cocycle is trivial,
$\hat\alpha^0\colon e_2+e_4\mapsto\bar0\in\mathbb R^2/\langle f_1\rangle$.
Otherwise, applying ${\rm Ad}\big({\rm diag}(\mathop{\rm sgn} c_2,c_2^{-1}),0\big)$, we gauge $c_2$ to~$1$.
Hence, the one-cocycle is equivalent to
$\hat\alpha^1\colon e_2+e_4\mapsto f_2+\langle f_1\rangle\in\mathbb R^2/\langle f_1\rangle$.
The resulting two cocycles are inequivalent.
Thus the corresponding subalgebras are $\mathfrak s_{2.4}^1=\langle e_2+e_4,f_1\rangle$ and $\mathfrak s_{2.6}=\langle e_2+e_4+f_2,f_1\rangle$.

The case with the submodule $\langle f_2\rangle$ is trivial.
Due to the fact that $\mathop{\rm im}(e_2+e_4)=\langle f_1\rangle$, we can always reduce
any cocycle $\alpha\colon\h_{1.3}^\kappa\to\mathbb R^2/\langle f_2\rangle$ to
the cocycle $\alpha\colon e_2+e_4\mapsto\langle f_2\rangle\in\mathbb R^2/\langle f_2\rangle$,
which is trivial.
Hence, the subalgebra corresponding to this one-cocycle is $\mathfrak s_{2.5}^1=\langle e_2+e_4,f_2\rangle$.
And for the last case of a two-dimensional submodule, the subalgebra is $\mathfrak s_{3.3}^1=\langle e_2+e_4,f_1,f_2\rangle$.

\medskip\par\noindent
$\boldsymbol{\kappa\in\mathbb R_{>0}\setminus\{1\}}$.
Since $\mathop{\rm rank}(e_2+\kappa e_4)=2$, each $1$-cocycle can be reduced to the trivial one.
To classify $\h_{1.3}^\kappa$-submodules of $\mathbb R^2$,
we use that $(e_2+\kappa e_4)(c_1f_1+c_2f_2)=\frac{c_1}2(\kappa+1)f_1+\frac{c_2}2(\kappa-1)f_2$,
which straightforwardly leads to the inequivalent submodules
$\{0\}$,
$\langle f_1\rangle$,
$\langle f_2\rangle$
and $\mathbb R^2$.
Hence, this subcase yields
$\mathfrak s_{1.4}^\kappa=\langle e_2+\kappa e_4\rangle$,
$\mathfrak s_{2.4}^\kappa=\langle e_2+\kappa e_4,f_1\rangle$,
$\mathfrak s_{2.5}^\kappa=\langle e_2+\kappa e_4,f_2\rangle$,
$\mathfrak s_{3.3}^\kappa=\langle e_2+\kappa e_4,f_1,f_2\rangle$.

The entire case of $\h_{1.3}^\kappa$ results in the subalgebras
$\mathfrak s_{1.4}^\kappa$,
$\mathfrak s_{1.5}$,
$\mathfrak s_{2.4}^\kappa$,
$\mathfrak s_{2.5}^{\kappa'}$,
$\mathfrak s_{2.6}$
and $\mathfrak s_{3.3}^\kappa$.

\medskip\par\noindent
$\boldsymbol{\h_{1.4}^\varepsilon=\langle e_1+\varepsilon e_4\rangle\textbf{ with }\varepsilon=\pm1.}$
Since $(e_1+\varepsilon e_4)(c_1f_1+c_2f_2)=\frac\varepsilon2 c_1f_1+(c_1+\frac\varepsilon2 c_2)f_2$,
we have that $c_1=0$.
Hence, the $\h_{1.4}^\varepsilon$-submodules of $\mathbb R^2$ are
$\{0\}$, $\langle f_2\rangle$ and $\langle f_1,f_2\rangle$,
which are definitely inequivalent.
The matrix $e_1+\varepsilon e_4$ has the full rank,
thus by Lemma~\ref{lem:FullRank} any cocycle is equivalent to the trivial one.
Therefore, we obtain the families of subalgebras $\mathfrak s_{1.6}^\varepsilon$,
$\mathfrak s_{2.7}^\varepsilon$ and $\mathfrak s_{3.4}^\varepsilon$ from the statement of Theorem~\ref{thm:SubalgebrasOfAffineAlg}.

\smallskip\par\noindent
$\boldsymbol{\h_{1.5}^\gamma=\langle e_1+e_3+\gamma e_4\rangle\textbf{ with }\gamma\in\mathbb R.}$
The only $\h_{1.5}^\gamma$-submodules of $\mathbb R^2$ are
$\{0\}$ and $\mathbb R^2$.
Using the fact that $\mathop{\rm rank}(e_1+e_3+\gamma e_4)=2$, we can reduce any cocycle to the trivial one using Lemma~\ref{lem:FullRank}.
Thus, we obtain the subalgebras~$\mathfrak s_{1.7}^\gamma$ and~$\mathfrak s_{3.5}^\gamma$.

\smallskip\par\noindent
$\boldsymbol{\h_{2.1}=\langle e_1, e_4\rangle.}$
The only $\h_{2.1}$-submodules of $\mathbb R^2$ are
$\{0\}$, $\langle f_2\rangle$ and $\mathbb R^2$.

For the submodule $\{0\}$, consider a linear map $\alpha\colon\h_{2.1}\to\mathbb R^2$
defined by
$\alpha\colon
e_1\mapsto c_{11}f_1+c_{12}f_2,
e_4\mapsto c_{21}f_1+c_{22}f_2$.
Acting on $\alpha$ by ${\rm Ad}\big({\rm diag}(1,1),2c_{21}f_1+2c_{22}f_2\big)$, we obtain
$\alpha\colon
e_1\mapsto c_{11}f_1+c_{12}f_2,
e_4\mapsto0$.
This map is a cocycle if and only if $c_{11}=c_{12}=0$,
\begin{gather*}
\alpha([e_1,e_4])=0,\quad
e_1\alpha(e_4)-e_4\alpha(e_1)=-\tfrac12c_{11}f_1-\tfrac12c_{12}f_2.
\end{gather*}
In other words, the only possible cocycle is the trivial one.

The consideration for the submodule $\langle f_2\rangle$ is analogous and it also yields the trivial one-cocycle.

As a result, in this case we obtain the subalgebras $\mathfrak s_{2.8}$, $\mathfrak s_{3.6}$ and $\mathfrak s_{4.1}$.

\medskip\par\noindent
$\boldsymbol{\h_{2.2}=\langle e_2, e_4\rangle.}$
The normalizer ${\rm Nor}\big(\h_{2.2},{\rm GL}^+_2(\mathbb R)\big)$ is generated by the diagonal and antidiagonal matrices (as in the case $\h_{1.3}^\kappa$).
Therefore, the inequivalent $\langle e_2,e_4\rangle$-submodules of $\mathbb R^2$
are $\{0\}$, $\langle f_1\rangle$ and $\mathbb R^2$.

For the case of the submodule $\{0\}$, consider the general form of a linear map from $\h_{2.2}$ to~$\mathbb R^2$,
$\alpha\colon
e_2\mapsto c_{11}f_1+c_{12}f_2,
e_4\mapsto c_{21}f_1+c_{22}f_2$.
Since $\mathop{\rm rank}e_2=2$, acting on the map~$\alpha$ by ${\rm Ad}\big({\rm diag}(1,1),e_2^{-1}(c_{11}f_1+c_{12}f_2)\big)$,
we reduce it to $\alpha\colon
e_2\mapsto 0,
e_4\mapsto c_{21}f_1+c_{22}f_2$.
This map is a cocycle if and only if $c_{21}=c_{22}=0$,
\begin{gather*}
\alpha([e_2,e_4])=0,\quad
e_2\alpha(e_4)-e_4\alpha(e_2)=\tfrac12c_{21}f_1-\tfrac12c_{22}f_2.
\end{gather*}

The consideration for the submodule $\langle f_1\rangle$ is analogous and it also results in the trivial one-cocycle.
Hence, we obtain the subalgebras $\mathfrak s_{2.9}$, $\mathfrak s_{3.7}$ and $\mathfrak s_{4.2}$.

\medskip\par\noindent
$\boldsymbol{\h_{2.3}=\langle e_1+e_3, e_4\rangle.}$
The matrix $e_1+e_3$ is the real two-by-two Jordan block with purely imaginary eigenvalues,
hence it  has no proper invariant subspaces.
Thus, the only inequivalent $\h_{2.3}$-submodules of $\mathbb R^2$
are $\{0\}$ and $\mathbb R^2$.

Since $\mathop{\rm rank}e_4=2$, any linear map $\alpha\colon\h_{2.3}\to\mathbb R^2$
is equivalent to a map of the form
$\alpha\colon
e_1+e_3\mapsto c_{11}f_1+c_{12}f_2,
e_4\mapsto 0$.
And the latter is a cocycle if and only if it is trivial,
\begin{gather*}
\alpha([e_1+e_3,e_4])=0,\quad
(e_1+e_3)\alpha(e_4)-e_4\alpha(e_1+e_3)=-\tfrac12c_{11}f_1-\tfrac12c_{12}f_2.
\end{gather*}
By this, we obtain the subalgebras $\mathfrak s_{2.10}$ and $\mathfrak s_{4.3}$.

\medskip\par\noindent
$\boldsymbol{\h_{2.4}^\gamma=\langle e_2+\gamma e_4, e_1\rangle.}$
The normalizer of the subalgebra $\h_{2.4}^\gamma$ is constituted
by the lower-triangular matrices from~${\rm GL}^+_2(\mathbb R)$.
Depending on the value of~$\gamma$, there are three separate cases to consider:
$\gamma=-1$, $\gamma=1$ and $\gamma\in\mathbb R\setminus\{-1,1\}$,
where the first two appear since the matrices $e_2\pm e_4$ are degenerate.

\smallskip\par\noindent
$\boldsymbol{\gamma=1.}$
The $\h_{2.4}^1$-submodules of $\mathbb R^2$ are $\{0\}$, $\langle f_2\rangle$ and $\mathbb R^2$.
In the case of the submodule~$\{0\}$, consider an arbitrary linear map $\alpha\colon\h_{2.4}^1\to\mathbb R^2$,
$\alpha\colon e_2+e_4\mapsto c_{11}f_1+c_{12}f_2,e_1\mapsto c_{21}f_1+c_{22}f_2$.
Acting on $\alpha$ by ${\rm Ad}\big({\rm diag}(1,1),c_{11}f_1\big)$, we can set $c_{11}=0$.
The one-cocycle condition $\alpha([e_2+e_4,e_1])=(e_2+e_4)\alpha(e_1)-e_1\alpha(e_2+e_4)$
gives us that $c_{21}=c_{22}=0$,
\begin{gather}\label{eq:CocycleCondH23^1}
\alpha([e_2+e_4,e_1])=-\alpha(e_1)=-c_{21}f_1-c_{22}f_2,\quad
(e_2+e_4)\alpha(e_1)-e_1\alpha(e_2+e_4)=c_{21}f_1.
\end{gather}
Thus, the potential one-cocycles are of the form $\alpha\colon e_2+e_4\mapsto c_{12}f_2,e_1\mapsto 0$.
If $c_{12}=0$, then $\alpha$ is the trivial cocycle.
If $c_{12}\ne0$, acting on $\alpha$ by ${\rm Ad}\big({\rm diag}(c_{12},c_{12}^{-1}),0\big)$,
we can set $c_{12}=1$.
Therefore, we obtain two inequivalent cocycles $\hat\alpha^0\colon e_2+e_4\mapsto 0,e_1\mapsto 0$
and $\hat\alpha^1\colon e_2+e_4\mapsto f_2,e_1\mapsto 0$.
They correspond to the subalgebras $\mathfrak s_{2.11}^1=\langle e_2+e_4,e_1\rangle$ and
$\mathfrak s_{2.12}=\langle e_2+e_4+f_2,e_1\rangle$.

The consideration for the submodule $\langle f_2\rangle$ is analogous:
factoring~\eqref{eq:CocycleCondH23^1} by $\langle f_2\rangle$, we derive $c_{21}=0$,
which thus leads to the trivial one-cocycle.
This cocycle corresponds to $\mathfrak s^1_{3.8}=\langle e_2+e_4,e_1,f_2\rangle$.

\medskip\par\noindent
$\boldsymbol{\gamma=-1.}$
The $\h_{2.4}^{-1}$-submodules of $\mathbb R^2$ are $\{0\}$, $\langle f_2\rangle$ and $\mathbb R^2$.
In the case of the submodule~$\{0\}$, consider an arbitrary linear map $\alpha\colon\h_{2.4}^{-1}\to\mathbb R^2$,
$\alpha\colon e_2-e_4\mapsto c_{11}f_1+c_{12}f_2,e_1\mapsto c_{21}f_1+c_{22}f_2$.
Acting on $\alpha$ by ${\rm Ad}\big({\rm diag}(1,1),c_{22}f_1\big)$, we can set $c_{22}=0$.
The one-cocycle condition $\alpha([e_2-e_4,e_1])=(e_2-e_4)\alpha(e_1)-e_1\alpha(e_2-e_4)$
gives us that $c_{21}=c_{11}=0$,
\begin{gather}\label{eq:CocycleCondH23^-1}
\alpha([e_2-e_4,e_1])=-\alpha(e_1)=-c_{21}f_1,\quad
(e_2-e_4)\alpha(e_1)-e_1\alpha(e_2-e_4)=-c_{11}f_2.
\end{gather}
Then, acting on this cocycle by ${\rm Ad}\big({\rm diag}(1,1),-c_{12}f_2\big)$,
we set $c_{12}=0$, thus the cocycle is trivial.
This yields the subalgebra $\mathfrak s_{2.11}^{-1}=\langle e_2-e_4,e_1\rangle$.

For the submodule~$\langle f_2\rangle$, the consideration is similar:
factoring~\eqref{eq:CocycleCondH23^-1} by the submodule~$\langle f_2\rangle$,
we can only obtain $c_{21}=0$. 
Therefore, $\alpha\colon e_2-e_4\mapsto c_{11}f_1+\langle f_2\rangle ,e_1\mapsto \langle f_2\rangle$
is a one-cocycle.
If $c_{11}\ne0$, acting on $\alpha$ by ${\rm Ad}\big({\rm diag}(c_{11}^{-1},c_{11}),0\big)$,
we can set $c_{11}=1$.
Otherwise, the cocycle $\alpha$ is trivial.
This case results in the subalgebras $\mathfrak s_{3.8}^{-1}=\langle e_2-e_4,e_1,f_2\rangle$
and $\mathfrak s_{3.9}=\langle e_2-e_4+f_1,e_1,f_2\rangle$ that correspond to the values $c_{11}=0$ and $c_{11}=1$,
respectively.

\medskip\par\noindent
$\boldsymbol{\gamma\in\mathbb R\setminus\{-1,1\}.}$
The inequivalent submodules of $\mathbb R^2$ are $\{0\}$, $\langle f_2\rangle$ and $\mathbb R^2$.
For the submodule~$\{0\}$, consider a linear map $\alpha\colon\h_{2.4}^{\gamma}\to\mathbb R^2$
given by $\alpha\colon e_2+\gamma e_4\mapsto c_{11}f_1+c_{12}f_2, e_1\mapsto c_{21}f_1+c_{22}f_2$.
Since $\mathop{\rm rank}(e_2+\gamma e_4)=2$, we can reduce $\alpha$ to
$\alpha\colon e_2+\gamma e_4\mapsto 0, e_1\mapsto c_{21}f_1+c_{22}f_2$
similarly as it is done in Lemma~\ref{lem:FullRank}.
The one-cocycle condition $\alpha([e_2+\gamma e_4,e_1])=(e_2+\gamma e_4)\alpha(e_1)-e_1\alpha(e_2+\gamma e_4)$,
\begin{gather*}
\alpha([e_2+\gamma e_4,e_1])=-\alpha(e_1)=-c_{21}f_1-c_{22}f_2,
\\
(e_2+\gamma e_4)\alpha(e_1)-e_1\alpha(e_2+\gamma e_4)
=\tfrac12(\gamma+1)c_{21}f_1+\tfrac12(\gamma-1)c_{22}f_2,
\end{gather*}
results in the system of equations $(\gamma+3)c_{21}=0$ and $(\gamma+1)c_{22}=0$.
The second equation straightforwardly gives $c_{22}=0$ since $\gamma\ne-1$.
Therefore, there are two cases to consider: $\gamma=-3$ or $\gamma\ne-3$.
If $\gamma\ne-3$, then $c_{21}=0$ and the one-cocycle is trivial,
which corresponds to the subalgebra $\mathfrak s_{2.11}^{\gamma}=\langle e_2+\gamma e_4,e_1\rangle$.
If $\gamma=-3$ and $c_{21}\ne0$, then applying ${\rm Ad}\big({\rm diag}(c_{21}^{-1},c_{21}),0\big)$
to $\alpha$ results in the one-cocycle $\alpha\colon e_2+\gamma e_4\mapsto 0, e_1\mapsto f_1$.
This cocycle corresponds to the subalgebra $\mathfrak s_{2.13}=\langle e_2-3e_4,e_1+f_1\rangle$.
Otherwise, $c_{21}=0$ yields the trivial one-cocycle
that corresponds to the subalgebra $\mathfrak s_{2.11}^{-3}=\langle e_2-3e_4,e_1\rangle$.

The cases with the submodule $\langle f_2\rangle$ are similar.
They also give us a trivial one-cocycle,
which corresponds to the subalgebra $\mathfrak s_{3.8}^\gamma=\langle e_2+\gamma e_4,e_1,f_2\rangle$,
and in the special case of $\gamma=-3$, we have a cocycle
$\alpha\colon e_2-3e_4\mapsto\langle f_2\rangle, e_1\mapsto f_1+\langle f_2\rangle$.
The corresponding to this special one-cocycle subalgebra is $\mathfrak s_{3.10}=\langle e_2-3e_4, e_1+f_1, f_2\rangle$.

The entire case of $\h_{2.4}^\gamma=\langle e_2+\gamma e_4, e_1\rangle$
gives the subalgebras $\mathfrak s_{2.11}^\gamma$, $\mathfrak s_{2.12}$,
$\mathfrak s_{2.13}$, $\mathfrak s_{3.8}^\gamma$, $\mathfrak s_{3.10}$, $\mathfrak s_{4.4}^\gamma$.

\medskip\par\noindent
$\boldsymbol{\h_{3.1}=\langle e_1,e_2,e_4\rangle.}$
The $\h_{3.1}$-submodules of $\mathbb R^2$ are $\{0\}$, $\langle f_2\rangle$ and $\mathbb R^2$.
Consider a linear map $\alpha\colon\h_{3.1}\to\mathbb R^2$ given by
$
\alpha\colon
(e_1,e_2,e_4)\mapsto(c_{11}f_1+c_{12}f_2,c_{21}f_1+c_{22}f_2,c_{31}f_1+c_{32}f_2).
$
Since $\mathop{\rm rank}e_2=2$ we can set $c_{21}=c_{22}=0$.
The one-cocycle condition
$\alpha([e_1,e_2])=e_1\alpha(e_2)-e_2\alpha(e_1)$
and $\alpha([e_2,e_4])=e_2\alpha(e_4)-e_4\alpha(e_2)$,
\begin{gather*}
\alpha([e_1,e_2])=-c_{11}f_1-c_{12}f_2,
\quad
e_1\alpha(e_2)-e_2\alpha(e_1)
=-\tfrac12 c_{11}f_1+\tfrac12 c_{12}f_2,
\\
\alpha([e_2,e_4])=0,
\quad
e_2\alpha(e_4)-e_4\alpha(e_2)
=\tfrac12 c_{31}f_1-\tfrac12 c_{32}f_2,
\end{gather*}
gives us that $c_{11}=c_{12}=c_{31}=c_{32}=0$.
Thus, $\alpha$ is the trivial cocycle.

The consideration for the rest of the submodules also gives us the trivial cocycles.
Hence, in this case, we obtain the subalgebras
$\mathfrak s_{3.11}$, $\mathfrak s_{4.5}$ and $\mathfrak s_{5.1}$.

\medskip\par\noindent
$\boldsymbol{\h_{3.2}=\langle e_1,e_2,e_3\rangle.}$
The algebra $\h_{3.2}$ is simple, thus by the Whitehead's lemma,
each cocycle is a coboundary, thus each one-cocycle is trivial.
This leads to the subalgebras $\mathfrak s_{3.12}$ and $\mathfrak s_{5.2}$ in the statement of the theorem.

\medskip\par\noindent
$\boldsymbol{\h_{4.1}=\langle e_1,e_2,e_3,e_4\rangle.}$
It is clear that each one-cocycle in this case is trivial,
thus yielding the subalgebra~$\mathfrak s_{4.6}$ and the entire algebra $\mathfrak{aff}_2(\mathbb R)$.
\end{proof}

Now, we use the classification from Theorem~\ref{thm:SubalgebrasOfAffineAlg} to classify
subalgebras of $\mathfrak{aff}_2(\mathbb R)$ modulo the action of the affine group ${\rm GL}_2(\mathbb R)\ltimes\mathbb R^2$.

\begin{corollary}\label{cor:ClassOfSubalgModGL2}
A complete list of ${\rm GL}_2(\mathbb R)\ltimes\mathbb R^2$-inequivalent subalgebras of
the rank-two affine algebra $\mathfrak{aff}_2(\mathbb R)$ is exhausted by the subalgebras presented in the list
of Theorem~\ref{thm:SubalgebrasOfAffineAlg}
subject to the following parameter gauges:
\begin{itemize}\itemsep=0ex
\item
$\varepsilon=1$ for $\mathfrak s_{1.6}^\varepsilon$,
$\mathfrak s_{2.7}^\varepsilon$ and
$\mathfrak s_{3.4}^\varepsilon$;

\item
$\gamma\geqslant0$ for $\mathfrak s_{1.7}^\gamma$ and
$\mathfrak s_{3.5}^\gamma$.
\end{itemize}
\end{corollary}

\begin{proof}
The group ${\rm GL}_2(\mathbb R)\ltimes\mathbb R^2$ is isomorphic to $\mathbb Z_2\ltimes{\rm Inn}(\mathfrak{aff}_2(\mathbb R))$,
where $\mathbb Z_2$ is isomorphic to the subgroup of ${\rm GL}_2(\mathbb R)$ generated by the diagonal matrix ${\rm diag}(1,-1)$.
Conjugating the matrices $e_1+\varepsilon e_4$ and $e_1+e_3+\gamma e_4$ by ${\rm diag}(1,-1)$ allows us to fix the gauges to $\varepsilon=1$ and $\gamma\geqslant0$, respectively.
Conjugating the subalgebra $\mathfrak s_{3.9}$ by ${\rm diag}(1,-1)$ yields the subalgebra $\langle e_2-e_4-f_1, e_1, f_2\rangle$,
which is ${\rm Inn}\big(\mathfrak{aff}_2(\mathbb R)\big)$-equivalent to $\mathfrak s_{3.9}$.

The remaining subalgebras are invariant under this conjugation.
\end{proof}

\section{Subalgebras of $\mathfrak{sl}_3(\mathbb R)$}\label{sec:sl3Classification}

The real order-three special linear Lie algebra $\mathfrak{sl}_3(\mathbb R)$
is the algebra of traceless $3\times3$ matrices with the standard matrix commutator as the Lie bracket
and it is spanned by the matrices
\begin{gather*}
E_1:=\begin{pmatrix}
0&0&0\\
1&0&0\\
0&0&0
\end{pmatrix},\
E_2:=\frac12\begin{pmatrix}
1&0&0\\
0&-1&0\\
0&0&0
\end{pmatrix},\
E_3:=\begin{pmatrix}
0&-1&0\\
0& 0&0\\
0& 0&0
\end{pmatrix},\ 
D:=\frac16\begin{pmatrix}
1&0& 0\\
0&1& 0\\
0&0&-2
\end{pmatrix},
\\
P_1:=\begin{pmatrix}
0&0&1\\
0&0&0\\
0&0&0
\end{pmatrix},\quad
P_2:=\begin{pmatrix}
0&0&0\\
0&0&1\\
0&0&0
\end{pmatrix},\quad
R_1:=\begin{pmatrix}
0&0&0\\
0&0&0\\
0&-1&0
\end{pmatrix},\quad
R_2:=\begin{pmatrix}
0&0&0\\
0&0&0\\
1&0&0
\end{pmatrix}.
\end{gather*}
In this way, the algebra $\mathfrak{sl}_3(\mathbb R)$ is defined through its
faithful irreducible representation of the minimal dimension,
which is exactly the vector space $\mathbb R^3$.
Using this representation, we should find all \textit{irreducibly} and \textit{reducibly embedded} maximal subalgebras of $\mathfrak{sl}_3(\mathbb R)$.

\subsection{Irreducibly embedded subalgebras}\label{subsec:SL3IrredEmbed}

Recall that an irreducibly embedded subalgebra is a reductive Lie algebra, i.e.,
semisimple or the direct product of semisimple Lie algebra with an abelian one (over a field of characteristic zero).
We can show that the only semisimple Lie algebras admitting faithful three-dimensional representation are
the special orthogonal algebras $\mathfrak{so}_3(\mathbb R)$ and $\mathfrak{so}_{2,1}(\mathbb R)$.
For this purpose, consider the following chain of propositions.

\begin{proposition}[see details in \mbox{\cite[p.~56]{bour1975LiePart1}}]\label{prop:reductive_g}
Let $\g \subset \mathfrak{gl}_n(\mathbb F)$ be a linear Lie algebra
that has no invariant subspaces.
Then $\g$ is reductive, $\g=\mathfrak s\oplus\mathfrak z$,
where $\mathfrak s$ is a semisimple Lie algebra and~$\mathfrak{z}$ is an abelian Lie algebra.
\end{proposition}

\begin{lemma}\label{lem:ReducOfRepresent}
Let $\g\subset\mathfrak{gl}_n(\mathbb R)$ be a linear Lie algebra
that has no invariant subspaces.
If the complexification $\g^{\mathbb C}\subset\mathfrak{gl}_n(\mathbb C)$ 
has an invariant subspace $U \subset \mathbb C^n$,
then $\mathbb C^n=U\oplus\overline U$,
where $U$ and $\overline U$ are irreducible $\g^{\mathbb C}$-submodules,
in particular, $n = 2 \dim U$.
\end{lemma}

\begin{proof}
Suppose $\dim_{\mathbb{C}} U<\frac{n}{2}$,
then the linear space $\Re U$ is a $\g$-invariant subspace of $\mathbb R^n$
and $\dim_{\mathbb R}\Re U<n$, which is a contradiction.
Here $\Re$ denotes the real part function extended to the vector space~$U$.
If $\dim_{\mathbb C}U>\frac n2$, then $\dim_{\mathbb R}U>n$ and
Grassmann's identity for vector spaces gives us that $U\cap\mathbb R^n\ne\{0\}$.
The latter subspace is a $\g$-invariant subspace of $\mathbb R^n$,
which is again a contradiction.
Therefore,  $\dim_{\mathbb C}U=\frac n2$.
	
If $U\cap\overline{U}\ne\{0\}$, then it admits a basis
consisting of elements of~$\mathbb R^n$ and hence $U\cap\overline U\cap\mathbb R^n\ne\{0\}$
is a $\g$-invariant subspace of~$\mathbb R^n$, which contradicts the statement of the proposition.
\end{proof}

The proposition below follows immediately from Schur's lemma.

\begin{proposition}\label{prop:CenterComplex}
Let $\g \subset \mathfrak{gl}_n(\mathbb C)$ be a reductive linear Lie algebra,
$\mathbb C^n = V_1 \oplus \dots \oplus V_k$ be the decomposition of $\mathbb C^n$
in the sum of its irreducible submodules.
Then any central element $x \in \mathop Z(\g)$ is of the form
$x = \oplus_{i=1}^k \lambda_i 1_{V_i}$,
where $\lambda_i\in\mathbb C$ and $1_{V_i}$ stands for the identity matrix on $V_i$ for each~$i$.
\end{proposition}

\begin{proposition}[see details in~\mbox{\cite[p.~68]{bour1975LiePart1}}]\label{prop:CenterReal}
Let $\g$ be a real reductive Lie algebra. Then its complexification
$\g^{\mathbb{C}}$ is a complex reductive Lie algebra and
$\mathop Z(\g^{\mathbb{C}}) = \mathop Z(\g)^{\mathbb{C}}$.
\end{proposition}

Lemma~\ref{lem:ReducOfRepresent} and Propositions~\ref{prop:reductive_g}, \ref{prop:CenterComplex} and \ref{prop:CenterReal}
imply that irreducibly embedded Lie subalgebras of~$\mathfrak{sl}_3(\mathbb R)$ are semisimple.
Their complexifications are complex semisimple Lie algebras of dimension less than eight
that have faithful representation on $\mathbb C^3$.
The following proposition shows that irreducibly embedded Lie subalgebras of~$\mathfrak{sl}_3(\mathbb R)$ can only be real forms of 
$\mathfrak{sl}_2(\mathbb C)$.

\begin{proposition}
The Lie algebra $\g = \mathfrak{sl}_2(\mathbb{C}) \oplus \mathfrak{sl}_2(\mathbb{C})$
has no faithful representations on $\mathbb{C}^3$.
\end{proposition}
\begin{proof}
Ad absurdum.
Suppose that $\rho\colon\g\to\mathfrak{gl}_3(\mathbb C)$ is a faithful representation of the Lie algebra~$\g$.
Denoting the direct summands of the Lie algebra $\g=\mathfrak{sl}_2(\mathbb{C}) \oplus \mathfrak{sl}_2(\mathbb{C})$
as $\g_1$ and $\g_2$, respectively,
consider the restriction $\rho\vert_{\g_1}$ of $\rho$ to the first summand.
If $\rho \vert_{\g_1}$ is irreducible, then
$[\rho(\g_1),\rho(\g_2)]=\rho([\g_1,\g_2])=0$
implies $\rho(\g_2)\subset\mathbb C 1_{\mathbb C^3}$ by Schur's lemma.
Since the Lie subalgebra $\g_2$ is simple, it follows that $\rho(\g_2)=0$,
which contradicts the fact that the representation~$\rho$ is faithful.
Therefore, the representation $\rho \vert_{\g_1}$ is reducible and thus is completely reducible.
Moreover, since $\dim\mathbb C^3=3$, the representation $\rho \vert_{\g_1}$
has a one-dimensional invariant subspace $\mathbb C\left< v \right> \subset \mathbb C^3$
and, in particular, $\rho\vert_{\g_1}v=0$. 
	
Consider $\g_1$-submodule $K:=\cap_{x\in\g_1}\ker\rho(x)$ of $\mathbb C^3$.
It is clear that $\mathbb C\langle v\rangle\subset K$.
The dimension of the submodule $K$ cannot be equal to three
since the action of $\g_1$ on $\mathbb C^3$ would be trivial in this case.
If $\dim K=2$, Weyl's theorem gives us that $\g_1$-submodule $K$ has a $\g_1$-invariant complement in $\mathbb C^3$.
The latter implies that the quotient representation of
$\rho\vert_{\g_1}$ on the space $\mathbb C^3/K$ is faithful,
which cannot be possible since $\dim\mathbb C^3/K=1$.
Therefore, $\dim K=1$ and $\mathbb C\langle v\rangle=K=\cap_{x\in\g_1}\ker\rho(x)$.
	
The equality $[\rho(\g_1),\rho(\g_2)] = 0$ implies that $\mathbb C\langle v\rangle$
is an invariant subspace of the representation~$\rho$ and hence
$\mathbb C\langle v\rangle = \cap_{x \in \g} \ker \rho(x)$.
Using the arguments analogous to those in the previous paragraph, 
we can show that the quotient representation
$\bar\rho\colon\g\to\mathfrak{gl}(\mathbb C^3/\mathbb C\langle v\rangle)$
is faithful as well, which is a contradiction, since
$6 = \dim\g\nleq\dim\mathfrak{gl}_2(\mathbb C)=4$.
\end{proof}

\subsection{Reducibly embedded subalgebras}\label{subsec:SL3RedEmbed}

A maximal subalgebra $\mathfrak s\subset\mathfrak{sl}_3(\mathbb R)$ that is reducibly embedded
possesses a proper invariant subspace $V\subset\mathbb R^3$.
Up to a change of basis, we can assume that the invariant subspaces are either $V_1=\langle (1,0,0),(0,1,0)\rangle$ or $V_2=\langle(0,0,1)\rangle$.
Consequently, the corresponding maximal subalgebras are
\begin{gather*}
\mathfrak a_1:=\left\{
\begin{pmatrix}
a_{11}&a_{12}&a_{13}\\
a_{21}&a_{22}&a_{23}\\
     0&     0&-a_{11}-a_{22}
\end{pmatrix}
\Bigm\vert a_{i,j}\in\mathbb R
\right\}
=\langle E_1,E_2,E_3,D,P_1,P_2\rangle,
\\
\mathfrak a_2:=\left\{
\begin{pmatrix}
a_{11}&a_{12}&0\\
a_{21}&a_{22}&0\\
a_{31}&a_{32}&-a_{11}-a_{22}
\end{pmatrix}
\Bigm\vert a_{i,j}\in\mathbb R
\right\}
=\langle E_1,E_2,E_3,D,R_1,R_2\rangle.
\end{gather*}
Both the subalgebras~$\mathfrak a_1$ and~$\mathfrak a_2$ are isomorphic to the rank-two real affine Lie algebra~$\mathfrak{aff}_2(\mathbb R)$,
thus Theorem~\ref{thm:SubalgebrasOfAffineAlg} provides the lists of their inequivalent subalgebras up to the corresponding inner automorphism groups.
However, these lists may contain subalgebras that are conjugate modulo the action of the Lie group ${\rm SL}_3(\mathbb R)$.
Therefore, our next task is to identify such ${\rm SL}_3(\mathbb R)$-equivalent subalgebras and eliminate redundancies. 
To this end, we first determine which subalgebras of $\mathfrak a_1$ are ${\rm SL}_3(\mathbb R)$-inequivalent.
We then supplement the obtained list with the subalgebras of $\mathfrak a_2$ that are both ${\rm SL}_3(\mathbb R)$-inequivalent to each other
and not ${\rm SL}_3(\mathbb R)$-conjugate to any subalgebra of $\mathfrak a_1$.

As Lemma~\ref{lem:SL3Decomp} establishes below, we should use the classification of subalgebras of $\mathfrak{aff}_2(\mathbb R)$ presented in Corollary~\ref{cor:ClassOfSubalgModGL2}. 
The isomorphism between the Lie algebras $\mathfrak{aff}_2(\mathbb R)$ and $\mathfrak a_1$ is given by a linear map~$\rho_1$ defined on the chosen bases of the algebras via the correspondence
\[
(e_1,e_2,e_3,e_4,f_1,f_2)\mapsto(E_1,E_2,E_3,D,P_1,P_2).
\]
Hence, a complete list of inequivalent subalgebras of $\mathfrak a_1$ (with respect to the action of its inner automorphism group)
is constituted by the subalgebras $\hat{\mathfrak s}_{i.j}^*:=\rho_1(\mathfrak s_{i.j}^*)$,
where~$\mathfrak s_{i.j}^*$ are subalgebras listed in Corollary~\ref{cor:ClassOfSubalgModGL2}.
To determine which of the subalgebras $\hat{\mathfrak s}_{i.j}^*$
are equivalent under the action of ${\rm SL}_3(\mathbb R)$, we use the following elementary observation,
which is a variant of the general approach discussed in Section~\ref{subsec:Simple}, see Proposition~\ref{prop:ListMerging} and Remark~\ref{rem:ListMerging}.

\begin{lemma}\label{lem:SL3Decomp}
(i) Any matrix $S$ from ${\rm SL}_3(\mathbb R)$ can be decomposed into a product of matrices $A\in {\rm A}_1\subset{\rm SL}_3(\mathbb R)$ and $M$, $S=AM$,
where ${\rm A}_1$ is isomorphic to the group ${\rm Aff}_2(\mathbb R)$,
\begin{gather*}
{\rm A}_1:=\left\{\begin{pmatrix}
a & b & x\\
c & d & y\\
0 & 0 & (ad-bc)^{-1}
\end{pmatrix}
\Bigm\vert
a,b,c,d,x,y\in\mathbb R~\text{such that}~ad-bc\ne0\right\},
\end{gather*}
and $M$ is of one of the following forms:
\begin{gather*}
M_1(c_1,c_2):=
\begin{pmatrix}
1   & 0   & 0\\
0   & 1   & 0\\
c_1 & c_2 & 1
\end{pmatrix},
\ 
M_2(c_1,c_2):=
\begin{pmatrix}
-1  & 0   & 0\\
 0  & 0   & 1\\
c_1 & 1   & c_2
\end{pmatrix},
\ 
M_3(c_1,c_2):=
\begin{pmatrix}
0  & 0   & -1\\
0  & 1   &  0\\
1  & c_1 & c_2
\end{pmatrix},
\end{gather*}
where the parameters~$c_1$ and~$c_2$ are arbitrary real numbers.
This decomposition corresponds to $\mathfrak m=\mathfrak a_1$ in Proposition~\ref{prop:ListMerging}.

(ii)
There is a similar decomposition for $\mathfrak m=\mathfrak a_2$ in notation of Proposition~\ref{prop:ListMerging}:
any matrix~$S$ from ${\rm SL}_3(\mathbb R)$ decomposes into a product of matrices
$A\in {\rm A}_2\subset{\rm SL}_3(\mathbb R)$ and $N$, $S=AN$,
where ${\rm A}_2$ is isomorphic to the group ${\rm Aff}_2(\mathbb R)$,
\begin{gather*}
{\rm A}_2:=\left\{\begin{pmatrix}
a & b & 0\\
c & d & 0\\
x & y & (ad-bc)^{-1}
\end{pmatrix}
\Bigm\vert
a,b,c,d,x,y\in\mathbb R~\text{such that}~ad-bc\ne0\right\},
\end{gather*}
and~$N$ takes one of the following forms:
\begin{gather*}
N_1(c_1,c_2):=
\begin{pmatrix}
1   & 0   & c_1\\
0   & 1   & c_2\\
0   & 0   & 1
\end{pmatrix},
\ 
N_2(c_1,c_2):=
\begin{pmatrix}
-1  & c_1 & 0\\
 0  & c_2 & 1\\
 0  & 1   & 0
\end{pmatrix},
\ 
N_3(c_1,c_2):=
\begin{pmatrix}
c_1 & 0   & -1\\
c_2 & 1   &  0\\
1   & 0   &  0
\end{pmatrix}.
\end{gather*}
\end{lemma}

Note that both the groups $A_1$ and $A_2$ in the previous lemma are isomorphic to the affine group ${\rm GL}_2(\mathbb R)\ltimes\mathbb R^2$.
This is why, as a starting point, we use the classification of subalgebras of $\mathfrak{aff}_2(\mathbb R)$
presented in Corollary~\ref{cor:ClassOfSubalgModGL2}.

\begin{corollary}\label{cor:IneqSubalgsOfA1}
Among the subalgebras of~$\mathfrak a_1$, the following are inequivalent with respect to the action of~${\rm SL}_3(\mathbb R)$:
\begin{align*}
&1{\rm D}\colon\
\langle E_1+\delta P_1\rangle,\ \
\langle E_1+E_3+\kappa D\rangle,\ \
\langle E_2+\mu'D\rangle,\ \
\langle E_1+D\rangle,
\\[1ex]
&2{\rm D}\colon\
\langle P_1,P_2\rangle,\ \
\langle E_1+\delta P_1,P_2\rangle,\ \
\langle E_2+D+P_2,P_1\rangle,\ \
\langle E_1+D,P_2\rangle,
\\
&\hphantom{2{\rm D}\colon\ }
\langle E_1,D\rangle,\ \
\langle E_2,D\rangle,\ \
\langle E_1+E_3,D\rangle,\ \
\langle E_2+\gamma D,E_1\rangle,\ \
\langle E_2-3D,E_1+P_1\rangle,
\\[1ex]
&3{\rm D}\colon\
\langle E_1,P_1,P_2\rangle,\ \
\langle D,P_1,P_2\rangle,\ \
\langle E_2+\kappa D,P_1,P_2\rangle,\ \
\langle E_1+D,P_1,P_2\rangle,
\\
&\hphantom{3{\rm D}\colon\ }
\langle E_1+E_3+\kappa D,P_1,P_2\rangle,\ \
\langle E_2+\mu D,E_1,P_2\rangle,\ \
\langle E_2-D+P_1,E_1,P_2\rangle,
\\
&\hphantom{3{\rm D}\colon\ }
\langle E_2-3D, E_1+P_1,P_2\rangle,\ \
\langle E_1,E_2,D\rangle,\ \
\langle E_1,E_2,E_3\rangle,
\\[1ex]
&4{\rm D}\colon\ 
\langle E_1,D,P_1,P_2\rangle,\ \
\langle E_2,D,P_1,P_2\rangle,\ \
\langle E_1\!+\!E_3,D,P_1,P_2\rangle,\ \ 
\langle E_2\!+\!\gamma D,E_1,P_1,P_2\rangle,
\\
&\hphantom{4{\rm D}\colon\ }
\langle E_1,E_2,D,P_2\rangle,\ \
\langle E_1,E_2,E_3,D\rangle,
\\[1ex]
&5{\rm D}\colon\
\langle E_1,E_2,D,P_1,P_2\rangle,\ \
\langle E_1,E_2,E_3,P_1,P_2\rangle,
\end{align*}
where $\delta\in\{0,1\}$,
$\kappa\geqslant0$,
$\mu\in[-1,3]$,
$\mu'\in[0,1]$
and $\gamma\in\mathbb R$.
\end{corollary}

\begin{proof}
The proof of this corollary is the most computationally complicated part of this research.
This is why for convenient cross-verification of the results we prepare the {\sf Sage} notebook following
all the computations throughout the proof.
%This notebook can be found following the link:
%\href{https://github.com/zchapovsky/sage-notebooks/blob/main/notebooks/sl3_subalgebras_verification.ipynb}{sl3\_subalgebras\_verification}.
For more details, see section Code availability.%\hyperlink{Code}{Code availability}

To find ${\rm SL}_3(\mathbb R)$-inequivalent subalgebras among $\hat{\mathfrak s}_{i.j}^*$,
Lemma~\ref{lem:SL3Decomp} implies that it is sufficient to consider the action of the set $M:=\{M_k(c_1,c_2)\mid k\in\{1,2,3\}\text{ and }c_1,c_2\in\mathbb R\}$ on this list.
Naturally, the action by $M$ does not preserve the set of subalgebras~$\{\hat{\mathfrak s}_{i.j}^*\}$,
nor does it preserve the set of subalgebras of $\mathfrak a_1$.
In other words, the set
$M(\h):=\big\{M_k(c_1,c_2)\h M_k(c_1,c_2)^{-1}\mid k\in\{1,2,3\}\text{ and } c_1,c_2\in\mathbb R\big\}$
may contain the subalgebras of $\mathfrak{sl}_3(\mathbb R)$
that are not subalgebras of~$\mathfrak a_1$.
In fact, as the computations below show, the subset of $M(\hat{\mathfrak s}_{i.j}^*)$
whose elements are subalgebras of~$\mathfrak a_1$ is no more than finite and of small cardinality for each subalgebra~$\hat{\mathfrak s}_{i.j}^*$.
Therefore, to find the canonical representative of the equivalence class of the subalgebra~$\h$,
we select from $M(\h)$ those subalgebras that are also subalgebras of $\mathfrak a_1$
(excluding $\h$ itself) and identify their canonical representatives from the list of subalgebras $\hat{\mathfrak s}_{i.j}^*$.

Throughout the proof, the tilde $\sim$ denotes the ${\rm SL}_3(\mathbb R)$-equivalence.
Matrix $S$ is useful in the course of proving the corollary's statement,
\begin{gather*}
S:=\begin{pmatrix}
 0& -1& 0\\
-1&  0& 0\\
 0&  0&-1
\end{pmatrix}.
\end{gather*}

\smallskip\par\noindent
$\boldsymbol{1\rm D.}$
Among the one-dimensional subalgebras of $\mathfrak a_1$,
the subalgebras $\hat{\mathfrak s}_{1.2}^{\delta}$,
$\hat{\mathfrak s}_{1.4}^\kappa$,
$\hat{\mathfrak s}_{1.7}^\kappa$ with $\kappa\geqslant0$
and~$\hat{\mathfrak s}_{1.6}^1$
are definitely inequivalent because their generators have distinct (real) Jordan normal forms.
Moreover, up to gauging parameter $\kappa$ in $\hat{\mathfrak s}_{1.7}^\kappa$ to $\kappa\in\mathbb R$,
the chosen generators of the subalgebras exhaust all the possible real Jordan forms of three-by-three matrices.
Therefore, this list inherently contains within itself the list of ${\rm SL}_3(\mathbb R)$-inequivalent subalgebras.
Consequently, any additional conjugation can only occur within the families parameterized by real parameters,
i.e., $\hat{\mathfrak s}_{1.4}^\kappa$ and $\hat{\mathfrak s}_{1.7}^\kappa$ with $\kappa\geqslant0$.
This is why these two cases deserve further investigation.

For any fixed $\kappa\geqslant0$, the set $M(\hat{\mathfrak s}_{1.7}^\kappa)$ contains no subalgebras of $\mathfrak a_1$ except
the subalgebra $\hat{\mathfrak s}_{1.7}^\kappa$ itself.
This is why there is no equivalence between $\hat{\mathfrak s}_{1.7}^\kappa$ and \smash{$\hat{\mathfrak s}_{1.7}^{\kappa'}$} with $\kappa\ne\kappa'$.

The only subalgebra of $\mathfrak a_1$ that is contained in the set $M(\hat{\mathfrak s}_{1.4}^1)\setminus\{\hat{\mathfrak s}_{1.4}^1\}$
is the subalgebra $\mathfrak s_{1.3}$.

As for the remaining subalgebras from the family $\hat{\mathfrak s}_{1.4}^\kappa$,
for any fixed $\kappa$ with $\kappa\ne1$, the set $M(\hat{\mathfrak s}_{1.4}^\kappa)\setminus\{\hat{\mathfrak s}_{1.4}^\kappa\}$
contains only two subalgebras of $\mathfrak a_1$.
These subalgebras are $\hat{\mathfrak s}_{1.4}^{\tilde\kappa}$ with $\tilde\kappa=\frac{\kappa+3}{\kappa-1}$
and~$\hat{\mathfrak s}_{1.4}^{\hat\kappa}$ with $\hat\kappa=-\frac{\kappa-3}{\kappa+1}$.
The equivalences  $\hat{\mathfrak s}_{1.4}^{\kappa}\sim\hat{\mathfrak s}_{1.4}^{\tilde\kappa}\sim\hat{\mathfrak s}_{1.4}^{\hat\kappa}$
allow us to gauge the range of $\kappa$ to the set $[0,1)$.

Summing up, the ${\rm SL}_3(\mathbb R)$-inequivalent one-dimensional subalgebras of $\mathfrak a_1$ are
$\hat{\mathfrak s}_{1.2}^{\delta}$,
$\hat{\mathfrak s}_{1.7}^\kappa$,
$\hat{\mathfrak s}_{1.4}^{\mu'}$ with $\mu'\in[0,1]$
and $\hat{\mathfrak s}_{1.6}^1$.

\medskip\par\noindent
$\boldsymbol{2\rm D.}$
For a two-dimensional subalgebra
$\h\in\{\hat{\mathfrak s}_{2.1},\hat{\mathfrak s}_{2.2}^{\delta},\hat{\mathfrak s}_{2.6},
\hat{\mathfrak s}_{2.9},\hat{\mathfrak s}_{2.10},\hat{\mathfrak s}_{2.13}\}$,
the set $M(\h)$ contains only one element that is a subalgebra of $\mathfrak a_1$.
This element is the subalgebra $\h$ itself.
This means that the subalgebras
$\hat{\mathfrak s}_{2.1}$, $\hat{\mathfrak s}_{2.2}^{\delta}$, $\hat{\mathfrak s}_{2.6}$,
$\hat{\mathfrak s}_{2.9}$, $\hat{\mathfrak s}_{2.10}$ and $\hat{\mathfrak s}_{2.13}$
are mutually inequivalent and they are not equivalent to other subalgebras from the list.

For a fixed value of $\kappa$, the only elements of $M(\hat{\mathfrak s}_{2.4}^\kappa)$
that are subalgebras of $\mathfrak a_1$
are $\hat{\mathfrak s}_{2.4}^\kappa$ itself and $\langle E_2+\tilde\kappa D,E_3\rangle$,
where $\tilde\kappa=-(\kappa-3)/(\kappa+1)$.
The latter subalgebra is equivalent to the subalgebra~$\hat{\mathfrak s}_{2.11}^{-\tilde\kappa}$
under the action of the matrix $S$.

Acting on the subalgebra~$\hat{\mathfrak s}_{2.3}$ 
by the matrix $SM_2(0,0)$, we obtain the subalgebra $\smash{\hat{\mathfrak s}_{2.11}^1}$.

Any subalgebra~$\hat{\mathfrak s}_{2.5}^{\kappa'}$ with $\kappa'\ne1$ is
${\rm SL}_3(\mathbb R)$-equivalent to $\hat{\mathfrak s}_{2.11}^{\tilde\kappa'}$ with $\tilde\kappa':=(\kappa'+3)/(\kappa'-1)$
under the conjugation by the matrix $M_3(0,0)$.
Meanwhile, the only element of the set $M(\hat{\mathfrak s}_{2.5}^{1})\setminus\{\hat{\mathfrak s}_{2.5}^{1}\}$
is the subalgebra $\hat{\mathfrak s}_{2.8}=\langle E_1,D\rangle$.

Analogously, we consider the subalgebra $\hat{\mathfrak s}_{2.7}^1$:
the only element of the set $M(\hat{\mathfrak s}_{2.7}^1)\setminus\{\hat{\mathfrak s}_{2.7}^1\}$
that is a subalgebra of $\mathfrak a_1$
is the subalgebra $\langle E_2+D-2P_2,E_1\rangle$.
Conjugating this subagebra by the matrix ${\rm diag}(-2,-1/2,1)$
yields the subalgebra $\hat{\mathfrak s}_{2.12}=\langle E_2+D+P_2,E_1\rangle$.
%Similarly, the subalgebra~$\hat{\mathfrak s}_{2.7}^{-1}$ is equivalent to the subalgebra $\langle E_2+D+2P_2,E_1\rangle$
%under the action by $M_3(0,0)$,
%and the latter is conjugate with the subalgebra $\hat{\mathfrak s}_{2.12}$
%by the matrix ${\rm diag}(2,1/2,1)$.

The equivalence between the elements of the set $M(\mathfrak s_{2.11}^\gamma)\setminus\{\mathfrak s_{2.11}^\gamma\}$
and the subalgebras from the list $\hat{\mathfrak s}_{i,j}^*$ are exhaustively described in the previous paragraphs:
$\mathfrak s_{2.4}^\kappa\sim\mathfrak s_{2.11}^{\tilde\kappa}$,
$\mathfrak s_{2.5}^{\kappa'}\sim\mathfrak s_{2.11}^{\tilde\kappa'}$,
\smash{$\mathfrak s_{2.3}\sim\mathfrak s_{2.11}^1$}. %\smash{$\mathfrak s_{2.7}^{-1}\sim\mathfrak s_{2.11}^{-1}$}
There is no equivalence between $\mathfrak s_{2.11}^\gamma$ and \smash{$\mathfrak s_{2.11}^{\gamma'}$}
when $\gamma\ne\gamma'$.

This step results in the subalgebras 
$\hat{\mathfrak s}_{2.1}$,
$\hat{\mathfrak s}_{2.2}^{\delta}$,
$\hat{\mathfrak s}_{2.6}$,
$\hat{\mathfrak s}_{2.7}^1$,
$\hat{\mathfrak s}_{2.8}$,
$\hat{\mathfrak s}_{2.9}$,
$\hat{\mathfrak s}_{2.10}$,
$\hat{\mathfrak s}_{2.11}^\gamma$,
$\hat{\mathfrak s}_{2.13}$.

\medskip\par\noindent
$\boldsymbol{3\rm D.}$
As before, we first identify the subalgebras
whose images under the action of~$M$ contain precisely one subalgebra of~$\mathfrak a_1$.
These are $\hat{\mathfrak s}_{3.1}$,
$\hat{\mathfrak s}_{3.2}$,
$\hat{\mathfrak s}_{3.3}^\kappa$,
$\hat{\mathfrak s}_{3.4}^1$,
$\hat{\mathfrak s}_{3.5}^\kappa$,
$\hat{\mathfrak s}_{3.9}$
and $\hat{\mathfrak s}_{3.10}$.
Consequently, they are mutually inequivalent and inequivalent to other subalgebras from the list.
Additionally, $\hat{\mathfrak{s}}_{3.12}$ is the unique Levi factor of $\mathfrak{a}_1$ present in the list,
ensuring it cannot be equivalent to any other subalgebra.

The only element of the set $M(\hat{\mathfrak s}_{3.6})\setminus\{\hat{\mathfrak s}_{3.6}\}$
that is a subalgebra of $\mathfrak a_1$ is the subalgebra~$\hat{\mathfrak s}_{3.8}^1$.

The only element of the set $M(\hat{\mathfrak s}_{3.7})$ distinct form $\hat{\mathfrak s}_{3.7}$
is the subalgebra $\langle E_2,E_3,D\rangle$, which is equivalent to $\hat{\mathfrak s}_{3.11}$ under the action of the matrix~$S$.

For any fixed value of $\gamma$ with $\gamma\ne1$, the only element of the set $M(\hat{\mathfrak s}_{3.8}^\gamma)$
distinct from $\hat{\mathfrak s}_{3.8}^\gamma$ that is a subalgebra of $\mathfrak a_1$
is the subalgebra $\hat{\mathfrak s}_{3.8}^{\tilde\gamma}$ with $\tilde\gamma=\frac{\gamma+3}{\gamma-1}$.
This is why we can restrict the range of $\gamma$ to the set $[-1,3]$.

As a result, we obtain the subalgebras $\hat{\mathfrak s}_{3.1}$,
$\hat{\mathfrak s}_{3.2}$,
$\hat{\mathfrak s}_{3.3}^\kappa$,
$\hat{\mathfrak s}_{3.4}^1$,
$\hat{\mathfrak s}_{3.5}^\kappa$,
$\hat{\mathfrak s}_{3.8}^\mu$ with $\mu\in[-1,3]$,
$\hat{\mathfrak s}_{3.9}$,
$\hat{\mathfrak s}_{3.10}$,
$\hat{\mathfrak s}_{3.11}$
and  $\hat{\mathfrak s}_{3.12}$.

\medskip\par\noindent
$\boldsymbol{\rm4 D\textbf{ and }5D.}$
The image under the action by $M$ on any four-dimensional subalgebra $\hat{\mathfrak s}_{4.j}^*$ 
contains no subalgebras of $\mathfrak a_1$ different from $\hat{\mathfrak s}_{4.j}^*$.
Therefore, the four-dimensional subalgebras are ${\rm SL}_3(\mathbb R)$-inequivalent.
For five-dimensional subalgebras $\hat{\mathfrak s}_{5.1}$ and $\hat{\mathfrak s}_{5.2}$,
notice that the former is solvable whereas the latter is not.
Hence, they cannot be conjugate under any automorphism.
\end{proof}

\subsection{Inequivalent subalgebras of $\mathfrak{sl}_3(\mathbb R)$}\label{subsec:MergingSL3}

\looseness=-1
To construct a complete list of ${\rm SL}_3(\mathbb R)$-inequivalent subalgebras of the Lie algebra $\mathfrak{sl}_3(\mathbb R)$,
we consider its maximal subalgebras $\mathfrak a_1$, $\mathfrak a_2$, $\mathfrak{so}_{2,1}(\mathbb R)$
and $\mathfrak{so}_3(\mathbb R)$ and combine their lists of ${\rm SL}_3(\mathbb R)$-inequivalent subalgebras up to ${\rm SL}_3(\mathbb R)$-equivalence.
\begin{enumerate}\itemsep=0ex
\item The list of ${\rm SL}_3(\mathbb R)$-inequivalent subalgebras of the algebra $\mathfrak a_1$ is given in Corollary~\ref{cor:IneqSubalgsOfA1}.
\item\looseness=-1
To the list from Corollary~\ref{cor:IneqSubalgsOfA1} we need to add subalgebras of $\mathfrak a_2$ that are ${\rm SL}_3(\mathbb R)$-inequivalent to the subalgebras of $\mathfrak a_1$.
Clearly, if a subalgebra $\mathfrak s\subset\mathfrak a_2$ leaves two-dimensional subspace of $\mathbb R^3$ invariant,
then it is equivalent to a subalgebra from $\mathfrak a_1$.
Thus, among the list of ${\rm SL}_3(\mathbb R)$-inequivalent subalgebras of $\mathfrak a_2$,
we omit those whose arbitrary element is a~block-diagonal or a lower-triangular matrix.
The only proper subalgebras of~$\mathfrak a_2$ that remain are
$\langle E_1+E_3+\kappa D,R_1,R_2\rangle$,
$\langle E_1+E_3,D,R_1,R_2\rangle$ and
$\langle E_1,E_2,E_3,R_1,R_2\rangle$,
where $\kappa\geqslant0$.
Verifying that these subalgebras do not leave two-dimensional invariant subspace
and confirming that those of the same dimension are ${\rm SL}_3(\mathbb R)$-inequivalent
is a computational problem, which we solve using {\sf Sage} computational system (see further section Code availability).
%\hyperlink{Code}{Code availability}
 %\href{https://github.com/zchapovsky/sage-notebooks/blob/main/notebooks/sl3_subalgebras_verification.ipynb}{sl3\_subalgebras\_verification}. 
\item
Include the irreducibly embedded subalgebras $\mathfrak{so}_3(\mathbb R)=\mathfrak f_{3.12}$ and $\mathfrak{so}_{2,1}(\mathbb R)=\mathfrak f_{3.13}$.
Since their proper subalgebras are conjugate to subalgebras of $\mathfrak a_1$,
these subalgebras should be omitted in the final classification list.
\end{enumerate}

\begin{theorem}\label{thm:SubalgebrasOfSL3}
A complete list of proper ${\rm SL}_3(\mathbb R)$-inequivalent subalgebras of the algebra $\mathfrak{sl}_3(\mathbb R)$
is exhausted by the following:
\begin{gather*}
1{\rm D}\colon\
\mathfrak f_{1.1}^{\delta}=\langle E_1+\delta P_1\rangle,\ \
\mathfrak f_{1.2}^\kappa  =\langle E_1+E_3+\kappa D\rangle,\ \
\mathfrak f_{1.3}^{\mu'}=\langle E_2+\mu'D\rangle,\ \
\mathfrak f_{1.4}=\langle E_1+D\rangle,
\\[1ex]
2{\rm D}\colon\
\mathfrak f_{2.1}          =\langle P_1,P_2\rangle,\
\mathfrak f_{2.2}^\delta   =\langle E_1+\delta P_1,P_2\rangle,\
\mathfrak f_{2.3}          =\langle E_2+D+P_2,P_1\rangle,\
\mathfrak f_{2.4}          =\langle E_1+D,P_2\rangle,
\\
\hphantom{2{\rm D}\colon\ }
\mathfrak f_{2.5}          =\langle E_1,D\rangle,\ \
\mathfrak f_{2.6}          =\langle E_2,D\rangle,\ \
\mathfrak f_{2.7}          =\langle E_1+E_3,D\rangle,
\\
\hphantom{2{\rm D}\colon\ }
\mathfrak f_{2.8}^\gamma   =\langle E_2+\gamma D,E_1\rangle,\ \
\mathfrak f_{2.9}       =\langle E_2-3D, E_1+P_1\rangle,
\\[1ex]
3{\rm D}\colon\
\mathfrak f_{3.1}            =\langle E_1,P_1,P_2\rangle,\ \
\mathfrak f_{3.2}            =\langle D,P_1,P_2\rangle,\ \
\mathfrak f_{3.3}^\kappa     =\langle E_2+\kappa D,P_1,P_2\rangle,
\\
\hphantom{3{\rm D}\colon\ }
\mathfrak f_{3.4}=\langle E_1\!+\!D,P_1,P_2\rangle,\
\mathfrak f_{3.5}^\kappa=\langle E_1\!+\!E_3\!+\!\kappa D,P_1,P_2\rangle,\
\mathfrak f_{3.6}^\kappa=\langle E_1\!+\!E_3\!+\!\kappa D,R_1,R_2\rangle,
\\
\hphantom{3{\rm D}\colon\ }
\mathfrak f_{3.7}^\mu=\langle E_2+\mu D,E_1,P_2\rangle,\ \
\mathfrak f_{3.8}    =\langle E_2-D+P_1,E_1,P_2\rangle,
\\
\hphantom{3{\rm D}\colon\ }
\mathfrak f_{3.9}    =\langle E_2-3D, E_1+P_1,P_2\rangle,\ \
\mathfrak f_{3.10}       =\langle E_1,E_2,D\rangle,\ \
\mathfrak f_{3.11}      =\langle E_1,E_2,E_3\rangle,\ \
\\
\hphantom{3{\rm D}\colon\ }
\mathfrak f_{3.12}=\langle E_1+E_3,P_1-R_2,P_2+R_1\rangle,\ \
\mathfrak f_{3.13}=\langle E_1+E_3,P_1+R_2,P_2-R_1\rangle,
\\[1ex]
4{\rm D}\colon\ 
\mathfrak f_{4.1}=\langle E_1,D,P_1,P_2\rangle,\ \
\mathfrak f_{4.2}=\langle E_2,D,P_1,P_2\rangle,\ \
\mathfrak f_{4.3}=\langle E_1+E_3,D,P_1,P_2\rangle,
\\
\hphantom{4{\rm D}\colon\ }
\mathfrak f_{4.4}^\gamma=\langle E_2+\gamma D,E_1,P_1,P_2\rangle,\ \
\mathfrak f_{4.5}=\langle E_1,E_2,D,P_2\rangle,\ \
\mathfrak f_{4.6}=\langle E_1+E_3,D,R_1,R_2\rangle,
\\
\hphantom{4{\rm D}\colon\ }
\mathfrak f_{4.7}=\langle E_1,E_2,E_3,D\rangle,
\\[1ex]
5{\rm D}\colon\
\mathfrak f_{5.1}=\langle E_1,E_2,D,P_1,P_2\rangle,\ \
\mathfrak f_{5.2}=\langle E_1,E_2,E_3,P_1,P_2\rangle,\ \
\mathfrak f_{5.3}=\langle E_1,E_2,E_3,R_1,R_2\rangle,
\\[1ex]
6{\rm D}\colon\
\mathfrak f_{6.1}=\langle E_1,E_2,E_3,D,P_1,P_2\rangle,\ \
\mathfrak f_{6.2}=\langle E_1,E_2,E_3,D,R_1,R_2\rangle,
\end{gather*}
where $\delta\in\{0,1\}$,
$\kappa\geqslant0$,
$\mu\in[-1,3]$,
$\mu'\in[0,1]$
and $\gamma\in\mathbb R$.
\end{theorem}

\begin{remark}\label{rem:DiffBetwClassificationsR}
We compare the subalgebras of $\mathfrak{sl}_3(\mathbb R)$ from Theorem~\ref{thm:SubalgebrasOfSL3} and~\cite[Table~1]{wint2004a}
in detail in Section~\ref{sec:ClassCompare}.
The results are summarized in Table~\ref{tab:ListsCompare}.
We mark by the bullet symbol~$\bullet$ the subalgebras and subalgebra families from~\cite[Table~1]{wint2004a}
that have misprints or improper restrictions on the parameters.
In this brief remark, we focus our attention only on the discrepancies in the list~\cite[Table~1]{wint2004a},
such as wrong parameter constraints, omitted subalgebras and conjugate subalgebras that are listed as distinct.
We discuss sublists of subalgebras from~\cite[Table~1]{wint2004a} of each possible dimension separately.

In the list of one-dimensional subalgebras, the constraints on the parameter $\alpha$ in the family $\smash{W_{1.1}^{(\alpha)}}$ are also incorrect and
should be replaced by $\alpha\in[0,\arctan\frac19]$.

The list of two-dimensional subalgebras in~\cite[Table~1]{wint2004a} coincides with those presented in Theorem~\ref{thm:SubalgebrasOfSL3}.

The list of three-dimensional subalgebras has two errors.
The subalgebra $\mathfrak f_{3.8}$ are missed.
The subalgebra $W_{3.11}$ is conjugate with $W_{3.5}^{(1/18)}$.
Therefore, it should be excluded from the list.

For the family of four-dimensional subalgebras $\smash{W^{(\alpha)}_{4.6}}$ we have that $\smash{W^{(0)}_{4.6}=W^{(\pi)}_{4.6}}$.
In fact, the requirement $\alpha\in[0,\pi)$ yields the correct subalgebra family.

In the list of five- and six-dimensional subalgebras, the linear spans $W_{5.2}$ and $W_{6.2}$ have misprints
and, in fact, do not form subalgebras of~$\mathfrak{sl}_3(\mathbb R)$.

In total, the classification list from~\cite{wint2004a}
has two incorrect families of subalgebras, one redundant subalgebra, one omitted subalgebra and
two subalgebras with misprints.

The omission of the subalgebra~$\mathfrak f_{3.8}$ from the list in~\cite{wint2004a}
requires a detailed analysis.
This subalgebra arises from the subalgebra $\mathfrak s_{3.9}$
of the affine Lie algebra $\mathfrak{aff}_2(\mathbb R)$.
It is also listed as subalgebra $S_{2,2;6}$ in the classification of the nonsplitting subalgebras of $\mathfrak{aff}_2(\mathbb R)$
presented in~\cite[Section~3.3]{wint2004a}.
Direct computations in Theorem~\ref{cor:IneqSubalgsOfA1}
confirm that the subalgebra~$\mathfrak s_{3.9}$ cannot be merged with any other subalgebra of~$\mathfrak{aff}_2(\mathbb R)$.
Therefore, it should be included in the final classification.

Despite this, the subalgebra $\mathfrak f_{3.8}$ is missing from the list in~\cite[Table~1]{wint2004a},
and as demonstrated in Section~\ref{sec:ClassCompare},
it is not equivalent to any other subalgebra therein. 
This contrasts with the presence of  $\mathfrak f_{3.8}$ in the classification by Douglas and Repka~\cite{doug2016a}.
\end{remark}

To finalize this part of our study, we also consider the equivalence generated by the action of the full automorphism group ${\rm Aut}\big(\mathfrak{sl}_3(\mathbb R)\big)$
of the Lie algebra~$\mathfrak{sl}_3(\mathbb R)$ on the set of its subalgebras.
It is well known that the quotient group ${\rm Aut}\big(\mathfrak{sl}_3(\mathbb R)\big)/{\rm Inn}\big(\mathfrak{sl}_3(\mathbb R)\big)$
is isomorphic to $\mathbb Z_2$
and a complete list of outer automorphisms of the algebra $\mathfrak{sl}_3(\mathbb R)$
that are independent up to combining with each other and with inner automorphisms
is exhausted by the single involution $\mathscr I\colon x\mapsto-x^{\mathsf T}$, $x\in\mathfrak{sl}_3(\mathbb R)$, see~\cite{mura1952a}.

\begin{corollary}\label{thm:SubalgebrasOfSL3ModAut}
A complete list of proper ${\rm Aut}(\mathfrak{sl}_3(\mathbb R))$-inequivalent subalgebras of the algebra $\mathfrak{sl}_3(\mathbb R)$
is exhausted by the following spans:
\begin{gather*}1{\rm D}\colon\
\hat{\mathfrak f}_{1.1}^{\delta}=\langle E_1+\delta P_1\rangle,\ \
\hat{\mathfrak f}_{1.2}^\kappa   =\langle E_1+E_3+\kappa D\rangle,\ \
\hat{\mathfrak f}_{1.3}^{\mu'}=\langle E_2+\mu'D\rangle,\ \
\hat{\mathfrak f}_{1.4}=\langle E_1+D\rangle,
\\[0.5ex]
2{\rm D}\colon\
\hat{\mathfrak f}_{2.1}^\delta   =\langle E_1+\delta P_1,P_2\rangle,\ \
\hat{\mathfrak f}_{2.2}          =\langle E_1+D,P_2\rangle,\ \ 
\hat{\mathfrak f}_{2.3}          =\langle E_1,D\rangle,\ \
\hat{\mathfrak f}_{2.4}          =\langle E_2,D\rangle,\ \
\\
\hphantom{2{\rm D}\colon\ }
\hat{\mathfrak f}_{2.5}          =\langle E_1+E_3,D\rangle,\ \
\hat{\mathfrak f}_{2.6}^\kappa   =\langle E_2+\kappa D,E_1\rangle,\ \
\hat{\mathfrak f}_{2.7}       =\langle E_2-3D, E_1+P_1\rangle,
\\[0.5ex]
3{\rm D}\colon\
\hat{\mathfrak f}_{3.1}       =\langle E_1,P_1,P_2\rangle,\ \
\hat{\mathfrak f}_{3.2}       =\langle E_1+D,P_1,P_2\rangle,\ \
\hat{\mathfrak f}_{3.3}^\kappa=\langle E_1\!+\!E_3\!+\!\kappa D,P_1,P_2\rangle,
\\
\hphantom{3{\rm D}\colon\ }
\hat{\mathfrak f}_{3.4}^\mu=\langle E_2+\mu D,E_1,P_2\rangle,\ \
\hat{\mathfrak f}_{3.5}    =\langle E_2-3D, E_1+P_1,P_2\rangle,\ \
\hat{\mathfrak f}_{3.6}       =\langle E_1,E_2,D\rangle,\ \
\\
\hphantom{3{\rm D}\colon\ }
\hat{\mathfrak f}_{3.7}      =\langle E_1,E_2,E_3\rangle,\ \
\hat{\mathfrak f}_{3.8}=\langle E_1+E_3,P_1-R_2,P_2+R_1\rangle,
\\
\hphantom{3{\rm D}\colon\ }
\hat{\mathfrak f}_{3.9}=\langle E_1+E_3,P_1+R_2,P_2-R_1\rangle,
\\[0.5ex]
4{\rm D}\colon\ 
\hat{\mathfrak f}_{4.1}^\mu=\langle E_2-\mu D,E_1,P_1,P_2\rangle,\ \
\hat{\mathfrak f}_{4.2}=\langle E_1,E_2,D,P_2\rangle,\ \
\hat{\mathfrak f}_{4.3}=\langle E_1+E_3,D,P_1,P_2\rangle,
\\
\hphantom{4{\rm D}\colon\ }
\hat{\mathfrak f}_{4.4}=\langle E_1,E_2,E_3,D\rangle,
\\[0.5ex]
5{\rm D}\colon\
\hat{\mathfrak f}_{5.1}=\langle E_1,E_2,D,P_1,P_2\rangle,\ \
\hat{\mathfrak f}_{5.2}=\langle E_1,E_2,E_3,P_1,P_2\rangle,\ \
\\[0.5ex]
6{\rm D}\colon\
\hat{\mathfrak f}_{6.1}=\langle E_1,E_2,E_3,D,P_1,P_2\rangle,\ \
\end{gather*}
where $\delta\in\{0,1\}$, $\kappa\geqslant0$,
$\mu\in[-1,3]$, $\mu'\in[0,1]$ and $\gamma\in\mathbb R$
\end{corollary}

\begin{proof}
It is clear that a complete list of ${\rm Aut}\big(\mathfrak{sl}_3(\mathbb R)\big)$-inequivalent proper subalgebras of the algebra $\mathfrak{sl}_3(\mathbb R)$
is contained within the list presented in Theorem~\ref{thm:SubalgebrasOfSL3}.
In order to find it, we should determine the ${\rm SL}_3(\mathbb R)$-equivalence classes
of the subalgebras $\mathscr J(\mathfrak f_{i.j}^*)$.
The following computations can also be found in the {\sf Sage} notebook (see further section Code availability).
%\hyperlink{Code}{Code availability}
%~\href{https://github.com/zchapovsky/sage-notebooks/blob/main/notebooks/sl3_subalgebras_verification.ipynb}{sl3\_subalgebras\_verification}.

\smallskip\par\noindent\looseness=-1
$\boldsymbol{1\rm D.}$
We have $\mathscr I(\mathfrak f_{1.3}^{\mu'})=\mathfrak f_{1.3}^{\mu'}$,
$\mathscr I(\mathfrak f_{1.2}^\kappa)=\mathfrak f_{1.2}^{-\kappa}$,
and the latter subalgebra is equivalent to $\mathfrak f_{1.2}^\kappa$.
Applying $S$ to the subalgebras $\mathscr I(\mathfrak f_{1.1}^0)$ and $\mathscr I(\mathfrak f_{1.4})$,
we obtain the subalgebras $\mathfrak f_{1.1}^0$ and $\mathfrak f_{1.4}$, respectively.
Acting by the product ${\rm diag}(-1,-1,1)\exp\big(\pi/2(P_2+R_1)\big)$ on the subalgebra $\mathscr I(\mathfrak f_{1.1}^1)$
maps it back to~$\mathfrak f_{1.1}^1$.

\smallskip\par\noindent
$\boldsymbol{2\rm D.}$
We have the following equalities:
$\mathscr I(\mathfrak f_{2.6})=\mathfrak f_{2.6}$
and $\mathscr I(\mathfrak f_{2.7})=\mathfrak f_{2.7}$.
The rotation matrix $\exp\big(\pi/2(P_2+R_1)\big)$ maps the subalgebras
$\mathscr I(\mathfrak f_{2.1})$,
$\mathscr I(\mathfrak f_{2.9})$ and $\mathscr I(\mathfrak f_{2.2}^1)$
to the subalgebras 
\[\mathfrak f_{2.2}^0,\quad
\langle E_2-3D,E_1-P_1\rangle\quad\mbox{and}\quad
\langle E_1-P_1,P_2\rangle,\]
respectively.
The subalgebras $\langle E_2-3D,E_1-P_1\rangle$ and $\langle E_1-P_1,P_2\rangle$
are ${\rm SL}_3(\mathbb R)$-equivalent to $\mathfrak f_{2.9}$ and $\mathfrak f_{2.2}^1$, respectively.

Applying the matrix $S$ to the subalgebras $\mathscr J(\mathfrak f_{2.5})$ and $\mathscr I(\mathfrak f_{2.8}^\kappa)$ 
yields $\mathfrak f_{2.5}$ and $\mathfrak f_{2.8}^{-\kappa}$,
and the latter subalgebra is equivalent to $\mathfrak f_{2.8}^{\kappa}$.
Acting by $S\exp\big(\pi/2(P_2+R_1)\big)$ on the subalgebra $\mathscr J(\mathfrak f_{2.4})$
gives the subalgebra $\langle E_2+D+2P_2,P_1\rangle$, which is equivalent to~$\mathfrak f_{2.3}$.

\smallskip\par\noindent
$\boldsymbol{3\rm D.}$
We have that $\mathscr I(\mathfrak f_{3.12})=\mathfrak f_{3.12}$,
$\mathscr I(\mathfrak f_{3.13})=\mathfrak f_{3.13}$ and
$\mathscr I(\mathfrak f_{3.5}^\kappa)=\mathfrak f_{3.6}^{-\kappa}$.
Acting by the matrix $S$ on the subalgebras $\mathscr I(\mathfrak f_{3.10})$ and $\mathscr I(\mathfrak f_{3.11})$
we obtain $\mathfrak f_{3.10}$ and $\mathfrak f_{3.11}$, respectively.

Conjugating by the matrix $\exp\big(\pi/2(P_2+R_1)\big)$ the subalgebras $\mathscr I(\mathfrak f_{3.1})$
and $\mathscr I(\mathfrak f_{3.2})$, we obtain $\mathfrak f_{3.1}$ and $\mathfrak f_{3.7}^{-1}$, respectively.
Applying $\exp\big(\pi/2(P_2+R_1)\big)$ to the subalgebras $\mathscr I(\mathfrak f_{3.4})$ and $\mathscr I(\mathfrak f_{3.9})$,
we obtain the subalgebras 
\[
\langle E_2-D-2P_1,E_1,P_2\rangle
\quad\mbox{and}\quad
\langle E_2-3D,E_1-P_1,P_2\rangle,
\] 
which are ${\rm SL}_3(\mathbb R)$-equivalent to $\mathfrak f_{3.8}$ and $\mathfrak f_{3.9}$, respectively.

The matrix $\exp\big(\pi/2(P_2+R_1)\big)$ also maps the subalgebra family $\mathscr I(\mathfrak f_{3.7}^\mu)$
to the subalgebra
\[\mathfrak f_{3.3}^{\tilde\mu}=\langle (\mu+1)E_2-(\mu-3)D,P_1,P_2\rangle
\quad\mbox{with}\quad\tilde\mu:=-\tfrac{\mu-3}{\mu+1},\]
which coincides with the family $\mathfrak f_{3.3}^\kappa$ united with the subalgebra~$\mathfrak f_{3.2}$.

\smallskip\par\noindent
$\boldsymbol{4\rm D.}$
We have the following equalities: $\mathscr I(\mathfrak f_{4.3})=\mathfrak f_{4.6}$
and $\mathscr I(\mathfrak f_{4.7})=\mathfrak f_{4.7}$.
Applying $\exp\big(\pi/2(P_2+R_1)\big)$ to $\mathscr I(\mathfrak f_{4.1})$ and $\mathscr I(\mathfrak f_{4.2})$
yields $\mathfrak f_{4.4}^{-1}$ and $\mathfrak f_{4.5}$, respectively.

Acting by $\exp\big(\frac{\pi}2(P_2+R_1)\big)$ on $\mathfrak f_{4.4}^\gamma$, we obtain
\[
\mathfrak f_{4.4}^{-\tilde\gamma}=\langle (\gamma+1)E_2-(\gamma-3)D,E_1,P_1,P_2\rangle
\quad\mbox{with}\quad
\tilde\gamma:=\tfrac{\gamma-3}{\gamma+1}.
\]
Consequently, when $\gamma\ne-1$, we can gauge the parameter $-\gamma$ to the set $(-1,3]$.

\smallskip\par\noindent
$\boldsymbol{\rm5 D\textbf{ and }6D.}$
This case is straightforward: $\mathscr I(\mathfrak f_{5.3})=\mathfrak f_{5.2}$ and $\mathscr I(\mathfrak f_{6.2})=\mathfrak f_{6.1}$.
\end{proof}

\section{Subalgebras of $\mathfrak{sl}_3(\mathbb C)$}\label{sec:sl3CClassification}

We classify the subalgebras of the complex order-three special linear Lie algebra~$\mathfrak{sl}_3(\mathbb C)$
using the same method as in Section~\ref{sec:sl3Classification}.
This is why we adopt the notations from that section for convenience.
The algebra~$\mathfrak{sl}_3(\mathbb C)$ is spanned by the matrices
$E_1,E_2,E_3,D,P_1,P_2,R_1,R_2$.
The only irreducibly embedded subalgebra of~$\mathfrak{sl}_3(\mathbb C)$ with respect to its defining representation on $\mathbb C^3$
is the subalgebra isomorphic to $\mathfrak{sl}_2(\mathbb C)$.
This subalgebra can be chosen as the complexification
of the real subalgebra $\mathfrak f_{3.12}=\langle E_1+E_3,P_1-R_2,P_2+R_1\rangle$ of $\mathfrak{sl}_3(\mathbb R)$.
The maximal reducibly embedded subalgebras are $\mathfrak a_1:=\langle E_1,E_2,E_3,D,P_1,P_2\rangle$
and $\mathfrak a_2:=\langle E_1,E_2,E_3,D,R_1,R_2\rangle$.
Both $\mathfrak a_1$ and $\mathfrak a_2$ are isomorphic to the complex rank-two affine Lie algebra $\mathfrak{aff}_2(\mathbb C)=\mathfrak{gl}_2(\mathbb C)\ltimes\mathbb C^2$,
where $\mathfrak{gl}_2(\mathbb C)=\langle e_1,e_2,e_3,e_4\rangle=\mathfrak{sl}_2(\mathbb C)\times\langle e_4\rangle$.

We present the result of the classification of subalgebras of~$\mathfrak{sl}_3(\mathbb C)$
as a sequence of assertions without proof.
Starting with listing inequivalent subalgebras of the algebra $\mathfrak{gl}_2(\mathbb C)$,
we construct the list of inequivalent subalgebras of the algebra $\mathfrak{aff}_2(\mathbb C)$
using the method from Section~\ref{subsec:semidirect}.
This step is essential in the course of listing subalgebras of~$\mathfrak{sl}_3(\mathbb C)$
and its proof can be derived from that in Section~\ref{sec:aff2Classification}.
The list of inequivalent subalgebras of~$\mathfrak{aff}_2(\mathbb C)$ is presented in Theorem~\ref{thm:SubalgebrasOfAffineAlgC}.
The last step is to combine the obtained lists of subalgebras of the maximal subalgebras of~$\mathfrak{sl}_3(\mathbb C)$
modulo the ${\rm SL}_3(\mathbb C)$-equivalence as it is done in Section~\ref{subsec:MergingSL3},
which results in Theorem~\ref{thm:SubalgebrasOfSL3C}.

\begin{theorem}\label{prop:SubalgebrasOfGL2C}
A complete list of inequivalent proper subalgebras of the algebra~$\mathfrak{gl}_2(\mathbb C)$
is exhausted by the following subalgebras:
\begin{gather*}
{\rm 1D}\colon\quad
\h_{1.1}=\langle e_1\rangle,\ \
\h_{1.2}=\langle e_4\rangle,\ \
\h_{1.3}^{\rho,\varphi}=\langle e_2+\rho{\rm e}^{i\varphi} e_4\rangle,\ \
\h_{1.4}=\langle e_1+e_4\rangle,
\\[0.5ex]
{\rm 2D}\colon\quad
\h_{2.1}       =\langle e_1,e_4\rangle,\quad
\h_{2.2}       =\langle e_2,e_4\rangle,\quad
\h_{2.3}^\gamma=\langle e_2+\gamma e_4,e_1\rangle,
\\[0.5ex]
{\rm 3D}\colon\quad
\h_{3.1}=\langle e_1,e_2,e_4\rangle,\quad
\h_{3.2}=\langle e_1,e_2,e_3\rangle,
\end{gather*}
where $\rho\in\mathbb R_{\geqslant0}$, $-\pi/2<\varphi\leqslant\pi/2$ and $\gamma\in\mathbb C$.
\end{theorem}

\begin{theorem}\label{thm:SubalgebrasOfAffineAlgC}
A complete list of inequivalent proper subalgebras of the real rank-two affine Lie algebra~$\mathfrak{aff}_2(\mathbb C)$
is given by
\begin{gather*}
1{\rm D}\colon\ 
\mathfrak s_{1.1}            =\langle f_1\rangle,\ \
\mathfrak s_{1.2}^{\delta}  =\langle e_1+\delta f_1\rangle,\ \
\mathfrak s_{1.3}            =\langle e_4\rangle,\ \
\mathfrak s_{1.4}^{\rho,\varphi}     =\langle e_2+\rho{\rm e}^{\rm i\varphi}e_4\rangle,\\
\hphantom{1{\rm D}\colon\ }
\mathfrak s_{1.5}            =\langle e_2+e_4+f_2\rangle,\ \
\mathfrak s_{1.6}            =\langle e_1+e_4\rangle,\ \
\\[0.5ex]
2{\rm D}\colon\ 
\mathfrak s_{2.1}          =\langle f_1,f_2\rangle,\ \
\mathfrak s_{2.2}^{\delta}=\langle e_1+\delta f_1,f_2\rangle,\ \
\mathfrak s_{2.3}          =\langle e_4,f_1\rangle,\ \
\mathfrak s_{2.4}^{\rho,\varphi}   =\langle e_2+\rho{\rm e}^{\rm i\varphi} e_4,f_1\rangle,\ \
\\
\hphantom{2{\rm D}\colon\ }
\mathfrak s_{2.5}^{\rho',\varphi}=\langle e_2+\rho'{\rm e}^{\rm i\varphi}e_4,f_2\rangle,\ \
\mathfrak s_{2.6}          =\langle e_2+e_4+f_2,f_1\rangle,\ \
\mathfrak s_{2.7}          =\langle e_1+e_4,f_2\rangle,\ \
\\
\hphantom{2{\rm D}\colon\ }
\mathfrak s_{2.8}            =\langle e_1,e_4\rangle,\ \
\mathfrak s_{2.9}            =\langle e_2,e_4\rangle,\ \
\mathfrak s_{2.10}^\gamma    =\langle e_2+\gamma e_4,e_1\rangle,\ \
\mathfrak s_{2.11}^\varepsilon=\langle e_2+\varepsilon e_4+f_2,e_1\rangle,
\\
\hphantom{2{\rm D}\colon\ }
\mathfrak s_{2.13}            =\langle e_2-3e_4,e_1+f_1\rangle,
\\[0.5ex]
3{\rm D}\colon\ 
\mathfrak s_{3.1}            =\langle e_1,f_1,f_2\rangle,\ \
\mathfrak s_{3.2}            =\langle e_4,f_1,f_2\rangle,\ \
\mathfrak s_{3.3}^{\rho,\varphi}=\langle e_2+\rho{\rm e}^{\rm i\varphi}e_4,f_1,f_2\rangle,\ \
\\
\hphantom{3{\rm D}\colon\ }
\mathfrak s_{3.4}            =\langle e_1+e_4,f_1,f_2\rangle,\ \
\mathfrak s_{3.5}            =\langle e_1,e_4,f_2\rangle,\ \
\mathfrak s_{3.6}            =\langle e_2,e_4,f_1\rangle,\ \
\\
\hphantom{3{\rm D}\colon\ }
\mathfrak s_{3.7}^\gamma     =\langle e_2+\gamma e_4,e_1,f_2\rangle,\ \
\mathfrak s_{3.8}         =\langle e_2-e_4+f_1, e_1, f_2\rangle,\ \
\mathfrak s_{3.9}            =\langle e_2-3e_4, e_1+f_1, f_2\rangle,
\\
\hphantom{3{\rm D}\colon\ }
\mathfrak s_{3.10}           =\langle e_1,e_2,e_4\rangle,\ \
\mathfrak s_{3.11}           =\langle e_1,e_2,e_3\rangle,
\\[0.5ex]
4{\rm D}\colon\ 
\mathfrak s_{4.1}        =\langle e_1,e_4,f_1,f_2\rangle,\ \
\mathfrak s_{4.2}        =\langle e_2,e_4,f_1,f_2\rangle,\ \
\mathfrak s_{4.3}^\gamma =\langle e_2+\gamma e_4,e_1,f_1,f_2\rangle,\ \
\\
\hphantom{4{\rm D}\colon\ }
\mathfrak s_{4.4}        =\langle e_1,e_2,e_4,f_2\rangle,\ \
\mathfrak s_{4.5}        =\langle e_1,e_2,e_3,e_4\rangle,
\\[0.5ex]
5{\rm D}\colon\ 
\mathfrak s_{5.1}=\langle e_1,e_2,e_4,f_1,f_2\rangle,\ \
\mathfrak s_{5.2}=\langle e_1,e_2,e_3,f_1,f_2\rangle,
\end{gather*}
where $\varepsilon\in\{-1,1\}$, $\delta\in\{0,1\}$,
$\rho\in\mathbb R_{\geqslant0}$, $\rho'\in\mathbb R_{>0}$, $-\pi/2<\varphi\leqslant\pi/2$
and $\gamma\in\mathbb C$.
\end{theorem}

Note that each subalgebra from Theorem~\ref{prop:SubalgebrasOfGL2C}
has a counterpart in the list from Theorem~\ref{thm:SubalgebrasOfGL2}.
Furthermore, it is straightforward to verify that the proof of Theorem~\ref{thm:SubalgebrasOfAffineAlg} does not
essentially depend on the base field being real.
In this way, up to the substitutions 
$\kappa\to\rho{\rm e}^{i \varphi}$ and $\kappa'\to\rho'{\rm e}^{i \varphi}$,
the proof of Theorem~\ref{thm:SubalgebrasOfAffineAlgC} follows immediately from that of Theorem~\ref{thm:SubalgebrasOfAffineAlg}.

It is clear that Lemma~\ref{lem:SL3Decomp} takes the same form when the base field is extended to the complex numbers~$\mathbb C$.
Since each subalgebra in the list from Theorem~\ref{thm:SubalgebrasOfAffineAlgC}
has a counterpart in the list from Theorem~\ref{thm:SubalgebrasOfAffineAlg}
and the proof of Corollary~\ref{cor:IneqSubalgsOfA1} does not depend on the underlying field being $\mathbb R$,
we can use the same {\sf Sage} notebook, dropping the assumption that all parameters are real,
to verify the following corollary.

\begin{corollary}\label{cor:IneqSubalgsOfA1C}
Among the subalgebras of $\mathfrak a_1$, the following are inequivalent with respect to the action of~${\rm SL}_3(\mathbb C)$,
\begin{align*}
&1{\rm D}\colon\
\langle E_1+\delta P_1\rangle,\ \
\langle E_2+\mu' D\rangle,\ \
\langle E_1+D\rangle,
\\[0.5ex]
&2{\rm D}\colon\
\langle P_1,P_2\rangle,\ \
\langle E_1+\delta P_1,P_2\rangle,\ \
\langle E_2+D+P_2,P_1\rangle,\ \
\langle E_1+D,P_2\rangle,\ \
\\
&\hphantom{2{\rm D}\colon\ }
\langle E_1,D\rangle,\ \
\langle E_2,D\rangle,\ \
\langle E_2+\gamma D,E_1\rangle,\ \
\langle E_2-3D, E_1+P_1\rangle,
\\[0.5ex]
&3{\rm D}\colon\
\langle E_1,P_1,P_2\rangle,\ \
\langle D,P_1,P_2\rangle,\ \
\langle E_2+\rho{\rm e}^{\rm i\varphi} D,P_1,P_2\rangle,\ \
\langle E_1+D,P_1,P_2\rangle,
\\
&\hphantom{3{\rm D}\colon\ }
\langle E_2+\mu D,E_1,P_2\rangle,\ \
\langle E_2-D+P_1,E_1,P_2\rangle,\ \
\langle E_2-3D, E_1+P_1,P_2\rangle,
\\
&\hphantom{3{\rm D}\colon\ }
\langle E_1,E_2,D\rangle,\ \
\langle E_1,E_2,E_3\rangle,
\\[0.5ex]
&4{\rm D}\colon\ 
\langle E_1,D,P_1,P_2\rangle,\ \
\langle E_2,D,P_1,P_2\rangle,\ \
\langle E_2+\gamma D,E_1,P_1,P_2\rangle,
\\
&\hphantom{3{\rm D}\colon\ }
\langle E_1,E_2,D,P_2\rangle,\ \
\langle E_1,E_2,E_3,D\rangle,
\\[0.5ex]
&5{\rm D}\colon\
\langle E_1,E_2,D,P_1,P_2\rangle,\ \
\langle E_1,E_2,E_3,P_1,P_2\rangle,
\end{align*}
where $\delta\in\{0,1\}$,
$\mu'\in({\rm B}_{2}(-1)\cap\{z\in\mathbb C\mid \Re z\geqslant0\})\setminus\{{\rm i}a\mid a\in\mathbb R_{\leqslant0}\}$,
$\mu\in{\rm B}_2(1)$, $\rho\in\mathbb R_{\geqslant0}$, $-\pi/2<\varphi\leqslant\pi/2$
and~$\gamma\in\mathbb C$.
Here and in what follows $B_r(a)$ denotes the closed ball of radius $r$ centred at $z=a$,
$B_r(a):=\{z\in\mathbb C\mid |z|\leqslant r\}$.
\end{corollary}

\begin{theorem}\label{thm:SubalgebrasOfSL3C}
A complete list of proper ${\rm SL}_3(\mathbb C)$-inequivalent subalgebras of the algebra $\mathfrak{sl}_3(\mathbb C)$
is exhausted by the following:
\begin{align*}
&1{\rm D}\colon\
\mathfrak f_{1.1}^{\delta}=\langle E_1+\delta P_1\rangle,\ \
\mathfrak f_{1.2}^{\mu'}  =\langle E_2+\mu'D\rangle,\ \
\mathfrak f_{1.3}         =\langle E_1+D\rangle,
\\[0.5ex]
&2{\rm D}\colon\
\mathfrak f_{2.1}          =\langle P_1,P_2\rangle,\ \
\mathfrak f_{2.2}^\delta  =\langle E_1+\delta P_1,P_2\rangle,\ \
\mathfrak f_{2.3}          =\langle E_2+D+P_2,P_1\rangle,
\\
&\hphantom{2{\rm D}\colon\ }
\mathfrak f_{2.4}          =\langle E_1+D,P_2\rangle,\ \
\mathfrak f_{2.5}          =\langle E_1,D\rangle,\ \
\mathfrak f_{2.6}          =\langle E_2,D\rangle,
\\
&\hphantom{2{\rm D}\colon\ }
\mathfrak f_{2.7}^\gamma   =\langle E_2+\gamma D,E_1\rangle,\ \
\mathfrak f_{2.8}          =\langle E_2-3D,E_1+P_1\rangle,
\\[0.5ex]
&3{\rm D}\colon\
\mathfrak f_{3.1}               =\langle E_1,P_1,P_2\rangle,\ \
\mathfrak f_{3.2}               =\langle D,P_1,P_2\rangle,\ \
\mathfrak f_{3.3}^{\rho,\varphi}=\langle E_2+\rho{\rm e}^{\rm i\varphi} D,P_1,P_2\rangle,\ \
\\
&\hphantom{3{\rm D}\colon\ }
\mathfrak f_{3.4}               =\langle E_1+D,P_1,P_2\rangle,\ \
\mathfrak f_{3.5}^\mu        =\langle E_2+\mu D,E_1,P_2\rangle,\ \
\mathfrak f_{3.6}            =\langle E_2-D+P_1,E_1,P_2\rangle,\ \
\\
&\hphantom{3{\rm D}\colon\ }
\mathfrak f_{3.7}            =\langle E_2-3D, E_1+P_1,P_2\rangle,\ \
\mathfrak f_{3.8}       =\langle E_1,E_2,D\rangle,\ \
\mathfrak f_{3.9}       =\langle E_1,E_2,E_3\rangle,
\\
&\hphantom{3{\rm D}\colon\ }
\mathfrak f_{3.10}=\langle E_1+E_3,P_1-R_2,P_2+R_1\rangle,\ \
\\[0.5ex]
&4{\rm D}\colon\ 
\mathfrak f_{4.1}       =\langle E_1,D,P_1,P_2\rangle,\ \
\mathfrak f_{4.2}       =\langle E_2,D,P_1,P_2\rangle,\ \
\mathfrak f_{4.3}^\gamma=\langle E_2+\gamma D,E_1,P_1,P_2\rangle,\ \
\\
&\hphantom{4{\rm D}\colon\ }
\mathfrak f_{4.4}=\langle E_1,E_2,D,P_2\rangle,\ \
\mathfrak f_{4.5}=\langle E_1,E_2,E_3,D\rangle,
\\[0.5ex]
&5{\rm D}\colon\
\mathfrak f_{5.1}=\langle E_1,E_2,D,P_1,P_2\rangle,\ \
\mathfrak f_{5.2}=\langle E_1,E_2,E_3,P_1,P_2\rangle,\ \
\mathfrak f_{5.3}=\langle E_1,E_2,E_3,R_1,R_2\rangle,
\\[0.5ex]
&6{\rm D}\colon\
\mathfrak f_{6.1}=\langle E_1,E_2,E_3,D,P_1,P_2\rangle,\ \
\mathfrak f_{6.2}=\langle E_1,E_2,E_3,D,R_1,R_2\rangle,
\end{align*}
where $\delta\in\{0,1\}$,
$\mu'\in\big({\rm B}_{2}(-1)\cap\{z\in\mathbb C\mid \Re z\geqslant0\}\big)\setminus\{{\rm i}a\mid a\in\mathbb R_{<0}\}$,
$\mu\in{\rm B}_2(1)$, $\rho\in\mathbb R_{\geqslant0}$, $-\pi/2<\varphi\leqslant\pi/2$
and~$\gamma\in\mathbb C$.
\end{theorem}

\begin{remark}\label{rem:DiffBetwClassificationsC}
We compare the subalgebras of~$\mathfrak{sl}_3(\mathbb C)$ from Theorem~\ref{thm:SubalgebrasOfSL3C}
and~\cite[Table~1]{doug2016a} in detail in Section~\ref{sec:ClassCompare2}.
The results are summarized in Table~\ref{tab:ListsCompare2}.
In this short remark, we discuss the differences between the lists.
For notation, see Section~\ref{sec:ClassCompare2}.

The subalgebra $K_2^2$ has a misprint and should read as $K_2^2=\langle x_1,-\frac13 h_1+\frac13 h_2+x_3\rangle$.
The subalgebra $M^2_{13,0}$ is equivalent to the subalgebra \smash{$M^2_{13,\psi(1/2)}$},
thus should be excluded from the classification.
In total, the classification from~\cite{doug2016a} has one redundant subalgebra and one subalgebra with a misprint.
\end{remark}

To make our research complete, we classify the subalgebras of $\mathfrak{sl}_3(\mathbb C)$ also
with respect to ${\rm Aut}\big(\mathfrak{sl}_3(\mathbb C)\big)$-equivalence.
The factor group ${\rm Aut}\big(\mathfrak{sl}_3(\mathbb C)\big)/{\rm Inn}\big(\mathfrak{sl}_3(\mathbb C)\big)$
is isomorphic to $\mathbb Z_2$,
and thus the group of all automorphisms of $\mathfrak{sl}_3(\mathbb C)$ is generated by the
subgroup of its inner automorphisms ${\rm Inn}\big(\mathfrak{sl}_3(\mathbb C)\big)$ and
the involution $\mathscr I\colon x\mapsto-x^{\mathsf T}$, $x\in\mathfrak{sl}_3(\mathbb C)$,
see~\cite[Theorem~5, p.~283]{jaco1972A}.

\begin{corollary}\label{thm:SubalgebrasOfSL3CAut}
A complete list of proper ${\rm Aut}(\mathfrak{sl}_3(\mathbb C))$-inequivalent subalgebras of the algebra $\mathfrak{sl}_3(\mathbb C)$
is exhausted by the following:
\begin{align*}
&1{\rm D}\colon\
\hat{\mathfrak f}_{1.1}^{\delta}=\langle E_1+\delta P_1\rangle,\ \
\hat{\mathfrak f}_{1.2}^{\mu'}  =\langle E_2+\mu'D\rangle,\ \
\hat{\mathfrak f}_{1.3}         =\langle E_1+D\rangle,
\\[1ex]
&2{\rm D}\colon\
\hat{\mathfrak f}_{2.1}^\delta =\langle E_1+\delta P_1,P_2\rangle,\ \
\hat{\mathfrak f}_{2.2}          =\langle E_1+D,P_2\rangle,\ \
\hat{\mathfrak f}_{2.3}          =\langle E_1,D\rangle,\ \
\\
&\hphantom{2{\rm D}\colon\ }
\hat{\mathfrak f}_{2.4}          =\langle E_2,D\rangle,\ \
\hat{\mathfrak f}_{2.5}^{\rho,\varphi}   =\langle E_2+\rho{\rm e}^{\rm i\varphi} D,E_1\rangle,\ \
\hat{\mathfrak f}_{2.6}          =\langle E_2-3D,E_1+P_1\rangle,
\\[1ex]
&3{\rm D}\colon\
\hat{\mathfrak f}_{3.1}               =\langle E_1,P_1,P_2\rangle,\ \
\hat{\mathfrak f}_{3.2}^{\rho,\varphi}=\langle E_2+\rho{\rm e}^{\rm i\varphi} D,P_1,P_2\rangle,\ \
\hat{\mathfrak f}_{3.3}               =\langle E_1+D,P_1,P_2\rangle,\ \
\\
&\hphantom{3{\rm D}\colon\ }
\hat{\mathfrak f}_{3.4}^\mu=\langle E_2+\mu D,E_1,P_2\rangle,\ \
\hat{\mathfrak f}_{3.5}            =\langle E_2-3D, E_1+P_1,P_2\rangle,\ \
\hat{\mathfrak f}_{3.6}       =\langle E_1,E_2,D\rangle,
\\
&\hphantom{3{\rm D}\colon\ }
\hat{\mathfrak f}_{3.7}       =\langle E_1,E_2,E_3\rangle,\ \
\hat{\mathfrak f}_{3.8}      =\langle E_1+E_3,P_1-R_2,P_2+R_1\rangle,
\\[1ex]
&4{\rm D}\colon\ 
\hat{\mathfrak f}_{4.1}^{\mu'}=\langle E_2-\mu' D,E_1,P_1,P_2\rangle,\ \
\hat{\mathfrak f}_{4.2}=\langle E_1,E_2,D,P_2\rangle,\ \
\hat{\mathfrak f}_{4.3}=\langle E_1,E_2,E_3,D\rangle,
\\[1ex]
&5{\rm D}\colon\
\hat{\mathfrak f}_{5.1}=\langle E_1,E_2,D,P_1,P_2\rangle,\ \
\hat{\mathfrak f}_{5.2}=\langle E_1,E_2,E_3,P_1,P_2\rangle,
\\[1ex]
&6{\rm D}\colon\
\hat{\mathfrak f}_{6.1}=\langle E_1,E_2,E_3,D,P_1,P_2\rangle,
\end{align*}
where $\delta\in\{0,1\}$,
$\mu'\in({\rm B}_{2}(-1)\cap\{z\in\mathbb C\mid \Re z\geqslant0\})\setminus\{{\rm i}a\mid a\in\mathbb R_{<0}\}$,
$\mu\in{\rm B}_2(1)$, $\rho\in\mathbb R_{\geqslant0}$, $-\pi/2<\varphi\leqslant\pi/2$
and~$\gamma\in\mathbb C$.
\end{corollary}

\section{Conclusion}\label{sec:Conclusion}

In this paper, we have reexamined and corrected the classification of subalgebras of
real and complex order-three special linear Lie algebras, building upon the foundational paper~\cite{wint2004a}.
This has been a long-standing and challenging problem with a number of applications in algebra and mathematical physics.

The initial point of the study is the review of the classical approaches for subalgebra classification
that were developed by Patera, Winternitz, Zassenhaus and others in the series of papers
\cite{burd1978a,pate1976b,pate1976c,pate1977b,pate1976a,pate1975a,pate1975b}
for the specific cases of finite-dimensional real and complex Lie algebras.
We have also suggested new perspectives on these methods and rigorously presented their theoretical framework.
As a result, we have suggested the schemes for the classification of subalgebras of a Lie algebra
based on whether it is simple, a direct product, or a semidirect product of its subalgebras
in Sections~\ref{subsec:Simple}, \ref{subsec:Direct} and~\ref{subsec:semidirect}, respectively.

To carry out the classification of subalgebras of a simple Lie algebra $\mathfrak{sl}_3(\mathbb R)$,
we used the approach outlined in Section~\ref{subsec:Simple},
which for the specific case of $\mathfrak g=\mathfrak{sl}_3(\mathbb R)$ required us to go through the following steps:
\begin{itemize}\itemsep=0ex
\item[$(i)$]
using the defining representation~$\mathbb R^3$ of the Lie algebra $\mathfrak{sl}_3(\mathbb R)$,
find all its maximal reducibly and irreducibly embedded subalgebras;
\item[$(ii)$]
for the obtained maximal subalgebras, construct the lists of inequivalent subalgebras
with respect to their corresponding inner automorphism groups;
\item[$(iii)$]
combine the obtained lists modulo the action of the group ${\rm SL}_3(\mathbb R)$.
\end{itemize}
Following~\cite{wint2004a}, in Sections~\ref{subsec:SL3IrredEmbed} and~\ref{subsec:SL3RedEmbed}
it is shown that the Lie algebra~$\mathfrak{sl}_3(\mathbb R)$ contains two irreducibly embedded maximal subalgebras,
namely the special orthogonal Lie algebras $\mathfrak{so}_3(\mathbb R)$ and $\mathfrak{so}_{2,1}(\mathbb R)$,
and two reducibly embedded maximal subalgebras $\mathfrak a_1$ and $\mathfrak a_2$,
each of which is isomorphic to the rank-two affine Lie algebra~$\mathfrak{aff}_2(\mathbb R)$.
In view of step $(ii)$, classifying the subalgebras of~$\mathfrak{aff}_2(\mathbb R)$
is an essential step in the course of solving the primary problem.
This is why we devoted the entire Section~\ref{sec:aff2Classification} to
constructing a complete list of inequivalent subalgebras of~$\mathfrak{aff}_2(\mathbb R)$.
To the best of our knowledge, such list has never been presented in full completeness in the literature before,
cf. \cite[Table 1]{poch2017a}, where the authors  classified only the ``appropriate'' subalgebras modulo a ``weaker'' equivalence
than that one generated by the action of the group ${\rm SL}_3(\mathbb R)$.
The Lie algebra~$\mathfrak{aff}_2(\mathbb R)$ can be viewed as the semidirect product $\mathfrak{gl}_2(\mathbb R)\ltimes\mathbb R^2$.
Therefore, to classify subalgebras of~$\mathfrak{aff}_2(\mathbb R)$,
we applied the approach from Section~\ref{subsec:semidirect} to the Lie algebra $\mathfrak{gl}_2(\mathbb R)\ltimes\mathbb R^2$.
As a result, we present a complete list of inequivalent subalgebras of the rank-two affine Lie algebra~$\mathfrak{aff}_2(\mathbb R)$
in Theorem~\ref{thm:SubalgebrasOfAffineAlg}.
In fact, the classification of subalgebras of the algebra~$\mathfrak{aff}_2(\mathbb R)$ was initiated in~\cite[Section~3.3]{wint2004a},
where its inequivalent ``twisted'' and ``nontwisted'' subalgebras were listed, however this classification was not completed.
Moreover, the validity of these lists is questionable, since to construct them it is essential to have
the correct classification of subalgebras of~$\mathfrak{gl}_2(\mathbb R)$,
which in~\cite[eq.~(3.11)]{wint2004a} was presented with an error and a number of misprints.
This was an additional motivation for us to thoroughly and comprehensively classify the subalgebras of~$\mathfrak{aff}_2(\mathbb R)$.

To complete the classification for~$\mathfrak{sl}_3(\mathbb R)$, in Section~\ref{subsec:MergingSL3} we merged the lists of inequivalent subalgebras
of the maximal subalgebras of~$\mathfrak{sl}_3(\mathbb R)$ modulo the ${\rm SL}_3(\mathbb R)$-equivalence using Lemma~\ref{lem:SL3Decomp}.
The latter lemma applies a variant of the method outlined in Section~\ref{subsec:Simple}, namely Proposition~\ref{prop:ListMerging} and Remark~\ref{rem:ListMerging},
to this particular case of the algebra~$\mathfrak{sl}_3(\mathbb R)$.
The classification results are presented in Theorem~\ref{thm:SubalgebrasOfSL3}.
We have also discussed the differences between the lists in Theorem~\ref{thm:SubalgebrasOfSL3}
and those in~\cite[Table~1]{wint2004a} in Section~\ref{sec:ClassCompare} and Remark~\ref{rem:DiffBetwClassificationsR}.
We found out that the classification in~\cite[Table~1]{wint2004a} has two incorrect families of subalgebras,
one single superfluous subalgebra, one omitted subalgebra and two subalgebras with misprints.

Following the same method as discussed in two preceding paragraphs (or, more specifically, in Section~\ref{sec:sl3Classification}),
we have classified the subalgebras of the order-three complex special linear Lie algebra~$\mathfrak{sl}_3(\mathbb C)$ in Section~\ref{sec:sl3CClassification}
and, as a byproduct, obtained such classification for the rank-two complex affine Lie algebra~$\mathfrak{aff}_2(\mathbb C)$.
We have presented the list of inequivalent subalgebras of the Lie algebras~$\mathfrak{aff}_2(\mathbb C)$ and~$\mathfrak{sl}_3(\mathbb C)$
in Theorems~\ref{thm:SubalgebrasOfAffineAlgC} and~\ref{thm:SubalgebrasOfSL3C}, respectively.
In Section~\ref{sec:ClassCompare2} and Remark~\ref{rem:DiffBetwClassificationsC}, we have thoroughly compared the obtained list
of subalgebras of~$\mathfrak{sl}_3(\mathbb C)$ with those provided in~\cite{doug2016a}.
Our study demonstrates that the classification in~\cite[Table~1]{doug2016a} has one superfluous subalgebra and one subalgebra with a misprint.

Some of the applications of the classifications we obtained and the possible avenues for future research
have been discussed in Section~\ref{sec:Introduction}.

%\subsection*{Acknowledgements}
\medskip\par\noindent{\bf Acknowledgements.}
The authors are grateful to Roman Popovych, Dennis The, Vyacheslav Boyko, Maryna Nesterenko,
Dmytro Popovych, Oleksandra Vinnichenko and Alex Bihlo for valuable discussions, references and comments.
The authors extend their gratitude to the anonymous reviewer, whose insightful corrections
significantly improved the quality of the paper.
This research was supported in part by the NAS of Ukraine under the project 0116U003059 (YeYuCh and SDK),
the Simons Foundation (1290607, YeYuCh), the National Research Foundation of Ukraine (2025.07/0405, YeYuCh)
and the AARMS Graduate Scholarship (SDK).

\medskip\par\noindent{{\bf Code availability.}}
%\subsection*{Code availability}
The {\sf Sage} notebook used to verify the computational results presented in this paper is publicly available.
A rendered version of the notebook can be viewed by following the link \href{https://github.com/zchapovsky/sage-notebooks/blob/main/notebooks/sl3_subalgebras_verification.ipynb}{sl3\_subalgebras\_verification}.
The complete source code is available in the \href{https://github.com/zchapovsky/sage-notebooks}{GitHub repository}.
The notebook can be executed interactively through Binder %\href{https://mybinder.org/}{Binder} 
by following the \href{https://mybinder.org/v2/gh/zchapovsky/sage-notebooks/main}{link}
 which launches a cloud-based {\sf Sage} environment without requiring a local installation.

\appendix

\section{Classifications comparison: real case}\label{sec:ClassCompare}

We thoroughly compare the list of subalgebras of $\mathfrak{sl}_3(\mathbb R)$ provided by Theorem~\ref{thm:SubalgebrasOfSL3}
with that in~\cite[Table~1]{wint2004a}.
The results of this comparison are summarized in Table~\ref{tab:ListsCompare},
where by the bullet symbol~$\bullet$ we mark the subalgebras or subalgebra families listed in~\cite[Table~1]{wint2004a}
with misprints, incorrect constraints on the parameters or conjugate with other subalgebras.
We use the notation~$\rightsquigarrow$ in the Comments column of Table~\ref{tab:ListsCompare}
to indicate that the parameters on the left-hand side of~$\rightsquigarrow$
are incorrect and should be replaced by those on the right-hand side.
We use~$\sim$ to denote ${\rm SL}_3(\mathbb R)$-conjugacy.
The following matrices from~${\rm SL}_3(\mathbb R)$ will be useful in the course of comparison:
\begin{gather*}
Q(\varepsilon):=\exp\big(\varepsilon(E_1+E_3)\big),\quad
S:=\begin{pmatrix}
 0& -1& 0\\
-1&  0& 0\\
 0&  0&-1
\end{pmatrix}\quad\mbox{and}\quad
M_3(0,0):=
\begin{pmatrix}
0  & 0   & -1\\
0  & 1   &  0\\
1  & 0   & 0
\end{pmatrix},
\end{gather*}
where $\varepsilon$ is a real parameter
and the matrix $M_3(0,0)$ comes from Lemma~\ref{lem:SL3Decomp}.

\medskip\par\noindent
$\boldsymbol{1\rm D.}$
For the bases elements of the subalgebras $W_{1.4}=\langle E_3\rangle$ and $W_{1.5}=\langle-E_3+P_2\rangle$,
we have that $-E_3=J_2(0)\oplus J_1(0)$ and $-E_3+P_2=J_3(0)$,
where $J_n(\lambda)$ is the Jordan block of size~$n$ with the eigenvalue $\lambda$.
Thus,~$W_{1.4}$ and~$W_{1.5}$ are equivalent to~$\mathfrak f_{1.1}^0$ and~$\mathfrak f_{1.1}^1$, respectively.

The subalgebra family $\mathfrak f_{1.2}^{\kappa}$ with $\kappa\geqslant0$ corresponds to the subalgebra family $W_{1.2}^{(a)}=\langle\frac12(E_1+E_3)+a D\rangle$.
Acting on the subalgebra $W_{1.3}=\langle 6D-E_3\rangle$ by the matrix $\exp(-\ln(6)E_2)S$, we obtain the subalgebra $\mathfrak f_{1.4}$.

Acting on the subalgebra $\mathfrak f_{1.3}^1$ by the matrix $M_3(0,0)$, we obtain the subalgebra $W_{1.1}^{(\pi/2)}=\langle D\rangle$.
The subalgebra family $\mathfrak f_{1.3}^{\mu}$ with $\mu\in[0,1)$ correspond to the subalgebra family $W_{1.1}^{(\alpha)}=\langle E_2+9\tan\alpha D\rangle$,
where $\alpha\in[0,\arctan\frac19)$.
The parameter restrictions on~$\alpha$ in \cite[Table~1]{wint2004a} are incorrect.

\medskip\par\noindent
$\boldsymbol{2\rm D.}$
We have $\mathfrak f_{2.1}=W_{2.5}$, $\mathfrak f_{2.3}=W_{2.9}$,
$\mathfrak f_{2.6}=W_{2.2}$ and $\mathfrak f_{2.7}=W_{2.1}$.
By applying the matrix~$S$ to the subalgebras $W_{2.3}$, $W_{2.4}$, $W_{2.6}$, $W_{2.7}^{(a)}$, where $a\in\mathbb R$, and~$W_{2.10}$ 
we obtain the subalgebras $\mathfrak f_{2.3}$, $\mathfrak f_{2.2}^0$, $\mathfrak f_{2.2}^1$, $\mathfrak f_{2.8}^{-a}$
and $\mathfrak f_{2.9}$, respectively.
The product $M_3(0,0)Q(\pi)S$ maps the subalgebra $W_{2.8}=\langle E_2-D+P_1,E_3\rangle$
to the subalgebra $\langle E_1+2D,P_2\rangle$.
Further applying  $\exp(-\ln(2)E_2)$ to this subalgebra,
we gauge the multiplier of~$D$ and thus obtain the subalgebra $\mathfrak f_{2.4}$.

\medskip\par\noindent
$\boldsymbol{3\rm D.}$
We have the following equalities:
$W_{3.6}^{(\pi/2)}=\mathfrak f_{3.2}$,
$W_{3.12}=\mathfrak f_{3.11}$,
$W_{3.13}=\mathfrak f_{3.12}$,
$W_{3.14}=\mathfrak f_{3.13}$ and
$W_{3.3}=\mathfrak f_{3.3}^1$.
The family \smash{$W_{3.8}^{(a)}=\langle E_1+E_3+6aD,P_1,P_2\rangle$} with $a\geqslant0$ corresponds to~$\mathfrak f_{3.5}^{6a}$.
The family $\mathfrak f_{3.5}^{\gamma}$, where $\gamma<0$, is omitted in~\cite[Table~1]{wint2004a}.

The subalgebra family $W_{3.6}^{(\alpha)}=\langle E_2+\tan(\alpha)D,P_1,P_2\rangle$ with $\alpha\in[0,\pi/2)$
corresponds to the subalgebra family $\mathfrak f_{3.3}^{\kappa}$ with $\kappa\geqslant0$.

Conjugation by the matrix~$S$ maps the subalgebras $W_{3.1}$, $W_{3.4}$, $W_{3.7}$ and $W_{3.10}$ to
the subalgebras~$\mathfrak f_{3.10}$, $\mathfrak f_{3.1}$, $\mathfrak f_{3.9}$ and $\langle E_1-D,P_1,P_2\rangle$, respectively.
The algebra  $\langle E_1-D,P_1,P_2\rangle$ is equivalent to $\mathfrak f_{3.4}$.

Acting by $S$ on the subalgebra family $W_{3.5}^{(a)}=\langle E_2+6aD,E_3,P_1\rangle$,
where $-\frac12\leqslant a\leqslant\frac16$ and $a\ne-\frac16$,
yields the subalgebra family~$\mathfrak f_{3.7}^{\mu}$ with $\mu\in[-1,3]\setminus\{1\}$.

By acting on~$W_{3.2}$ with~$S$, we obtain the subalgebra $\langle E_1,D,P_2\rangle$,
which is equivalent to~$\mathfrak f_{3.7}^1$ under $M_3(0,0)\in{\rm SL}_3(\mathbb R)$.
Therefore, $W_{3.2}\sim\mathfrak f_{3.7}^1\sim W_{3.5}^{(-1/6)}$.

Conjugation by the matrix $\exp(2P_2)$ maps the subalgebra $W_{3.11}$ to the subalgebra
$\langle E_2+2D,E_3,P_1\rangle$, which is therefore equivalent to the subalgebra \smash{$\mathfrak f_{3.7}^{-1/3}$}
under the action of $M_3(0,0)S$.
The subalgebra \smash{$\mathfrak f_{3.7}^{-1/3}$} is equivalent to $W_{3.2}^{(1/18)}$.

Applying the matrix~$M_3(0,0)$ to the subalgebra family~$W_{3.9}^{(a)}=\langle P_2+R_1+6aD,E_3,P_1\rangle$, where $a\leqslant0$,
we obtain the family $\langle E_1+E_3-E_2+3aD,R_1,R_2\rangle$ with $a\leqslant0$,
which is equivalent to~$\mathfrak f_{3.6}^\gamma$ with $\gamma\geqslant0$.

\medskip\par\noindent
$\boldsymbol{4\rm D.}$
It is clear that $W_{4.1}=\mathfrak f_{4.7}\simeq\mathfrak{gl}_2(\mathbb R)$, $W_{4.3}=\mathfrak f_{4.2}$ and $W_{4.4}=\mathfrak f_{4.3}$.
The matrix~$S$ maps the subalgebra~$W_{4.2}$ to~$\mathfrak f_{4.5}$ and the matrix~$M_3(0,0)$ maps the subalgebra~$W_{4.5}$ to~$\mathfrak f_{4.6}$.

Acting by the matrix $Q(\pi/2)$ on the subalgebra family $W_{4.6}^{(\alpha)}$, we obtain that $W_{4.6}^{(\pi/2)}$
is conjugate with $\mathfrak f_{4.1}$ and the family \smash{$W_{4.6}^{(\alpha)}$}, where $\alpha\in[0,\pi)\setminus\{\pi/2\}$,
correspond to the family $\mathfrak f_{4.4}^\gamma$ with $\gamma\in\mathbb R$.
We also have $W_{4.6}^{(0)}=W_{4.6}^{(\pi)}$, thus the parameter restrictions on~$\alpha$
for the family $W_{4.6}^{(\alpha)}$ in the original text should be modified. 

\medskip\par\noindent
$\boldsymbol{5\rm D}\textbf{ and }\boldsymbol{6\rm D.}$
We have $\mathfrak f_{5.1}=W_{5.3}$, $\mathfrak f_{5.2}=W_{5.1}$ and $\mathfrak f_{6.1}=W_{6.1}$.
The linear spans~$W_{5.2}$ and~$W_{6.2}$ in~\cite[Table~1]{wint2004a} have misprints and
do not form subalgebras of~$\mathfrak {sl}_3(\mathbb R)$.

\section{Classifications comparison: complex case}\label{sec:ClassCompare2}

The classification of all subalgebras of the algebra~$\mathfrak{sl}_3(\mathbb C)$
was first obtained in~\cite{doug2016a}.
As a basis of~$\mathfrak{sl}_3(\mathbb C)$ the canonical Chevalley basis
$\{h_1,h_2,x_1,x_2,x_3,y_1,y_2,y_3\}$ was chosen,
which is related to the basis $\{E_1,E_2,E_3,D,P_1,P_2,R_1,R_2\}$
as follows:
\begin{gather*}
h_1=2E_2,\quad
h_2=3D-E_2,
\\
x_1=-E_3,\quad
x_2=P_2,\quad
x_3=-P_1,
\\
y_1=E_1,\quad
y_2=-R_1,\quad
y_3=-R_2.
\end{gather*}
The complete subalgebra classification was presented in~\cite[Table~1]{doug2016a}
and also spread within the paper in Theorems 5.1, 6.1--6.3, 6.5, 6.7 and 6.8.
We thoroughly compare the list given in Theorem~\ref{thm:SubalgebrasOfSL3C}
with those detailed in~\cite[Table~1]{doug2016a}.
The results of this comparison are summarized in Table~\ref{tab:ListsCompare2}.

\medskip\par\noindent
$\boldsymbol{1\rm D.}$
Acting on the subalgebras~$J^1$ and~$J^2$ from~\cite[Theorem~6.1]{doug2016a} by the matrix~$S$,
we obtain the subalgebras $\mathfrak f_{1.1}^1$ and $\mathfrak f_{1.1}^0$, respectively.
The matrix $\exp(\ln(6)E_2)S$, when applied to the subalgebra~$J^3$, yields the subalgebra~$\mathfrak f_{1.3}$.

The subalgebra family $J^{4,\alpha}=\langle(2-\alpha)E_2+3\alpha D\rangle$
coincides with the family $\mathfrak f_{1.2}^{\mu'}$ when $\alpha\ne2$.
Moreover, as it is stated in~\cite[Theorem~6.1]{doug2016a} $J^{4,\alpha}\simeq J^{4,\beta}$ if and only if
$\beta$ is an elements of the orbit of $\alpha$ under the action of the transformation group $G=\langle f_1,f_2\rangle$,
where $f_1\colon z\mapsto 1/z$, $f_2\colon z\mapsto 1-z$, $z\in\mathbb C$.
The canonical representatives of these orbits are precisely such $\alpha$
satisfying \[\frac{3\alpha}{(2-\alpha)}\in({\rm B}_{2}(-1)\cap\{z\in\mathbb C\mid \Re z\geqslant0\})\setminus\{{\rm i}a\mid a\in\mathbb R_{<0}\}.\]
In the case when $\alpha=2$, the subalgebra $J^{4,2}$ is equivalent under the action of $M_3(0,0)$ to the subalgebra $\mathfrak f_{1.2}^1$.
Nevertheless, this value is excluded from the parameters of the family.

\medskip\par\noindent
$\boldsymbol{2\rm D.}$
The two-dimensional subalgebras of~$\mathfrak{sl}_3(\mathbb C)$ are described in \cite[Theorems~6.2 and~6.3]{doug2016a}.
It is clear that $K_1^5=\mathfrak f_{2.6}$.
Acting on the subalgebras $K_1^1$, $K_1^2$, $K_1^3$, $K_2^1$ and $K_2^{4,\alpha}$
by the matrix~$S$, we obtain the subalgebras
$\mathfrak f_{2.2}^1$, $\mathfrak f_{2.5}$, $\mathfrak f_{2.2}^0$, $\mathfrak f_{2.8}$,
and  $\mathfrak f_{2.7}^{6\alpha+3}$, respectively.
The matrix $\exp\big(\pi/2(R_1+P_2)\big)$ maps the subalgebras $K_1^4$ and $K_2^3$ 
to the subalgebras $\mathfrak f_{2.1}$ and $\mathfrak f_{2.3}$, 
respectively.
The remaining subalgebra $K_2^2=\langle E_3,E_2-D+P_1\rangle\simeq\mathfrak f_{2.4}$
using the arguments from Section~\ref{sec:ClassCompare} case 2D.

\medskip\par\noindent
$\boldsymbol{3\rm D.}$
\cite[Theorem~6.5]{doug2016a} provides the classification of three-dimensional solvable subalgebras of~$\mathfrak{sl}_3(\mathbb C)$.
In notation of this theorem,
acting on the subalgebras $L_2^1$, \smash{$L_{3,-2/9}^1$}, $L_{3,0}^1$, $L_4^1$ and $L_5^1$ by the matrix $S$,
we obtain the subalgebras \smash{$\mathfrak f_{3.5}^{-1}$}, $\mathfrak f_{3.7}$, $\mathfrak f_{3.8}$,
\smash{$\mathfrak f_{3.5}^3$} and $\mathfrak f_{3.1}$, respectively.

The matrix~$S$, when applied to the subalgebra \smash{$L_{3,-1/4}^1$},
yields the subalgebra $\langle E_2-D-\frac13 P_1,E_1,P_2\rangle$,
which is equivalent to~$\mathfrak f_{3.6}$.

The matrix $\exp\big(-\ln(6)E_2\big)\exp\big(\pi/2(R_2-P_1)\big)$ maps the subalgebra \smash{$L_{3,-1/4}^2$}
to the subalgebra~$\mathfrak f_{3.4}$,
while $S\exp\big(\pi/2(R_1+P_2)\big)$ maps the subalgebra $L_2^2$ to the subalgebra $\mathfrak f_{3.2}$.

Acting by the matrix~$S$ on the subalgebra family \smash{$L_{3,\chi(\alpha)}^{1,\alpha}$},
where
\[
\chi(\alpha):=-\frac{(2\alpha-1)(\alpha-2)}{9(\alpha-1)^2},
\]
results in the family~$\langle(\alpha-2)E_2-3\alpha D,E_1,P_2\rangle$.
If $\alpha\ne2$, then $\langle(\alpha-2)E_2-3\alpha D,E_1,P_2\rangle=\mathfrak f_{3.5}^{-3\alpha/(\alpha-2)}$.
Moreover, it is stated in \cite[Theorem~6.5]{doug2016a} that
\[
L_{3,\chi(\alpha)}^{1,\alpha}\simeq L_{3,\chi(\beta)}^{1,\beta}
\quad\mbox{if and only if}\quad
\alpha=\beta\mbox{ or }\alpha=\beta^{-1}.
\]
These constraints agree with those applied to the parameter~$\mu$ in the subalgebra~$\mathfrak f_{3.5}^\mu$.
In the case when $\alpha=2$, the subalgebra $\langle D,E_1,P_2\rangle$
is equivalent under the action of the matrix $M_3(0,0)$ to the subalgebra~$\mathfrak f_{3.5}^1$.
Summing up, the subalgebra family $L_{3,\chi(\alpha)}^{1,\alpha}$ ($\alpha\ne\pm1$) corresponds to the family $\mathfrak f_{3.5}^\mu$
with $\mu\ne-1,3$.

Applying the matrix~$\exp\big(\pi/2(R_1+P_2)\big)$ to the subalgebra family~$L_{3,\chi(\alpha)}^{2,\alpha}$ ($\alpha\ne\pm1$)
gives us subalgebras $\langle-(\alpha+1)E_2+3(\alpha-1)D,P_1,P_2\rangle$
from the family~$\mathfrak f_{3.3}^{\rho,\varphi}$.
According to~\cite[Theorem~6.5]{doug2016a}, the subalgebras
$L_{3,\chi(\alpha)}^{2,\alpha}$ and $L_{3,\chi(\beta)}^{2,\beta}$
are equivalent if and only if $\alpha=\beta$ or $\alpha\beta=1$,
The latter is consistent with the constraints imposed on the parameters $\rho$ and $\varphi$ of $\mathfrak f_{3.3}^{\rho,\varphi}$
as these parameters are chosen to ensure the equivalence between $\langle E_2+\alpha D,P_1,P_2\rangle$
and $\langle E_2-\alpha D,P_1,P_2\rangle$.
Thus, the family~$L_{3,\chi(\alpha)}^{2,\alpha}$ ($\alpha\ne\pm1$) corresponds to the family~$\mathfrak f_{3.3}^{\rho,\varphi}$
with $\rho\ne0$.

Acting on the subalgebras~$A_1$ and~$L_4^2$ by the matrix~$\exp\big(\pi/2(R_1+P_2)\big)$,
we obtain the subalgebras~$\mathfrak f_{3.9}$ and $\mathfrak f_{3.3}^0$, respectively.

%The subalgebra $\mathfrak f_{3.10}=\langle E_1+E_3,P_1-R_2,P_2+R_1\rangle$ is equivalent to $A^2_1=\langle E_3-P_2,E_1-R_1,E_2+3D\rangle$

\medskip\par\noindent
$\boldsymbol{4\rm D.}$
In notation of \cite[Theorem~6.7]{doug2016a},
acting on the subalgebras
$M_8^1$, $M_{12}^1$, $M_{13,2}^2$ and $M_{14}^1$
by the matrix $S$ we obtain the subalgebras $\mathfrak f_{4.4}$, $\mathfrak f_{4.3}^{-3}$, $\mathfrak f_{4.3}^0$ and $\mathfrak f_{4.3}^1$,
respectively.

Applying the matrix $S$ to the subalgebra family $M_{13,\psi(\alpha)}$ ($\alpha\ne\pm1$),
\[
\psi(\alpha):=\frac{(2\alpha-1)(\alpha-2)}{(\alpha+1)^2},
\]
we obtain the family $\langle E_1,(2\alpha-1)E_2-3D,P_1,P_2\rangle$.
If $\alpha=1/2$, then this subalgebra coincides with $\mathfrak f_{4.1}$,
otherwise it coincides with the family~$\mathfrak f_{4.3}^\gamma$ with $\gamma\ne-3,0,1$.

The matrix $\exp\big(\pi/2(R_1+P_2)\big)$ maps the subalgebras $M_8^2$ and $A_1\oplus J$
to the subalgebras $\mathfrak f_{4.2}$ and $\mathfrak f_{4.5}$, respectively,
while the matrix $\exp\big(\pi/2(R_2-P_1)\big)$ maps $M_{13,0}^2$ to $\mathfrak f_{4.1}$.

It follows from these considerations that $M_{13,0}^2\simeq M_{13,\psi(1/2)}$,
thus one subalgebra among them should be excluded from the final list.

\medskip\par\noindent
$\boldsymbol{5\rm D}\textbf{ and }\boldsymbol{6\rm D.}$
Acting by the matrix $\exp\big(\pi/2(R_1+P_2)\big)$
on the subalgebras $B$, $(A_1\lsemioplus K_1)^1$, $(A_1\lsemioplus K_1)^2$, $(A_1\lsemioplus L_2)^1$ and $(A_1\lsemioplus L_2)^2$
we obtain the subalgebras $\mathfrak f_{5.1}$, $\mathfrak f_{5.2}$, $\mathfrak f_{5.3}$, $\mathfrak f_{6.1}$ and $\mathfrak f_{6.2}$,
respectively.

\begin{landscape}
\begin{table}[!ht]\footnotesize
\begin{center}
\caption{Comparison of the classifications lists: real case}
\label{tab:ListsCompare}${ }$\\[-1ex]
\renewcommand{\arraystretch}{1.6}
\begin{tabular}{|c|l|l|l|}
\hline
\hfil dim & \hfil List in Theorem~\ref{thm:SubalgebrasOfSL3}		    & \hfil List in \cite[Table~1]{wint2004a}      &\hfil Comments on \cite[Table~1]{wint2004a} \\
\hline
$1$      &$\mathfrak f_{1.1}^0=\langle E_1\rangle$              & $W_{1.4}=\langle E_3\rangle$         &             \\
         &$\mathfrak f_{1.1}^1=\langle E_1+P_1\rangle$          & $W_{1.5}=\langle E_3-P_2\rangle$     &             \\
         &$\mathfrak f_{1.2}^\kappa=\langle E_1+E_3+\gamma D\rangle$, $\kappa\geqslant0$ 
            &$ W_{1.2}^{(a)}=\langle\frac12(E_1+E_3)+a D\rangle$, $a\geqslant0$                    &   \\
         &$\mathfrak f_{1.3}^{\mu'}=\langle E_2+\mu' D\rangle$, $\mu'\in[0,1]$ 
            & $\bullet W_{1.1}^{(\alpha)}=\langle\cos(\alpha)E_2+9\sin(\alpha)D\rangle$, $\alpha\in[0,\pi)$ 
                                                                               & $\alpha\in[0,\pi)\rightsquigarrow \alpha\in[0,\arctan\frac19]$  \\
         &$\mathfrak f_{1.4}=\langle E_1+D\rangle$              & $W_{1.3}=\langle D-E_3\rangle$       &             \\
\hline

$2$      &$\mathfrak f_{2.1}  =\langle P_1,P_2\rangle$               & $W_{2.5}=\langle P_1,P_2\rangle$        &               \\
         &$\mathfrak f_{2.2}^0=\langle E_1,P_2\rangle$               & $W_{2.4}=\langle E_3,P_1\rangle$        &              \\
         &$\mathfrak f_{2.2}^1=\langle E_1+P_1,P_2\rangle$           & $W_{2.6}=\langle E_3-P_2,P_1\rangle$    &              \\
         &$\mathfrak f_{2.3}  =\langle E_2+D+P_2,P_1\rangle$         & $W_{2.9}=\langle E_2+D+P_2,P_1\rangle$ &    \\
         &$\mathfrak f_{2.4}  =\langle E_1+D,P_2\rangle$             & $W_{2.8}=\langle E_2-D+P_1,E_3\rangle$         &               \\
         &$\mathfrak f_{2.5}  =\langle E_1,D\rangle$                 & $W_{2.3}=\langle E_3,D\rangle$         &               \\
         &$\mathfrak f_{2.6}  =\langle E_2,D\rangle$                 & $W_{2.2}=\langle E_2,D\rangle$          &               \\
         &$\mathfrak f_{2.7}  =\langle E_1+E_3,D\rangle$             & $W_{2.1}=\langle E_1+E_3,D\rangle$ &              \\
         &$\mathfrak f_{2.8}^\gamma=\langle E_2+\gamma D,E_1\rangle$, $\gamma\in\mathbb R$ 
                                                                     & $W_{2.7}^{(a)}=\langle E_2+a D,E_3\rangle$, $a\in\mathbb R$ &  \\
         &$\mathfrak f_{2.9}=\langle E_2-3D, E_1+P_1\rangle$     & $W_{2.10}=\langle E_2+3D,E_3-P_2\rangle$  & \\

\hline
$3$      &$\mathfrak f_{3.1}=\langle E_1,P_1,P_2\rangle$             & $W_{3.4}=\langle E_3,P_1,P_2\rangle$              & \\
         &$\mathfrak f_{3.2}=\langle D,P_1,P_2\rangle$               & $W_{3.6}^{(\pi/2)}=\langle D,P_1,P_2\rangle$      & \\
         &$\mathfrak f_{3.3}^\kappa=\langle E_2+\kappa D,P_1,P_2\rangle$, $\kappa\in\mathbb R_{\geqslant0}$
             & $\bullet W_{3.6}^{(\alpha)}=\langle\cos(\alpha)E_2+\sin(\alpha)D,P_1,P_2\rangle$, $\alpha\in[0,\pi/2)\setminus\{\pi/4\}$, 
                                                                      &  $W_{3.6}^{(\alpha)}$ can be united with $W_{3.3}$  \\
         &  & $\bullet W_{3.3}=\langle E_2+D,P_1,P_2\rangle=W_{3.6}^{(\pi/4)}$                                            & \\        
         &$\mathfrak f_{3.4}=\langle E_1+D,P_1,P_2\rangle$         & $W_{3.10}=\langle E_3+D,P_1,P_2\rangle$        &  \\  
         &$\mathfrak f_{3.5}^\kappa=\langle E_1\!+\!E_3\!+\!\gamma D,P_1,P_2\rangle$, $\kappa\geqslant0$
                            & $W_{3.8}^{(a)}=\langle E_1+E_3+2a D,P_1,P_2\rangle$, $a\leqslant0$            
                                                       &             \\
         &$\mathfrak f_{3.6}^\kappa=\langle E_1\!+\!E_3\!+\!\gamma D,R_1,R_2\rangle$, $\kappa\geqslant0$
               & $W_{3.9}^{(a)}=\langle P_2\!+\!R_1\!+\!a D,E_3,P_1\rangle$, $a\leqslant0$ 
                                                       &               \\
\end{tabular}
\end{center}
\end{table}
\end{landscape}

\newpage

\begin{landscape}

\centerline{\small Table 1. Comparison of the classifications lists: real case (continuation)}

\vspace{-2mm}

\begin{table}[!ht]\footnotesize
\begin{center}
%\caption{Comparison of the classifications lists}
%\label{tab:ListsCompare}${ }$\\[-1ex]
\renewcommand{\arraystretch}{1.6}
\begin{tabular}{|c|l|l|l|}
\hline
\hfil dim &\hfil List in Theorem~\ref{thm:SubalgebrasOfSL3}		    &\hfil List in \cite[Table~1]{wint2004a}  &\hfil Comments on \cite[Table~1]{wint2004a}  \\
\hline
$3$          &$\mathfrak f_{3.7}^\mu=\langle E_2\!+\!\mu D,E_1,P_2\rangle$, $\mu\in[-1,3]$
             &$W_{3.5}^{(a)}=\langle E_2+6aD,E_3,P_1\rangle$, $-\frac12\leqslant a\leqslant\frac16$, $a\ne-\frac16$  
             & \\
         &            & $W_{3.2}=\langle D,P_1,E_3\rangle$  
                      & $W_{3.2}\sim \mathfrak f_{3.7}^{1}\sim W_{3.5}^{(-1/6)}$ can be united with family $W_{3.5}^{(a)}$ \\
         &            & $\bullet W_{3.11}=\langle E_2+2D+P_2,E_3,P_1\rangle$  
                      & $W_{3.11}\sim\mathfrak f_{3.7}^{-1/3}\sim W_{3.5}^{(1/18)}$  \\
         &$\mathfrak f_{3.8}=\langle E_2-D+P_1,E_1,P_2\rangle$ &  missed  &  \\      
         &$\mathfrak f_{3.9}=\langle E_2-3D, E_1+P_1,P_2\rangle$ & $W_{3.7}=\langle E_2+3D,E_3+P_2,P_1\rangle$  &  \\
         &$\mathfrak f_{3.10}=\langle E_1,E_2,D\rangle$               & $W_{3.1}=\langle E_2,E_3,D\rangle$                & \\
         &$\mathfrak f_{3.11}=\langle E_1,E_2,E_3\rangle$             & $W_{3.12}=\langle E_1,E_2,E_3\rangle$             & \\
         &$\mathfrak f_{3.12}=\langle E_1+E_3,P_1-R_2,P_2+R_1\rangle$& $W_{3.13}=\langle E_1+E_3,P_1-R_2,P_2+R_1\rangle$ & \\
         &$\mathfrak f_{3.13}=\langle E_1+E_3,P_1+R_2,P_2-R_1\rangle$& $W_{3.14}=\langle E_1+E_3,P_1+R_2,P_2-R_1\rangle$ & \\

\hline
$4$      &$\mathfrak f_{4.1}=\langle E_1,D,P_1,P_2\rangle$                   & $W_{4.6}^{(\pi/2)}=\langle D,E_3,P_1,P_2\rangle$      &      \\
         &$\mathfrak f_{4.2}=\langle E_2,D,P_1,P_2\rangle$                   & $W_{4.3}=\langle E_2,D,P_1,P_2\rangle$                &      \\
         &$\mathfrak f_{4.3}=\langle E_1+E_3,D,P_1,P_2\rangle$               & $W_{4.4}=\langle E_1+E_3,D,P_1,P_2\rangle$            &      \\
         &$\mathfrak f_{4.4}^\gamma=\langle E_2+\gamma D,E_1,P_1,P_2\rangle$, $\gamma\in\mathbb R$ 
              & $\bullet W_{4.6}^{(\alpha)}=\langle\cos(\alpha)E_2+\sin(\alpha)D,E_3,P_1,P_2\rangle$, $\alpha\in[0,\pi]\setminus\{\pi/2\}$      
              & $\alpha\in[0,\pi]\setminus\{\pi/2\}\rightsquigarrow\alpha\in[0,\pi)\setminus\{\pi/2\}$             \\
         &$\mathfrak f_{4.5}=\langle E_1,E_2,D,P_2\rangle$                   & $W_{4.2}=\langle E_2,E_3,D,P_1\rangle$                &       \\         
         &$\mathfrak f_{4.6}=\langle E_1+E_3,D,R_1,R_2\rangle$               & $W_{4.5}=\langle P_2+R_1,E_2+D,E_3,P_1\rangle$ &      \\         
         &$\mathfrak f_{4.7}=\langle E_1,E_2,E_3,D\rangle$                   & $W_{4.1}=\langle E_1,E_2,E_3,D\rangle$                &       \\ 

\hline
$5$      &$\mathfrak f_{5.1}=\langle E_1,E_2,D,P_1,P_2\rangle$               & $W_{5.3}=\langle E_2,D,E_3,P_1,P_2\rangle$            &\\
         &$\mathfrak f_{5.2}=\langle E_1,E_2,E_3,P_1,P_2\rangle$             & $W_{5.1}=\langle E_1,E_2,E_3,P_1,P_2\rangle$          &\\
         &$\mathfrak f_{5.3}=\langle E_1,E_2,E_3,R_1,R_2\rangle$             & $\bullet W_{5.2}=\langle E_2-\frac13D,E_3,P_2,R_2,P_2\rangle$   & $W_{5.2}$ has a misprint\\

\hline
$6$      &$\mathfrak f_{6.1}=\langle E_1,E_2,E_3,D,P_1,P_2\rangle$           & $W_{6.1}=\langle E_1,E_2,E_3,D,P_1,P_2\rangle$       &\\
         &$\mathfrak f_{6.2}=\langle E_1,E_2,E_3,D,R_1,R_2\rangle$           & $\bullet W_{6.2}=\langle E_2,E_3,P_2,D,R_1,P_2\rangle$  &$W_{6.2}$ has a misprint\\
\hline
\end{tabular}
\end{center}

Here the bullet symbol~$\bullet$ indicates an incorrect subalgebra or subalgebra family from~\cite[Table~1]{wint2004a},
$\rightsquigarrow$ means that the parameters on the left-hand side should be replaced by those on the right-hand side,
$\sim$ denotes ${\rm SL}_3(\mathbb R)$-conjugacy.
\end{table}
\end{landscape}

\begin{landscape}
\begin{table}[!ht]\footnotesize
\begin{center}
\caption{Comparison of the classifications lists: complex case}
\label{tab:ListsCompare2}${ }$\\[-1ex]
\renewcommand{\arraystretch}{1.6}
\begin{tabular}{|c|l|l|l|}
\hline
\hfil dim & \hfil List in Theorem~\ref{thm:SubalgebrasOfSL3C}		    & \hfil List in \cite[Table~1]{doug2016a}      &\hfil Comments on \cite[Table~1]{doug2016a} \\
\hline
$1$      &$\mathfrak f_{1.1}^0=\langle E_1\rangle$               & $J^2=\langle E_3\rangle$                                   &    \\
         &$\mathfrak f_{1.1}^1=\langle E_1+P_1\rangle$           & $J^1=\langle P_2-E_3\rangle$                               &    \\
         &$\mathfrak f_{1.2}^{\mu'}  =\langle E_2+\mu'D\rangle$, & $J^{4,\alpha}=\langle2E_2+\alpha(3D-E_2)\rangle$   &    \\
         &$\mathfrak f_{1.3}         =\langle E_1+D\rangle$,     & $J^3=\langle 6D-E_3\rangle$                                &    \\
\hline

$2$      &$\mathfrak f_{2.1}  =\langle P_1,P_2\rangle$               & $K_1^4=\langle E_3,R_1\rangle$        &               \\
         &$\mathfrak f_{2.2}^0=\langle E_1,P_2\rangle$               & $K_1^3=\langle E_3,P_1\rangle$        &              \\
         &$\mathfrak f_{2.2}^1=\langle E_1+P_1,P_2\rangle$           & $K_1^1=\langle P_2-E_3,P_1\rangle$    &              \\
         &$\mathfrak f_{2.3}  =\langle E_2+D+P_2,P_1\rangle$         & $K_2^3=\langle E_3,E_2+D+R_1\rangle$  &              \\
         &$\mathfrak f_{2.4}  =\langle E_1+D,P_2\rangle$             & $\bullet K_2^2=\langle E_3,D-E_2-P_1\rangle$  
                &  corrected misprint:~$K_2^2=\langle x_1,-\frac13 h_1+\frac13 h_2+x_3\rangle$              \\
         &$\mathfrak f_{2.5}  =\langle E_1,D\rangle$                 & $K_1^2=\langle E_3,D\rangle$          &               \\
         &$\mathfrak f_{2.6}  =\langle E_2,D\rangle$                 & $K_1^5=\langle E_2,D\rangle$                 &               \\
         &$\mathfrak f_{2.7}^\gamma=\langle E_2+\gamma D,E_1\rangle$ & $K_2^{4,\alpha}=\langle E_3,E_2-(6\alpha+3)D\rangle$     &              \\
         &$\mathfrak f_{2.8}=\langle E_2-3D, E_1+P_1\rangle$     & $K_2^1=\langle E_3-P_2,E_2+3D\rangle$  & \\

\hline
$3$      &$\mathfrak f_{3.1}=\langle E_1,P_1,P_2\rangle$                & $L_5^1=\langle E_3,P_1,P_2\rangle$              & \\
         &$\mathfrak f_{3.2}=\langle D,P_1,P_2\rangle$                  & $L_2^2=\langle E_3,R_1,E_2-D\rangle$     & \\
         &$\mathfrak f_{3.3}^{\rho,\varphi}=\langle E_2+\rho{\rm e}^{\rm i\varphi} D,P_1,P_2\rangle$
             & $L_{3,\chi(\alpha)}^{2,\alpha}=\langle E_3, R_1,(\alpha-2)E_2-3\alpha D\rangle$, $\alpha\ne\pm1$ &    \\
         &  & $L_4^2=\langle E_3,R_1,E_2+3D\rangle\simeq\mathfrak f_{3.3}^0$                                             & \\ 
         &$\mathfrak f_{3.4}=\langle E_1+D,P_1,P_2\rangle$              & $L_{3,-1/4}^2=\langle E_1,R_2,E_2+D+\frac13P_2\rangle$               &  \\
         &$\mathfrak f_{3.5}^\mu=\langle E_2+\mu D,E_1,P_2\rangle$      
                  & $L_{3,\chi(\alpha)}^{1,\alpha}=\langle E_3,P_1,(\alpha-2)E_2+3\alpha D\rangle$, $\alpha\ne\pm1$ &  \\
         &  & $L_2^1=\langle E_3,P_1,E_2+D\rangle\simeq\mathfrak f_{3.5}^{-1}$                                         & \\
         &  & $L_4^1=\langle E_3,P_1,E_2-3D\rangle\simeq\mathfrak f_{3.5}^3$                                             & \\
         &$\mathfrak f_{3.6}=\langle E_2-D+P_1,E_1,P_2\rangle$
                            & $L_{3,-1/4}^1=\langle E_3,P_1,E_2+D+\frac13P_2\rangle$               &               \\
\end{tabular}
\end{center}
\end{table}
\end{landscape}

\newpage

\begin{landscape}

\centerline{\small Table 2. Comparison of the classifications lists: complex case (continuation)}

\vspace{-2mm}

\begin{table}[!ht]\footnotesize
\begin{center}
%\caption{Comparison of the classifications lists}
%\label{tab:ListsCompare}${ }$\\[-1ex]
\renewcommand{\arraystretch}{1.6}
\begin{tabular}{|c|l|l|l|}
\hline
\hfil dim &\hfil List in Theorem~\ref{thm:SubalgebrasOfSL3C}		  &\hfil List in \cite[Table~1]{doug2016a}        &\hfil Comments on \cite[Table~1]{doug2016a}  \\
\hline
$3$      &$\mathfrak f_{3.7}=\langle E_2-3D, E_1+P_1,P_2\rangle$  &$L_{3,-2/9}^1=\langle E_3-P_2, P_1,E_2+3D\rangle$   &  \\
         &$\mathfrak f_{3.8}=\langle E_1,E_2,D\rangle$                &$L_{3,0}^1=\langle E_3,E_2,D\rangle$                &  \\
         &$\mathfrak f_{3.9}       =\langle E_1,E_2,E_3\rangle$       &$A_1=\langle P_1,R_2,E_2+3D\rangle$                 &  \\ 
         &$\mathfrak f_{3.10}=\langle E_1+E_3,P_1-R_2,P_2+R_1\rangle$ &$A^2_1=\langle E_3-P_2,E_1-R_1,E_2+3D\rangle$               &\\   
\hline
$4$      &$\mathfrak f_{4.1}=\langle E_1,D,P_1,P_2\rangle$            & $M_{13,\psi(1/2)}=\langle E_3,P_2,P_1,D\rangle$   
                     &   \\
         &                                                          & $\bullet M_{13,0}=\langle E_1,E_2+D,P_2,R_2\rangle$   
                     & $M^2_{13,0}\simeq M_{13,\psi(1/2)}$  \\
         &$\mathfrak f_{4.2}=\langle E_2,D,P_1,P_2\rangle$            & $M_8^2=\langle E_3,R_1,E_2,D\rangle$                &  \\
         &$\mathfrak f_{4.3}^\gamma=\langle E_2+\gamma D,E_1,P_1,P_2\rangle$ 
                   & $M_{13,\psi(\alpha)}=\langle E_3,P_2,P_1,(2\alpha-1)E_2+3D\rangle$, $\alpha\ne1/2$             &  \\
         &$\mathfrak f_{4.4}=\langle E_1,E_2,D,P_2\rangle$             & $M_8^1=\langle E_3,P_1,E_2,D\rangle$               &  \\
         &$\mathfrak f_{4.5}=\langle E_1,E_2,E_3,D\rangle$             & $A_1\oplus J=\langle P_1,R_2,E_2,D\rangle$         &  \\         
\hline
$5$      &$\mathfrak f_{5.1}=\langle E_1,E_2,D,P_1,P_2\rangle$               & $B=\langle E_3,P_2,P_1,E_2,D\rangle$                             & \\
         &$\mathfrak f_{5.2}=\langle E_1,E_2,E_3,P_1,P_2\rangle$             & $(A_1\lsemioplus K_1)^1=\langle P_1,R_2,E_3,R_1,E_2+3D\rangle$   & \\
         &$\mathfrak f_{5.3}=\langle E_1,E_2,E_3,R_1,R_2\rangle$             & $(A_1\lsemioplus K_1)^2=\langle P_1,R_2,P_2,E_1,E_2+3D\rangle$   & \\

\hline
$6$      &$\mathfrak f_{6.1}=\langle E_1,E_2,E_3,D,P_1,P_2\rangle$           & $(A_1\lsemioplus L_2)^1=\langle P_1,R_2,E_3,R_1,E_2,D\rangle$    &\\
         &$\mathfrak f_{6.2}=\langle E_1,E_2,E_3,D,R_1,R_2\rangle$           & $(A_1\lsemioplus L_2)^2=\langle P_1,R_2,P_2,E_1,E_2,D\rangle$    &\\
\hline
\end{tabular}

Here the bullet symbol~$\bullet$ indicates an incorrect subalgebra from~\cite[Table~1]{doug2016a},
$\sim$ denotes ${\rm SL}_3(\mathbb C)$-conjugacy.
\end{center}
\end{table}
\end{landscape}

\end{document}